\documentclass[review, 11pt]{elsarticle}

\usepackage{lineno,hyperref}
\modulolinenumbers[5]

\journal{} 

\makeatletter
\def\ps@pprintTitle{%
 \let\@oddhead\@empty
 \let\@evenhead\@empty
 \def\@oddfoot{}%
 \let\@evenfoot\@oddfoot}
\makeatother

\usepackage[mathscr]{euscript}
\usepackage{amssymb}
\newcommand*{\IP}{\mathbb{P}}
\newcommand*{\IR}{\mathbb{R}}
\newcommand*{\IQ}{\mathbb{Q}}
\newcommand*{\IE}{\mathbb{E}}
\usepackage{amsmath}
\newtheorem{theorem}{Theorem}
\newtheorem{remark}[theorem]{Remark}
\newtheorem{lemma}[theorem]{Lemma}
\newtheorem{proof}{Proof}

\expandafter\def\csname ver@subfig.sty\endcsname{}
\usepackage{subfigure}

\usepackage{tabularx}       

\usepackage[onehalfspacing]{setspace}

\usepackage[left=1.8cm, right=1.8cm, top=1.8cm, bottom=1.8cm]{geometry}

\setcitestyle{authoryear,open={(},close={)}}

\begin{document}

\begin{frontmatter}

\title{Optimal life-cycle consumption and investment decisions under age-dependent risk preferences\footnote{\linebreak \textit{\normalsize Preprint}}}

%
%
\author[addressALRZ]{Andreas Lichtenstern}
\ead{andreas.lichtenstern@tum.de}
\address[addressALRZ]{Department of Mathematics, Technical University of Munich, Munich, Germany}

\author[addressPS]{Pavel V. Shevchenko\corref{mycorrespondingauthor}}
\ead{pavel.shevchenko@mq.edu.au}
\cortext[mycorrespondingauthor]{Corresponding author}
\address[addressPS]{Department of Actuarial Studies and Business Analytics, Macquarie University, Sydney, Australia}

\author[addressALRZ]{Rudi Zagst}
\ead{zagst@tum.de}

\begin{abstract}
In this article we solve the problem of maximizing the expected utility of future consumption and terminal wealth to determine the optimal pension or life-cycle fund strategy for a cohort of pension fund investors. The setup is strongly related to a DC pension plan where additionally (individual) consumption is taken into account. The consumption rate is subject to a time-varying minimum level and terminal wealth is subject to a terminal floor. Moreover, the preference between consumption and terminal wealth as well as the intertemporal coefficient of risk aversion are time-varying and therefore depend on the age of the considered pension cohort. The optimal consumption and investment policies are calculated in the case of a Black-Scholes financial market framework and hyperbolic absolute risk aversion (HARA) utility functions. We generalize Ye (2008) (2008 American Control Conference, 356-362) by adding an age-dependent coefficient of risk aversion and extend Steffensen (2011) (Journal of Economic Dynamics and Control, 35(5), 659-667), Hentschel (2016) (Doctoral dissertation, Ulm University) and Aase (2017) (Stochastics, 89(1), 115-141) by considering consumption in combination with terminal wealth and allowing for consumption and terminal wealth floors via an application of HARA utility functions. A case study on fitting several models to realistic, time-dependent life-cycle consumption and relative investment profiles shows that only our extended model with time-varying preference parameters provides sufficient flexibility for an adequate fit. This is of particular interest to life-cycle products for (private) pension investments or pension insurance in general.
\end{abstract}


\end{frontmatter}

\setcounter{footnote}{0}

\section{Introduction}
\label{sec:Introduction}
A suitable management of pensions needs to consider earnings/contributions and investment, but should also account for the required consumption during the accumulation and/or decumulation phase. For this sake, in this paper we consider the finite horizon portfolio problem of maximizing expected utility of future consumption and terminal wealth to determine the optimal pension or life-cycle fund strategy for a cohort of pension fund investors. The setup is strongly related to a DC pension plan where additionally (individual) consumption is taken into account. Within this framework, \cite{LaknerNygren2006} describe the trade-off the investor faces as a compromise between `living well' (consumption) and `becoming rich' (terminal wealth). Classical consumption-investment problems consider constant risk aversion in the intertemporal utility functions for consumption besides a personal discount rate or impatience factor, see \cite{Merton1969} or \cite{Merton1971}. Within classical models (where constant relative risk aversion (CRRA) utilities are applied), optimal portfolio policies turn out to be constant over the life-cycle, meaning time and wealth independent. According to \cite{Aase2017} this is `against empirical evidence, and against the typical recommendations of portfolio managers'. Furthermore, \cite{Aase2017} and \cite{YangFang2014} argue that the tendency of stocks to outperform bonds over long horizons in the past is one of the reasons why people at a younger age are advised to allocate a higher proportion of wealth to equities compared to older people. Evidence for changing risk aversion over the life-cycle is reported in the literature, although there is no broad agreement on its behavior: \cite{Morin1983}, \cite{Bakshi1994}, \cite{Palsson1996}, \cite{Bellante2004}, \cite{AlAjmi2008}, \cite{Ho2009}, \cite{Yao2011} and \cite{AlbertDuffy2012} observe increasing risk aversion by age, \cite{Bellante1986} and \cite{Wang1997} find risk aversion decreasing by age and \cite{Riley1992} detect different behavior between the pre- and post-retirement phase. Age-depending risk preferences can economically be motivated by the observed behavior of people to stepwise reduce their investment risk the closer to retirement. This behavior is reflected in many life-cycle fund allocation policies, see for instance \cite{Milliman2010} or \cite{TIAACREF_LifecycleFunds2015}. An important economic reasoning behind is that the older the person, the less time to retirement entrance is left and therefore the less likely it is for her to overcome a potential market crash, strongly connected to the fear of having an insufficient wealth left for retirement. Moreover, it is reasonable that the closer to retirement time, the more satisfaction is connected with savings, i.e. with a lower consumption surplus, which yields a higher initial wealth for the decumulation phase. Based on these economic reasons, it is meaningful to consider age-varying preference parameters (dependent on the age of the pension cohort or the individual investor) in form of a coefficient of risk aversion, later called $b(t)$, and a weighting factor, later referred to as $a(t)$, that governs the relative importance of consumption at different points in time. The latter has no impact on risk aversion but can control for the varying preference between consumption and terminal wealth over time. In an analysis of the optimal controls in Section \ref{sec:NCS} we show that our proposed model can explain and describe people's observed behavior of reducing relative risky investments by time while simultaneously targeting a certain function for the consumption rate on average. In opposite, we find that the previously described existing models are not able to capture this behavior. Therefore, particulary Section \ref{sec:NCS} shows that it is economically important to have separate functions or parameters for risk aversion and preference of consumption over terminal wealth, $a(t)$ and $b(t)$.

In addition, consumption and wealth floors are introduced which have an economic meaning as minimum required levels of consumption and wealth. This motivates the development of a dynamic life-cycle model with time-varying risk preferences such as coefficient of risk aversion and consumption and wealth floors which can capture age-depending consumption and investment behavior of investors.

Related literature to this topic consider stochastic income and unemployment risks, see \cite{BodieMertonSamuelson1992}, \cite{Koo1998}, \cite{Munk2000}, \cite{Viceira2001}, \cite{Huang2008}, \cite{Jang2013}, \cite{Bensoussan2016}, \cite{Wang2016} or \cite{Chen2018}. Setups where the investor faces uncertain lifetime, mortality and optimal life insurance are considered in \cite{Yaari1965}, \cite{PliskaYe2007}, \cite{MenoncinRegis2017}, \cite{ZouCadenillas2014}, \cite{KronborgSteffensen2015}, \cite{ShenWei2016}, \cite{Duarte2011}, \cite{Huang2012}, \cite{KronborgSteffensen2015}, \cite{ShenWei2016} and \cite{Ye2008}; optimal consumption and investment under insurer default risk is studied by \cite{JangKooPark2019}. \linebreak \cite{KraftMunk2011}, \cite{KraftMunkWagner2018}, \cite{AndreassonShevchenkoNovikov2017}, \cite{CuocoLiu2000} and \cite{DamgaardFuglsbjergMunk2003} analyze optimal housing as a durable good. Constraints in the optimization problem are considered in \cite{Cuoco1997}, \cite{Elie2008} and \cite{Grandits2015}. Moreover, \cite{Akian1996}, \cite{Altarovici2017} and \cite{Dai2009} analyze the portfolio problem under transaction costs. The application of HARA utility functions in a life-cycle context can be found in \cite{Huang2008}, \cite{Ye2008}, \cite{ChangRong2014}, \cite{ChangChang2017} and \cite{Wang2017}. Moreover, \cite{BackLiuTeguia2019} study a life-cycle consumption problem for HARA utility with time-independent, increasing risk aversion and examine the relation between age and portfolio risk by using Monte Carlo analysis. \cite{TangPurcalZhang2018} study an optimal consumption-investment problem under CRRA utility function with age-independent risk aversion, but examine the impact of hyperbolic discounting, where the rate of time preference is a function of time. We generalize this approach by considering general $a(t)$ or ${e^{- \beta t} a(t)}$, respectively, and by introducing age-varying risk aversion.

In this paper we apply HARA utility functions on both the consumption and terminal wealth and consider time-varying preferences: an age-depending preference between consumption and terminal wealth and an age-depending coefficient of risk aversion in the intertemporal consumption utility. For simplicity, income is treated as a deterministic process. Furthermore, we do not model mortality and consider a fixed time horizon $T$ that corresponds to a retirement age, thus we assume the agent to survive up to the age of retirement. A positive, fixed floor in the terminal utility ensures a minimum liquid asset wealth level at the age of retirement, which is meaningful as the retiree needs wealth to live from and could possibly afford housing from this wealth. In addition, a positive, time-varying floor in the consumption utility guarantees a minimum (time-dependent) consumption rate. This is essential during the accumulation phase as for instance living expenses, rental payments when home is rented or mortgage payments and maintenance costs when home is bought and financed by debt or only maintenance costs when the agent already fully owns a house (e.g. inherited) need to be covered. Therefore, the economic demand for both a positive minimum level of consumption and terminal wealth can be motivated.

Most related to our work are \cite{Ye2008}, \cite{Steffensen2011}, \cite{Hentschel2016} and \cite{Aase2017}. The difference of our approach to these papers is as follows. \cite{Ye2008} considers income, mortality and HARA utilities for both consumption and terminal wealth under a constant coefficient of risk aversion, i.e. constant $b(t)$, but where the age-dependent preference between consumption and wealth $a(t)$ is incorporated. We generalize the results by introducing a time-dependent coefficient of risk aversion $b(t)$. \cite{Steffensen2011} provides a first insight into the optimal policy when the utility parameters of the intertemporal utility, which is of a CRRA type, are time-varying; thus $a(t)$ and $b(t)$ are captured. But the model disregards terminal wealth, consumption floor and labor income. In a similar fashion, \cite{Hentschel2016} studies the consumption problem for CRRA utility with habit formation and considers $a(t)$ and $b(t)$. Similar to \cite{Steffensen2011}, neither terminal wealth nor consumption floor nor income are included in their model. Finally, \cite{Aase2017} uses the martingale method (that allows to reformulate the optimal stochastic control problem to a simpler maximization problem with constraint) to determine optimal consumption and investment under mortality and a CRRA utility with age-depending risk aversion $b(t)$. But the model does not consider terminal wealth, consumption floor, income or time-varying preference $a(t)$.

The main contributions and innovations of this paper can be summarized as follows: we consider all the `ingredients' of the models in the above mentioned papers ($a(t)$, $b(t)$, terminal wealth, floors for consumption and terminal wealth via HARA utilities, income process) that leads to a novel, very flexible and more realistic dynamic life-cycle model framework. We extend or generalize \cite{Ye2008} by adding an age-dependent coefficient of risk aversion $b(t)$ and \cite{Steffensen2011}, \cite{Hentschel2016} and \cite{Aase2017} by considering terminal wealth and allowing for consumption and terminal wealth floors via an application of HARA utility functions. The corresponding consumption-investment problem is solved analytically and interpretation is provided. In a case study, where we fit realistic predetermined target policies for consumption and relative allocation to several models, we realize that only our proposed and most general model is sufficiently flexible to describe human preferences on consumption and investment in a suitable fashion. This implies that modeling the agent's preferences in an age-depending fashion is inevitable.

To solve the respective portfolio problem, we follow a separation approach similar to the ones developed by \cite{KaratzasShreve1998} and \cite{LaknerNygren2006}. It divides the original consumption-terminal wealth optimization problem into two sub-problems, the corresponding consumption problem and the terminal wealth problem. These separate problems are to be solved individually. Due to time-dependent preference parameters we apply the martingale method in line with \cite{Aase2017} to solve the individual problems in closed form. Afterwards, we show how the individual solutions have to be glued together in order to obtain the general solution to the original consumption-terminal wealth problem.

The remainder of this paper is organized as follows. Section \ref{sec:FinancialMarketModel} introduces the financial market and the portfolio problem of interest, Section \ref{sec:SeparationTechnique} shows the separation approach and the solution to the problem. A fit of the analytic strategy to suitable consumption and investment curves is conducted in Section \ref{sec:NCS}, followed by an investigation of the optimal controls and corresponding wealth process. Section \ref{sec:Conclusion} concludes. \ref{app:Proofs} summarizes all proofs of the claimed statements: the proofs for Section \ref{sec:ConsumptionProblem:y} on the consumption problem can be found in \ref{app:ProofsConsumptionProblem}, the proofs related to Section \ref{sec:TerminalWealthProblem} on the terminal wealth problem in \ref{app:ProofsTerminalWealthProblem}, and for the proofs associated with Section \ref{sec:OptimalMerging} on merging both individual solutions, see \ref{app:ProofsOptimalMerging}.

\section{The financial market model and consumption-investment problem}
\label{sec:FinancialMarketModel}
We consider a frictionless financial market $M$ which consists of $N+1$ continuously traded assets, one risk-free asset and $N$ risky assets. Let $[0,T]$ represent the fixed and finite investment horizon. Uncertainty in the continuous-time financial market is modeled by a complete, filtered probability space $(\Omega, \CMcal{F}, \left(\CMcal{F}_{t}\right)_{t \in [0,T]}, \IP)$, where $\Omega$ is the sample space, $\IP$ the real-world probability measure, $\CMcal{F}_{t}$ is the natural filtration generated by $W(s)$, ${0 \le s \le t}$, augmented by all the null sets, and ${W = \left(W(t)\right)_{t \in [0,T]}}$, $W(t) = (W_{1}(t), \hdots, W_{N}(t))'$, $N \in \mathbb{N}$, is a standard $N$-dimensional Brownian motion. The price of the risk-free asset at time $t$ is denoted by $P_{0}(t)$ and is subject to the equation
\begin{align}
dP_{0}(t) = r P_{0}(t)dt,\ P_{0}(0) = 1,
\end{align}

with constant risk-free interest rate $r > 0$. The remaining $N$ assets in the market are risky assets with price $P_{i}(t)$, $i = 1, \hdots, N$, at time $t$ subject to the stochastic differential equations

\begin{align}
dP_{i}(t) = {} & P_{i}(t) \left(\mu_{i} dt + \sigma_{i} dW(t)\right) = P_{i}(t)  \left(\mu_{i} dt + \sum_{j = 1}^{N} \sigma_{ij} dW_{j}(t)\right),\ P_{i}(0) = p_{i} > 0,
\end{align}
with constant drift ${\mu = \left(\mu_{1},\hdots,\mu_{N}\right)' \in \mathbb{R}_{+}^{N}}$, ${\mu - r \mathbf{1} > \mathbf{0}}$, and constant volatility vector $\sigma_{i} = (\sigma_{i1}, \hdots, \sigma_{iN}) \in \mathbb{R}_{+}^{1 \times N}$. The volatility matrix is defined by ${\sigma = \left(\sigma_{ij}\right)_{i,j = 1,\hdots,N}}$, the covariance matrix of the log-returns is ${\Sigma = \sigma \sigma'}$ which is assumed to be strongly positive definite, i.e. there exists $K > 0$ such that $\IP$-a.s. it holds ${x' \Sigma x \ge K x' x}$, ${\forall x \in \IR^{N}}$. Furthermore, let ${\gamma = \sigma^{-1} (\mu - r \mathbf{1})}$ denote the market price of risk. In this case of Black-Scholes market dynamics, according to \cite{KaratzasShreve1998}, there exists a unique risk-neutral probability measure ${\IQ \sim \IP}$ defined by ${\frac{d\IQ}{d\IP} | _{\CMcal{F}_{t}} := e^{-\frac{1}{2} \|\gamma\|^{2} t - \gamma' W(t)}}$ and the market is complete (that allows to value payment streams under the measure $\IQ$ as expected discounted values, meaning that the cost of a portfolio replicating the contract is given by its expected discounted value under $\IQ$). The corresponding pricing kernel or state price deflator, denoted by $\tilde{Z}(t)$, is defined as
\begin{align} \label{eq:PricingKernel}
\tilde{Z}(t) := e^{- \left(r + \frac{1}{2}\|\gamma\|^{2}\right) t - \gamma' W(t)}
\end{align}
and can be used for the valuation of payment streams under the real-world probability measure. Its dynamics are subject to the stochastic differential equation
\begin{align*}
d \tilde{Z}(t) = - \tilde{Z}(t) \left(r dt + \gamma' dW(t)\right),\ \tilde{Z}(0) = 1.
\end{align*}
We consider $\CMcal{F}_{t}$-progressively measurable trading strategies $\varphi = (\varphi_{0}, \hat{\varphi})'$, ${\hat{\varphi} = (\varphi_{1}, \hdots, \varphi_{N})'}$, such that $\IP$-a.s. it holds $\int_{0}^{T} |\varphi_{0}(t)| dt < \infty$ and $\int_{0}^{T} \varphi_{i}(t)^{2} dt < \infty$. $\varphi_{i}(t)$ represents the number of individual shares of asset $i$ held by the investor at time $t$. The corresponding relative portfolio process is denoted by ${\pi = (\pi_{0}, \hat{\pi}')'}$ with risky relative investment ${\hat{\pi} = (\pi_{1}, \hdots, \pi_{N})'}$ and risk-free relative investment ${\pi_{0}(t) = 1-\hat{\pi}(t)'\mathbf{1}}$, where $\pi_{i}(t)$ denotes the fraction of wealth allocated to asset $i$ at time $t$. It is to satisfy $\int_{0}^{T} \pi_{i}(t)^{2} dt < \infty$, $\IP$-a.s.. Moreover, let $(c(t))_{t \in [0,T]}$ denote a non-negative, progressively measurable, real-valued stochastic consumption rate process with $\int_{0}^{T} c(t) dt < \infty$, $\IP$-a.s., and $(y(t))_{t \in [0,T]}$ a non-negative, deterministic income-rate process with $\int_{0}^{T} y(t) dt < \infty$. Those technical conditions are assumed to get a solution for the subsequently formulated stochastic problem. The dynamics of the investor's wealth process $V = \left(V(t)\right)_{t \in [0,T]}$ under the strategy $(\pi,c)$ to initial wealth $V(0) = v_{0} > 0$, including liquid assets, consumption and income, is then given by
\begin{align} \label{eq:SDE:V:y}
dV(t) = {} & V(t)  \left[\left(r + \hat{\pi}(t)' \left(\mu - r \mathbf{1}\right)\right) dt + \hat{\pi}(t)'\sigma dW(t)\right] - c(t) dt + y(t) dt.
\end{align}
The relative investment in the risk-free asset is ${\pi_{0}(t) = 1 - \hat{\pi}(t)' \mathbf{1}}$. We consider the objective of maximizing expected utility of future terminal wealth and consumption, starting at time $0$ and ending at $T$. Hence the objective function to be maximized is
\begin{align} \label{eq:ObjectiveFunction}
J(\pi,c;v_{0}) = \IE\left[\int_{0}^{T} U_{1}(t,c(t)) dt + U_{2}(V(T))\right],
\end{align}
where $v_{0} > 0$ denotes the initial endowment of the investor. All expectations in this paper are with respect to the real-world measure $\IP$. The general portfolio optimization problem with initial wealth ${V(0) = v_{0} > 0}$ to be solved is then given by
\begin{align} \label{eq:OptimizationProblem}
\mathcal{V}(v_{0}) = \sup_{(\pi,c) \in \Lambda} J(\pi,c;v_{0})
\end{align}
subject to \eqref{eq:SDE:V:y}. $\mathcal{V}(v_{0})$ is the value function of the problem. $\Lambda$ denotes the set of admissible strategies $(\pi,c)$ such that ${V(t) + \int_{t}^{T} e^{- r (s-t)} y(s) ds \ge 0}$, $\IP$-a.s., $\forall t \in [0,T]$, and which admit a unique solution to \eqref{eq:SDE:V:y} while satisfying the integrability condition ${\IE\left[\int_{0}^{T} |U_{1}(t,c(t))| dt\right] < \infty}$. The so-called budget constraint reads
\begin{align} \label{eq:BudgetConstraint:y}
\IE\left[\int_{0}^{T} \tilde{Z}(t) c(t) dt + \tilde{Z}(T) V(T)\right] \le v_{0} + \IE\left[\int_{0}^{T} \tilde{Z}(t) y(t) dt\right] = v_{0} + \int_{0}^{T} e^{-r t} y(t) dt.
\end{align}
It describes the requirement that today's value of future consumption and terminal wealth, less income, must not exceed the initial endowment. It can be shown that for the optimal $(\hat{\pi}^{\star}, c^{\star})$ to Problem \eqref{eq:OptimizationProblem}, Equation \eqref{eq:BudgetConstraint:y} holds with equality. We consider a preference utility model given by the utility functions
\begin{align}
\begin{split} \label{eq:utilitymodel:new}
U_{1}(t,c) = {} & \left(e^{- \beta  t} a(t)\right)  \frac{1-b(t)}{b(t)}  \left(\frac{1}{1-b(t)}  \left(c - \bar{c}(t)\right)\right)^{b(t)}, \\
U_{2}(v) = {} & e^{- \beta  T} \hat{a}  \frac{1-\hat{b}}{\hat{b}}    \left(\frac{1}{1-\hat{b}}  (v-F)\right)^{\hat{b}},
\end{split}
\end{align}
for $\beta \ge 0$, $b : [0,T] \to (- \infty, 1)\backslash \{0\}$ continuous, $\hat{b} < 1$, $\hat{b} \neq 0$, $a(t) > 0$, $\hat{a} > 0$, $c(t) > \bar{c}(t)$, $\bar{c}(t) \ge 0$ deterministic, and $v > F$ with $F \ge 0$. $U_{2}$ is a continuously differentiable and strictly concave terminal utility function, $U_{1}$ denotes a continuous (intertemporal consumption) utility function which is continuously differentiable and strictly concave in the second argument. This utility model accounts for several desired aspects: minimum liquid asset wealth level $F \ge 0$ at the age of retirement $T$, minimum consumption rate $\bar{c}(t) \ge 0$ and time-varying preference of consumption over terminal wealth in terms of $a(t)$. Moreover, the coefficient of risk aversion $b(t)$ in the consumption utility is now a continuous function in time.

\begin{remark}
Notice that the associated \textit{Arrow-Pratt measure} ${\mathcal{A}(v) := - \frac{U''(v)}{U'(v)} = - \frac{\partial}{\partial v} \ln U'(v)}$ of absolute risk aversion, developed by \cite{Pratt1964} and \cite{Arrow1970}, admits the following hyperbolic representation
\begin{align*}
\mathcal{A}_{1}(t,c) = \frac{1 - b(t)}{c - \bar{c}(t)},\ \mathcal{A}_{2}(v) = \frac{1 - \hat{b}}{v - F}.
\end{align*}
For this reason, we use the notation of an increasing $b(t)$ as a synonym for a decreasing coefficient of risk aversion and vice versa. Further note that $a(t)$ does not appear in $\mathcal{A}_{1}(t,c)$. Therefore we have two input functions $a(t)$ and $b(t)$ where $a(t)$ has no influence on risk aversion, but $b(t)$ determines it; hence a very flexible model.
\end{remark}
Since we have $c(t) > \bar{c}(t)$ and $V(T) > F$ by definition of the utility functions in \eqref{eq:utilitymodel:new}, we restrict
\begin{align} \label{eq:Condition:v0:minimalrequirement}
v_{0} > \int_{0}^{T} e^{- r s} \left(\bar{c}(s) - y(s)\right) ds + e^{- r T} F =: F(0)
\end{align}
on the initial endowment in \eqref{eq:BudgetConstraint:y}. It is useful to define
\begin{align}
F(t) = {} & \IE\left[\int_{t}^{T} \frac{\tilde{Z}(s)}{\tilde{Z}(t)} \bar{c}(s) ds + \frac{\tilde{Z}(T)}{\tilde{Z}(t)} F - \int_{t}^{T} \frac{\tilde{Z}(s)}{\tilde{Z}(t)} y(s) ds \Big| \mathcal{F}_{t}\right] \nonumber \\
= {} & \int_{t}^{T} \IE\left[\frac{\tilde{Z}(s)}{\tilde{Z}(t)} \Big| \mathcal{F}_{t}\right] \left(\bar{c}(s) - y(s)\right) ds + F \IE\left[\frac{\tilde{Z}(T)}{\tilde{Z}(t)} \Big| \mathcal{F}_{t}\right] = \int_{t}^{T} e^{- r (s-t)} \left(\bar{c}(s) - y(s)\right) ds + e^{- r (T-t)} F. \label{eq:def:F(t)}
\end{align}
$F(t)$ can be interpreted as the time $t$ value of all future minimal liabilities less income. $F(t)$ equals the sum of the time $t$ wealth necessarily required to meet all the future minimum living expenses and expenditures $\bar{c}(s)$, ${s \in [t,T]}$ during the remaining time and the time $t$ value of the minimum desired terminal wealth level $F$; future salary income is subtracted as it reduces the time $t$ value of the minimum required capital.

\section{Solution: Separation technique}
\label{sec:SeparationTechnique}
In the sequel we follow the separation technique approach by \cite{KaratzasShreve1998} and \cite{LaknerNygren2006} for solving the consumption-terminal wealth problem as defined by \eqref{eq:OptimizationProblem}. We split the problem into two sub-problems: the consumption-only and terminal wealth-only problem. Both individual problems are separately solved via the martingale method, similar to the approach by \cite{Aase2017}. The individual problem solutions are optimally merged at the end. For this sake, let us consider the two individual problems first.

\subsection{The consumption problem}
\label{sec:ConsumptionProblem:y}
The consumption-only problem is
\begin{align}
\begin{split} \label{eq:OptimizationProblem:ConsumptionOnly}
J_{1}(\pi,c;v_{1}) = {} & \IE\left[\int_{0}^{T} U_{1}(t,c(t)) dt\right], \\
\mathcal{V}_{1}(v_{1}) = {} & \sup_{(\pi,c) \in \Lambda_{1}} J_{1}(\pi,c;v_{1})
\end{split}
\end{align}
subject to the budget constraint
\begin{align} \label{eq:BudgetConstraint:ConsumptionOnly:y}
\IE\left[\int_{0}^{T} \tilde{Z}(t) c(t) dt\right] \le v_{1} + \IE\left[\int_{0}^{T} \tilde{Z}(t) y(t) dt\right] = v_{1} + \int_{0}^{T} e^{-r t} y(t) dt.
\end{align}
$\Lambda_{1}$ denotes the set of admissible strategies $(\pi,c)$ such that ${V(t) + \int_{t}^{T} e^{- r (s-t)} y(s) ds \ge 0}$, $\IP$-a.s., $\forall t \in [0,T]$, and which admit a unique solution to \eqref{eq:SDE:V:y} while satisfying ${\IE\left[\int_{0}^{T} |U_{1}(t,c(t))| dt\right] < \infty}$.

\cite{Steffensen2011} provides a proof for CRRA utility functions by solving the associated Hamilton-Jacobi-Bellman (HJB) equation. We follow the approach by \cite{Aase2017}, likewise for a HARA utility function. We extend the findings of \cite{Aase2017} by introducing a time-varying, deterministic consumption floor $\bar{c}(t)$, a time-varying preference function $a(t)$ of consumption over terminal wealth and an income-rate process $y(t)$.

In order to guarantee the consumption rate floor, note $c(t) > \bar{c}(t)$, let us assume the following lower boundary for $v_{1}$ which equals the integral over the discounted consumption floor rate minus income rate over the whole horizon of interest:
\begin{align}
\begin{split} \label{eq:Condition:v1_&_eq:def:F1(t)}
v_{1} > {} & \int_{0}^{T} e^{- r s} \left(\bar{c}(s) - y(s)\right) ds =: F_{1}(0), \\
F_{1}(t) := {} & \IE\left[\int_{t}^{T} \frac{\tilde{Z}(s)}{\tilde{Z}(t)} \left(\bar{c}(s) - y(s)\right) ds \Big| \mathcal{F}_{t}\right] = \int_{t}^{T} e^{- r (s-t)} \left(\bar{c}(s) - y(s)\right) ds. 
\end{split}
\end{align}
Notice that ${v_{1} < 0}$ is possible since a sufficiently large positive income stream can be high enough to finance consumption. Using the martingale method we solve the problem as summarized by the theorem below.

\begin{theorem} \label{thm:Solution:ConsumptionOnly}
The solution to the optimal stochastic control problem \eqref{eq:OptimizationProblem:ConsumptionOnly} with intertemporal utility function $U_{1}$ in \eqref{eq:utilitymodel:new} is
\begin{align*}
\hat{\pi}_{1}(t ; v_{1}) = {} & \frac{1}{1 - b(\tilde{t}_{1})} \Sigma^{-1} (\mu - r \mathbf{1}) \frac{V_{1}(t ; v_{1}) - F_{1}(t)}{V_{1}(t ; v_{1})}, \\
c_{1}(t;v_{1}) = {} & g(t,t; v_{1}) \tilde{Z}(t)^{\frac{1}{b(t)-1}} + \bar{c}(t) = (1-b(t)) \left(\lambda_{1} \frac{e^{\beta  t}}{a(t)} \tilde{Z}(t)\right)^{\frac{1}{b(t)-1}} + \bar{c}(t), \\
V_{1}(t ; v_{1}) = {} & \int_{t}^{T} g(s,t; v_{1}) \tilde{Z}(t)^{\frac{1}{b(s)-1}} ds + F_{1}(t), \\
V_{1}(T ; v_{1}) = {} & 0,
\end{align*}
for all $t \in [0,T]$, where
\begin{align*}
g(s,t; v_{1}) = (1-b(s)) \left(\frac{e^{\beta  s - b(s) \left(r - \frac{1}{2} \frac{1}{b(s)-1} \|\gamma\|^{2}\right) (s-t)}}{a(s)}\right)^{\frac{1}{b(s)-1}} \lambda_{1}^{\frac{1}{b(s)-1}}.
\end{align*}
$\lambda_{1} = \lambda_{1}(v_{1}) > 0$ satisfies the budget constraint uniquely and is subject to the equation
\begin{align} \label{eq:ConsumptionOnly:lambda}
\int_{0}^{T} g(t,0; v_{1}) dt = v_{1} - F_{1}(0).
\end{align}
$\tilde{t}_{1} = \tilde{t}_{1}(v_{1}) \in (t,T)$ is the solution to the equation
\begin{align} \label{eq:ConsumptionOnly:tautilde}
\int_{t}^{T} \frac{1}{b(s)-1} g(s,t; v_{1}) \tilde{Z}(t)^{\frac{1}{b(s)-1}} ds = \frac{1}{b(\tilde{t}_{1})-1} \int_{t}^{T} g(s,t; v_{1}) \tilde{Z}(t)^{\frac{1}{b(s)-1}} ds.
\end{align}
For the optimal $c_{1}(t;v_{1})$, Equation \eqref{eq:BudgetConstraint:ConsumptionOnly:y} is fulfilled with equality.
\end{theorem}
We remind the reader that all proofs can be found in \ref{app:Proofs}. It is clear that ${c_{1}(t;v_{1}) > \bar{c}(t)}$, a.s.. We now aim to interpret the optimal investment strategy as proportional portfolio insurance (PPI) strategy. The first strategy family corresponds to a constant multiple, the latter one is more general and also covers proportional strategies with time-varying or even state-dependent multiples. \cite{ZielingMahayniBalder2014} evaluate the performance of such strategies. Theorem \ref{thm:Solution:ConsumptionOnly} shows that the optimal investment strategy generally is a PPI strategy with time-varying floor $F_{1}(t)$ at time $t$, equal to the time $t$ value of the accumulated outstanding future consumption floor minus income. Notice that $\tilde{t}_{1}$ can firstly be determined at time $t$, since the value depends on the stochastic $\tilde{Z}(t)$ which is not known before time $t$. Hence, $\tilde{t}_{1}$ is time- and also state-dependent and thus the optimal PPI strategy itself is time- and state-dependent through its PPI multiple. The PPI multiple in summary is time-varying, state-dependent and depends on all future coefficients of risk aversion via $b(\tilde{t}_{1})$.

Furthermore, ${V_{1}(0 ; v_{1}) > F_{1}(0)}$ holds by the assumption in \eqref{eq:Condition:v1_&_eq:def:F1(t)}. In addition, $\hat{\pi}_{1}(t ; v_{1})$ converges to $0$ when ${V_{1}(t ; v_{1})}$ approaches $F_{1}(t)$. Thus, ${V_{1}(t ; v_{1}) > F_{1}(t)}$ a.s., which additionally follows directly from the formula for $V_{1}(t ; v_{1})$ in Theorem \ref{thm:Solution:ConsumptionOnly}. This further implies that $(\hat{\pi}_{1},c_{1})$ is an admissible pair, i.e. ${(\hat{\pi}_{1},c_{1}) \in \Lambda_{1}}$. The next remark provides the solution under time-independent risk aversion.

\begin{remark}
When $b(t) \equiv b$, then
\begin{align*}
\hat{\pi}_{1}(t ; v_{1}) = \frac{1}{1 - b} \Sigma^{-1} (\mu - r \mathbf{1}) \frac{V_{1}(t ; v_{1}) - F_{1}(t)}{V_{1}(t ; v_{1})}
\end{align*}
which is a conventional CPPI strategy with constant multiple. Moreover, if $\bar{c}(t) - y(t) \equiv 0$, i.e. the minimum consumption is eating up the whole income, then
\begin{align*}
\hat{\pi}_{1}(t ; v_{1}) = \frac{1}{1 - b} \Sigma^{-1} (\mu - r \mathbf{1}),
\end{align*}
which is a constant mix strategy and represents the standard, well-known result for CRRA utility with constant risk aversion parameter.
\end{remark}
Some comments on the initial capital $v_{1}$ and the sign of the risky investments come next. As already pointed out, a start with a negative initial capital ${V_{1}(0 ; v_{1}) = v_{1} < 0}$ to Problem \eqref{eq:OptimizationProblem:ConsumptionOnly} is possible and might be reasonable in a sense that accumulated income over the life-cycle is expected to exceed total consumption. Hence, there is no need to require positive capital to this problem. For this reason, ${V_{1}(t ; v_{1}) < 0}$ can happen and might be reasonable, too.

Theorem \ref{thm:Solution:ConsumptionOnly} tells that the optimal relative investment strategy is given by
\begin{align*}
\hat{\pi}_{1}(t ; v_{1}) = {} & \frac{1}{1 - b(\tilde{t}_{1})} \Sigma^{-1} (\mu - r \mathbf{1}) \frac{V_{1}(t ; v_{1}) - F_{1}(t)}{V_{1}(t ; v_{1})},
\end{align*}
where ${V_{1}(t ; v_{1}) > F_{1}(t)}$ a.s.. Let ${\left(\Sigma^{-1} (\mu - r \mathbf{1})\right)_{i} > 0}$ for ${i \in \left\{1, \hdots, N\right\}}$, which for instance is the case when there is only one risky asset ($N = 1$) because then ${\Sigma^{-1} \left(\mu - r\right) = \frac{\mu - r}{\sigma^{2}} > 0}$ since $\mu - r \mathbf{1} > 0$ was assumed. Then
\begin{align*}
& \left(\hat{\pi}_{1}(t ; v_{1})\right)_{i} > 0\ \Leftrightarrow\ V_{1}(t ; v_{1}) > 0, \\
& \left(\hat{\pi}_{1}(t ; v_{1})\right)_{i} < 0\ \Leftrightarrow\ V_{1}(t ; v_{1}) < 0.
\end{align*}
Even if the first part of the remark argues that ${V_{1}(t ; v_{1}) < 0}$ is a meaningful case, the conclusion ${\left(\hat{\pi}_{1}(t ; v_{1})\right)_{i} < 0}$ under ${\left(\Sigma^{-1} (\mu - r \mathbf{1})\right)_{i} > 0}$ sounds odd at a first glance. But when looking at the optimal exposure to risky asset $i$, one finds that
\begin{align*}
\left(\hat{\pi}_{1}(t ; v_{1}) V_{1}(t ; v_{1})\right)_{i} = {} & \frac{1}{1 - b(\tilde{t}_{1})} \left(\Sigma^{-1} (\mu - r \mathbf{1})\right)_{i} \left(V_{1}(t ; v_{1}) - F_{1}(t)\right),
\end{align*}
which, under ${\left(\Sigma^{-1} (\mu - r \mathbf{1})\right)_{i} > 0}$, is positive no matter if ${V_{1}(t ; v_{1}) < 0}$ or ${V_{1}(t ; v_{1}) > 0}$. Therefore, the amount of money invested in the risky assets is always positive. The opposite inequalities and conclusions for $\left(\hat{\pi}_{1}(t ; v_{1})\right)_{i}$ and $\left(\hat{\pi}_{1}(t ; v_{1}) V_{1}(t ; v_{1})\right)_{i}$ apply if ${\left(\Sigma^{-1} (\mu - r \mathbf{1})\right)_{i} < 0}$. In summary, the sign of the optimal exposure to the single risky assets is determined by
\begin{align*}
\left(\hat{\pi}_{1}(t ; v_{1}) V_{1}(t ; v_{1})\right)_{i} > 0\ \Leftrightarrow\ \left(\Sigma^{-1} (\mu - r \mathbf{1})\right)_{i} > 0.
\end{align*}
Thus, ${\left(\hat{\pi}_{1}(t ; v_{1}) V_{1}(t ; v_{1})\right)_{i} > 0}$ is possible although it might be ${\left(\hat{\pi}_{1}(t ; v_{1})\right)_{i} < 0}$.

Finally, let ${\left(\Sigma^{-1} (\mu - r \mathbf{1})\right)_{i} > 0}$ for all ${i \in \left\{1, \hdots, N\right\}}$. When ${V_{1}(t ; v_{1}) < 0}$, the optimal exposure to the risk-free asset is negative because
\begin{align*}
\underbrace{V_{1}(t ; v_{1})}_{< 0} \left(1 - \underbrace{\hat{\pi}_{1}(t ; v_{1})' \mathbf{1}}_{< 0}\right) < V_{1}(t ; v_{1}) < 0.
\end{align*}
This in turn implies that in case of ${V_{1}(t ; v_{1}) < 0}$, the investor takes leverage by borrowing from the risk-free account to achieve her investment goals. Leverage at this point can make sense as future income provides some security; note that ${V_{1}(t ; v_{1}) < 0}$ immediately implies that the time $t$ value of accumulated future income exceeds the expected value of consumption.

Some more properties of $\hat{\pi}_{1}(t ; v_{1})$ can be found analytically as follows. The first and second derivative of $\left(\hat{\pi}_{1}(t ; v_{1})\right)_{i}$, $i = 1, \hdots, N$, with respect to wealth $V_{1}(t ; v_{1})$ are
\begin{align*}
\frac{\partial}{\partial V_{1}(t ; v_{1})} \left(\hat{\pi}_{1}(t ; v_{1})\right)_{i} = {} & \frac{1}{1 - b(\tilde{t}_{1})} \left(\Sigma^{-1} (\mu - r \mathbf{1})\right)_{i} \frac{F_{1}(t)}{V_{1}(t ; v_{1})^{2}}, \\
\frac{\partial^{2}}{\partial V_{1}(t ; v_{1})^{2}} \left(\hat{\pi}_{1}(t ; v_{1})\right)_{i} = {} & - 2 \frac{1}{1 - b(\tilde{t}_{1})} \left(\Sigma^{-1} (\mu - r \mathbf{1})\right)_{i} \frac{F_{1}(t)}{V_{1}(t ; v_{1})^{3}}.
\end{align*}
Let ${\left(\Sigma^{-1} (\mu - r \mathbf{1})\right)_{i} > 0}$ for ${i \in \left\{1, \hdots, N\right\}}$, then
\begin{enumerate}
\item ${\frac{\partial}{\partial V_{1}(t ; v_{1})} \left(\hat{\pi}_{1}(t ; v_{1})\right)_{i} \stackrel{(>)}{\ge} 0\ \Leftrightarrow\ F_{1}(t) \stackrel{(>)}{\ge} 0}$.
\item ${\frac{\partial^{2}}{\partial V_{1}(t ; v_{1})^{2}} \left(\hat{\pi}_{1}(t ; v_{1})\right)_{i} \stackrel{(<)}{\le} 0\ \Leftrightarrow}$ either ${F_{1}(t) \stackrel{(>)}{\ge} 0}$ and ${V_{1}(t ; v_{1}) > 0}$ or ${F_{1}(t) \stackrel{(<)}{\le} 0}$ and ${V_{1}(t ; v_{1}) < 0}$.
\end{enumerate}
This implies that at time $t$:
\begin{enumerate}
\item $\left(\hat{\pi}_{1}(t ; v_{1})\right)_{i}$ is increasing in $V_{1}(t ; v_{1})$ if and only if ${F_{1}(t) \ge 0}$, and decreasing in $V_{1}(t ; v_{1})$ otherwise.
\item $\left(\hat{\pi}_{1}(t ; v_{1})\right)_{i}$ is concave in $V_{1}(t ; v_{1})$ if and only if
\begin{enumerate}
\item either ${F_{1}(t) \ge 0}$ and ${V_{1}(t ; v_{1}) > 0}$
\item or ${F_{1}(t) \le 0}$ and ${V_{1}(t ; v_{1}) < 0}$,
\end{enumerate}
and convex in $V_{1}(t ; v_{1})$ otherwise.
\end{enumerate}
The opposite inequalities and conclusions for $\left(\hat{\pi}_{1}(t ; v_{1})\right)_{i}$ and its derivatives apply if ${\left(\Sigma^{-1} (\mu - r \mathbf{1})\right)_{i} < 0}$.

The optimal controls in Theorem \ref{thm:Solution:ConsumptionOnly} determine the value function and the value for $\lambda_{1}$ as follows.

\begin{theorem} \label{thm:Solution:ConsumptionOnly:ValueFunction:lambda}
The optimal value function $\mathcal{V}_{1}(v_{1})$ to Problem \eqref{eq:OptimizationProblem:ConsumptionOnly} is strictly increasing and concave in $v_{1}$. Its value and first and second derivative with respect to the initial budget $v_{1}$ are given by
\begin{align*}
\mathcal{V}_{1}(v_{1}) = {} & \int_{0}^{T} \frac{1-b(t)}{b(t)} \left(\frac{e^{\left[\beta - b(t) \left(r - \frac{1}{2} \frac{1}{b(t)-1} \|\gamma\|^{2}\right)\right] t}}{a(t)}\right)^{\frac{1}{b(t)-1}} \lambda_{1}^{\frac{b(t)}{b(t)-1}} dt, \\
\mathcal{V}_{1}^{\prime}(v_{1}) = {} & \lambda_{1} > 0, \\
\mathcal{V}_{1}^{\prime\prime}(v_{1}) = {} & \lambda_{1}^{\prime} = - \left(\int_{0}^{T} \left(\frac{e^{\left[\beta - b(t) \left(r - \frac{1}{2} \frac{1}{b(t)-1} \|\gamma\|^{2}\right)\right] t}}{a(t)}\right)^{\frac{1}{b(t)-1}} \lambda_{1}^{- \frac{b(t)-2}{b(t)-1}}  dt\right)^{-1} < 0.
\end{align*}
\end{theorem}

\subsection{The terminal wealth problem}
\label{sec:TerminalWealthProblem}
The terminal wealth-only problem is
\begin{align}
\begin{split} \label{eq:OptimizationProblem:TerminalWealthOnly}
J_{2}(\pi,c;v_{2}) = {} & \IE\left[U_{2}(V(T))\right], \\
\mathcal{V}_{2}(v_{2}) = {} & \sup_{(\pi,c) \in \Lambda_{2}} J_{2}(\pi,c;v_{2})
\end{split}
\end{align}
subject to the budget constraint
\begin{align} \label{eq:BudgetConstraint:TerminalWealthOnly}
\IE\left[\tilde{Z}(T) V(T)\right] \le v_{2},\ v_{2} \ge 0.
\end{align}
$\Lambda_{2}$ denotes the set of admissible strategies $(\pi,c)$ such that ${V(t) \ge 0}$, $\IP$-a.s., $\forall t \in [0,T]$, and which admit a unique solution to \eqref{eq:SDE:V:y} for ${y(t) \equiv 0}$.

In order to guarantee the terminal wealth floor, note $V(T) > F$, let us assume the following lower bound for $v_{2}$ which equals the discounted terminal floor:
\begin{align}
v_{2} > {} & e^{- r T} F =: F_{2}(0),\ F_{2}(t) := \IE\left[\frac{\tilde{Z}(T)}{\tilde{Z}(t)} F \Big| \mathcal{F}_{t}\right] = e^{- r (T-t)} F \ge 0. \label{eq:Condition:v2:WithoutProbabilityConstraint_&_eq:def:F2(t)}
\end{align}
Applying the martingale approach leads to the solution to the terminal wealth problem according to the upcoming theorem.

\begin{theorem} \label{thm:Solution:TerminalWealthOnly:WithoutProbabilityConstraint}
The solution to Problem \eqref{eq:OptimizationProblem:TerminalWealthOnly} with terminal utility function $U_{2}$ in \eqref{eq:utilitymodel:new} is
\begin{align*}
\hat{\pi}_{2}(t ; v_{2}) = {} & \frac{1}{1-\hat{b}} \Sigma^{-1}  (\mu - r \mathbf{1}) \frac{V_{2}(t ; v_{2}) - F_{2}(t)}{V_{2}(t ; v_{2})}, \\
c_{2}(t;v_{2}) = {} & 0, \\
V_{2}(t ; v_{2}) = {} & \left(v_{2} - e^{- r T} F\right) e^{\frac{\hat{b}}{\hat{b}-1} \left(r - \frac{1}{2} \frac{1}{\hat{b}-1} \|\gamma\|^{2}\right) t} \tilde{Z}(t)^{\frac{1}{\hat{b}-1}} + F_{2}(t), \\
V_{2}(T ; v_{2}) = {} & \left(v_{2} - e^{- r T} F\right) e^{\frac{\hat{b}}{\hat{b}-1} \left(r - \frac{1}{2} \frac{1}{\hat{b}-1} \|\gamma\|^{2}\right) T} \tilde{Z}(T)^{\frac{1}{\hat{b}-1}} + F,
\end{align*}
for all $t \in [0,T]$. For the optimal $\hat{\pi}_{2}(t ; v_{2})$, Equation \eqref{eq:BudgetConstraint:TerminalWealthOnly} is fulfilled with equality.
\end{theorem}
Theorem \ref{thm:Solution:TerminalWealthOnly:WithoutProbabilityConstraint} shows that the optimal fraction of wealth allocated to the risky assets follows a CPPI strategy with floor $F_{2}(t) \ge 0$ at time $t$, with constant multiple. Moreover, ${V_{2}(0 ; v_{2}) > F_{2}(0) = e^{- r T} F}$ by the assumption in \eqref{eq:Condition:v2:WithoutProbabilityConstraint_&_eq:def:F2(t)}. In addition, $\hat{\pi}_{2}(t ; v_{2})$ converges to $0$ when ${V_{2}(t ; v_{1})}$ approaches $F_{2}(t)$. Thus, it follows ${V_{2}(t ; v_{2}) > F_{2}(t)}$ a.s., which additionally yields that $(\hat{\pi}_{2},0)$ is an admissible pair, i.e. ${(\hat{\pi}_{2},0) \in \Lambda_{2}}$. The characteristics ${V_{2}(t ; v_{2}) > F_{2}(t)}$ a.s. also directly follows from the formula for ${V_{2}(t ; v_{2})}$ in Theorem \ref{thm:Solution:TerminalWealthOnly:WithoutProbabilityConstraint}. The next remark shows that the optimal proportion allocated to the risky assets is constant over time if one disregards the floor $F$.

\begin{remark}
When $F = 0$, then
\begin{align*}
\hat{\pi}_{2}(t ; v_{2}) = \frac{1}{1 - \hat{b}} \Sigma^{-1} (\mu - r \mathbf{1})
\end{align*}
which is a constant mix strategy and equals the standard result for CRRA utility with constant risk aversion parameter, where the optimal fraction of wealth allocated to the single risky assets does not depend on time or wealth.
\end{remark}
In what follows we analyze some characteristics of the optimal strategy $\hat{\pi}_{2}(t ; v_{2})$. The first and second derivative of $\left(\hat{\pi}_{2}(t ; v_{2})\right)_{i}$, $i = 1, \hdots, N$, with respect to wealth $V_{2}(t ; v_{2})$ are
\begin{align*}
\frac{\partial}{\partial V_{2}(t ; v_{2})} \left(\hat{\pi}_{2}(t ; v_{2})\right)_{i} = {} & \frac{1}{1 - \hat{b}} \left(\Sigma^{-1} (\mu - r \mathbf{1})\right)_{i} \frac{F_{2}(t)}{V_{2}(t ; v_{2})^{2}}, \\
\frac{\partial^{2}}{\partial V_{2}(t ; v_{2})^{2}} \left(\hat{\pi}_{2}(t ; v_{2})\right)_{i} = {} & - 2 \frac{1}{1 - \hat{b}} \left(\Sigma^{-1} (\mu - r \mathbf{1})\right)_{i} \frac{F_{2}(t)}{V_{2}(t ; v_{2})^{3}}.
\end{align*}
Let ${\left(\Sigma^{-1} (\mu - r \mathbf{1})\right)_{i} > 0}$ for ${i \in \left\{1, \hdots, N\right\}}$. Then ${\frac{\partial}{\partial V_{2}(t ; v_{2})} \left(\hat{\pi}_{2}(t ; v_{2})\right)_{i} \ge 0}$ and ${\frac{\partial^{2}}{\partial V_{2}(t ; v_{2})^{2}} \left(\hat{\pi}_{2}(t ; v_{2})\right)_{i} \le 0}$, where the inequalities hold strictly when $F > 0$. Hence, $\left(\hat{\pi}_{2}(t ; v_{2})\right)_{i}$ increases and is concave in the wealth $V_{2}(t ; v_{2})$. Otherwise, if ${\left(\Sigma^{-1} (\mu - r \mathbf{1})\right)_{i} < 0}$ for ${i \in \left\{1, \hdots, N\right\}}$, then $\left(\hat{\pi}_{2}(t ; v_{2})\right)_{i}$ decreases and is convex in the wealth $V_{2}(t ; v_{2})$. For the optimal exposure to the risky assets it therefore holds
\begin{align*}
\left(\hat{\pi}_{2}(t ; v_{2}) V_{2}(t ; v_{2})\right)_{i} > 0\ \Leftrightarrow\ \left(\Sigma^{-1} (\mu - r \mathbf{1})\right)_{i} > 0.
\end{align*}
Thus, either it is ${\left(\hat{\pi}_{2}(t ; v_{2})\right)_{i} > 0}$ and ${\left(\hat{\pi}_{2}(t ; v_{2}) V_{2}(t ; v_{2})\right)_{i} > 0}$ or ${\left(\hat{\pi}_{2}(t ; v_{2})\right)_{i} < 0}$ and ${\left(\hat{\pi}_{2}(t ; v_{2}) V_{2}(t ; v_{2})\right)_{i} < 0}$.

The optimal controls in Theorem \ref{thm:Solution:TerminalWealthOnly:WithoutProbabilityConstraint} determine the value function and the value for $\lambda_{2}$.

\begin{theorem} \label{thm:Solution:TerminalWealthOnly:ValueFunction:lambda}
The optimal value function $\mathcal{V}_{2}(v_{2})$ to Problem \eqref{eq:OptimizationProblem:TerminalWealthOnly} is strictly increasing and concave in $v_{2}$. Its value and first and second derivative with respect to the initial budget $v_{2}$ are given by
\begin{align*}
\mathcal{V}_{2}(v_{2}) = {} & e^{\left[- \beta + \hat{b} \left(r - \frac{1}{2} \frac{1}{\hat{b}-1} \|\gamma\|^{2}\right)\right] T}  \frac{\left(1-\hat{b}\right)^{1-\hat{b}}}{\hat{b}}  \hat{a} \left(v_{2} - F_{2}(0)\right)^{\hat{b}}, \\
\mathcal{V}_{2}^{\prime}(v_{2}) = {} & \lambda_{2} > 0, \\
\mathcal{V}_{2}^{\prime\prime}(v_{2}) = {} & \lambda_{2}^{\prime} = - e^{\left[- \beta + \hat{b} \left(r - \frac{1}{2} \frac{1}{\hat{b}-1} \|\gamma\|^{2}\right)\right] T}  \left(1-\hat{b}\right)^{2-\hat{b}}  \hat{a} \left(v_{2} - F_{2}(0)\right)^{\hat{b}-2} < 0.
\end{align*}
The Lagrange multiplier is given by \eqref{eq:Lagrange:TerminalWealthOnly} as
\begin{align*}
\lambda_{2} = e^{-\left[\beta - \hat{b} \left(r - \frac{1}{2} \frac{1}{\hat{b}-1} \|\gamma\|^{2}\right)\right] T}  \left(1-\hat{b}\right)^{1-\hat{b}}  \hat{a} \left(v_{2} - F_{2}(0)\right)^{\hat{b}-1} > 0.
\end{align*}
\end{theorem}

\subsection{Optimal merging of the individual solutions}
\label{sec:OptimalMerging}
Let $(\pi_{1}(t;v_{1}), c_{1}(t;v_{1}))$ denote the optimal controls to Problem \eqref{eq:OptimizationProblem:ConsumptionOnly} with optimal wealth process $V_{1}(t;v_{1})$ to the initial wealth ${v_{1} \ge \int_{0}^{T} e^{- r t} \left(\bar{c}(t) - y(t)\right) dt = F_{1}(0)}$ and $(\pi_{2}(t;v_{2}), c_{2}(t;v_{2}))$ the optimal controls to Problem \eqref{eq:OptimizationProblem:TerminalWealthOnly} with optimal wealth process $V_{2}(t;v_{2})$ to the initial wealth ${v_{2} \ge e^{- r T} F  = F_{2}(0)}$. Then merging the two solutions to solve Problem \eqref{eq:OptimizationProblem} is based on the following theorem.

\begin{theorem} \label{thm:ConnectionValueFunctions}
The connection between the value functions is
\begin{align*}
\mathcal{V}(v_{0}) = \sup_{v_{1} \ge F_{1}(0),\ v_{2} \ge F_{2}(0),\ v_{1} + v_{2} = v_{0}} \left\{\mathcal{V}_{1}(v_{1}) + \mathcal{V}_{2}(v_{2})\right\}.
\end{align*}
\end{theorem}
Notice that ${F(t) = F_{1}(t) + F_{2}(t)}$, hence \eqref{eq:Condition:v0:minimalrequirement} ensures that ${v_{0} = v_{1} + v_{2} > F_{1}(0) + F_{2}(0)}$ is claimed. When discounted future income exceeds consumption over the considered period, i.e. when the initial budget to the consumption problem is negative (${v_{1} < 0}$), then ${v_{2} > v_{0}}$ and a higher amount of money $v_{2}$ is invested according to the terminal wealth problem at initial time as the initial endowment $v_{0}$ of the investor.

Theorem \ref{thm:ConnectionValueFunctions} shows that an optimal allocation to consumption and terminal wealth at $t = 0$ together with the solution to the two separate problems equals the solution to the original optimization problem. The optimal initial budgets are denoted by $v_{1}^{\star}$ and $v_{2}^{\star}$. The next lemma provides a condition for $v_{1}^{\star}$ and $v_{2}^{\star}$.

\begin{lemma} \label{lemma:EqualityConditionValueFunctions}
The optimal $v_{1}^{\star}$ solves
\begin{align} \label{eq:OptimalMerging:ValueFunction}
\mathcal{V}_{1}^{\prime}(v_{1}) - \mathcal{V}_{2}^{\prime}(v_{0} - v_{1}) = 0
\end{align}
and is subject to ${F_{1}(0) \le v_{1}^{\star} \le v_{0} - F_{2}(0)}$. The optimal $v_{2}^{\star}$ is then given by ${v_{2}^{\star} = v_{0} - v_{1}^{\star}}$.
\end{lemma}
Within our specified setup, we can address the condition in Lemma \ref{lemma:EqualityConditionValueFunctions} in more detail, the result is provided next.

\begin{lemma} \label{lemma:EqualityConditionValueFunctions:Solution}
The optimal $v_{1}^{\star}$ to \eqref{eq:OptimalMerging:ValueFunction} exists uniquely and satisfies the boundary condition $F_{1}(0) \le v_{1}^{\star} \le v_{0} - F_{2}(0)$. $v_{1}^{\star}$ is the solution to the equation
\begin{align} \label{eq:Separation:Optimalv1}
v_{1} - \int_{0}^{T} \chi(t) \left(v_{0} - v_{1} - F_{2}(0)\right)^{\frac{\hat{b}-1}{b(t)-1}} dt = F_{1}(0)
\end{align}
with
\begin{align} \label{eq:definition:chi}
\chi(t) = (1-b(t)) \left(1-\hat{b}\right)^{\frac{1-\hat{b}}{b(t)-1}} \left(\frac{\hat{a}}{a(t)}\right)^{\frac{1}{b(t)-1}} \left(\frac{e^{\left[\beta - b(t) \left(r - \frac{1}{2} \frac{1}{b(t)-1} \|\gamma\|^{2}\right)\right] t}}{e^{\left[\beta - \hat{b} \left(r - \frac{1}{2} \frac{1}{\hat{b}-1} \|\gamma\|^{2}\right)\right] T}}\right)^{\frac{1}{b(t)-1}} > 0.
\end{align}
The optimal $v_{2}^{\star}$ is given by ${v_{2}^{\star} = v_{0} - v_{1}^{\star}}$.

Moreover, the optimal Lagrange multiplier $\lambda_{1}^{\star} = \lambda_{1}(v_{1}^{\star})$ is given by
\begin{align*}
\lambda_{1}^{\star} = \left(1-\hat{b}\right)^{1-\hat{b}} \hat{a} e^{-\left[\beta - \hat{b} \left(r - \frac{1}{2} \frac{1}{\hat{b}-1} \|\gamma\|^{2}\right)\right] T} \left(v_{0} - v_{1}^{\star} - F_{2}(0)\right)^{\hat{b}-1}.
\end{align*}
\end{lemma}
For general $a(t)$ and $b(t)$, $v_{1}^{*}$ as the unique solution to Equation \eqref{eq:Separation:Optimalv1} can for instance be determined numerically. Denote by $v_{1}^{\star} \ge F_{1}(0),\ v_{2}^{\star} \ge F_{2}(0)$ with $v_{1}^{\star} + v_{2}^{\star} = v_{0}$ the optimal allocation of the initial wealth according to Lemma \ref{lemma:EqualityConditionValueFunctions:Solution} in what follows and denote $\lambda_{1}^{\star} = \lambda_{1}(v_{1}^{\star})$ and $\tilde{t}_{1}^{\star} = \tilde{t}_{1}(v_{1}^{\star})$. We use the individual solutions to the two separate Problems \eqref{eq:OptimizationProblem:ConsumptionOnly} and \eqref{eq:OptimizationProblem:TerminalWealthOnly} and merge both solutions optimally according to Lemma \ref{lemma:EqualityConditionValueFunctions:Solution} to obtain the solution to the original Problem \eqref{eq:OptimizationProblem}.

\begin{theorem} \label{thm:Solution:Merging:OriginalProblem}
The optimal wealth process is given by ${V^{\star}(t; v_{0}) = V_{1}(t;v_{1}^{\star}) + V_{2}(t;v_{2}^{\star})}$. The optimal controls to Problem \eqref{eq:OptimizationProblem} are 
\begin{align*}
c^{\star}(t; v_{0}) = c_{1}(t;v_{1}^{\star}),\ \hat{\pi}^{\star}(t; v_{0}) = \frac{\hat{\pi}_{1}(t;v_{1}^{\star}) V_{1}(t;v_{1}^{\star}) + \hat{\pi}_{2}(t;v_{2}^{\star}) V_{2}(t;v_{2}^{\star})}{V_{1}(t;v_{1}^{\star}) + V_{2}(t;v_{2}^{\star})}.
\end{align*}
The optimal controls and the optimal wealth process to Problem \eqref{eq:OptimizationProblem} under the utility function setup \eqref{eq:utilitymodel:new} are given by
\begin{align*}
\hat{\pi}^{\star}(t ; v_{0}) = {} & \Sigma^{-1} (\mu - r \mathbf{1}) \frac{\frac{1}{1 - b(\tilde{t}_{1}^{\star})} \left(V_{1}(t ; v_{1}^{\star}) - F_{1}(t)\right) + \frac{1}{1-\hat{b}} \left(V_{2}(t ; v_{2}^{\star}) - F_{2}(t)\right)}{V^{\star}(t ; v_{0})}, \\
c^{\star}(t;v_{0}) = {} & g(t,t; v_{1}^{\star}) \tilde{Z}(t)^{\frac{1}{b(t)-1}} + \bar{c}(t) = (1-b(t)) \left(\lambda_{1}^{\star} \frac{e^{\beta  t}}{a(t)} \tilde{Z}(t)\right)^{\frac{1}{b(t)-1}} + \bar{c}(t), \\
V^{\star}(t ; v_{0}) = {} & \int_{t}^{T} g(s,t; v_{1}^{\star}) \tilde{Z}(t)^{\frac{1}{b(s)-1}} ds + \left(v_{2}^{\star} - F_{2}(0)\right) e^{\frac{\hat{b}}{\hat{b}-1} \left(r - \frac{1}{2} \frac{1}{\hat{b}-1} \|\gamma\|^{2}\right) t} \tilde{Z}(t)^{\frac{1}{\hat{b}-1}} + F(t), \\
V^{\star}(T ; v_{0}) = {} & \left(v_{2}^{\star} - F_{2}(0)\right) e^{\frac{\hat{b}}{\hat{b}-1} \left(r - \frac{1}{2} \frac{1}{\hat{b}-1} \|\gamma\|^{2}\right) T} \tilde{Z}(T)^{\frac{1}{\hat{b}-1}} + F, \\
V_{1}(t ; v_{1}^{\star}) = {} & \int_{t}^{T} g(s,t; v_{1}^{\star}) \tilde{Z}(t)^{\frac{1}{b(s)-1}} ds + F_{1}(t), \\
V_{2}(t ; v_{2}^{\star}) = {} & \left(v_{2}^{\star} - F_{2}(0)\right) e^{\frac{\hat{b}}{\hat{b}-1} \left(r - \frac{1}{2} \frac{1}{\hat{b}-1} \|\gamma\|^{2}\right) t} \tilde{Z}(t)^{\frac{1}{\hat{b}-1}} + F_{2}(t) \text{,\ $\forall t \in [0,T]$, with}
\end{align*}
${g(s,t; v_{1}^{\star}) = \chi(s) e^{\frac{b(s)}{b(s)-1} \left(r - \frac{1}{2} \frac{1}{b(s)-1} \|\gamma\|^{2}\right) t} \left(v_{0} - v_{1}^{\star} - F_{2}(0)\right)^{\frac{\hat{b}-1}{b(s)-1}}}$, and ${\tilde{t}_{1}^{\star} = \tilde{t}_{1}(v_{1}^{\star}) \in (t,T)}$ solves
\begin{align*}
b(\tilde{t}_{1}^{\star}) = {} & 1 + \frac{\int_{t}^{T} g(s,t; v_{1}^{\star}) \tilde{Z}(t)^{\frac{1}{b(s)-1}} ds}{\int_{t}^{T} \frac{1}{b(s)-1} g(s,t; v_{1}^{\star}) \tilde{Z}(t)^{\frac{1}{b(s)-1}} ds}.
\end{align*}
For the optimal $(\hat{\pi}^{\star}(t ; v_{0}), c^{\star}(t;v_{0}))$, Equation \eqref{eq:BudgetConstraint:y} holds with equality.
\end{theorem}
It follows immediately that ${c_{1}(t;v_{1}) > \bar{c}(t)}$, a.s.. Theorem \ref{thm:Solution:Merging:OriginalProblem} furthermore proves that the general optimal relative investment strategy can be written as a mixture of a PPI and a CPPI strategy, but is not necessarily of a PPI or even CPPI type itself. The PPI comes from the consumption-only problem, see Theorem \ref{thm:Solution:ConsumptionOnly}, the CPPI arises as the solution to the terminal wealth-only problem, see Theorem \ref{thm:Solution:TerminalWealthOnly:WithoutProbabilityConstraint}. The way which of the two strategies dominates the overall optimal investment policy is initially determined by the wealth distribution through $v_{1}^{\star}$ and $v_{2}^{\star}$ and later through $V_{1}(t ; v_{1}^{\star})$ and $V_{2}(t ; v_{2}^{\star})$. The special case where the coefficient of risk aversion $b(t)$ from consumption equals the one from terminal wealth $\hat{b}$ at any time is covered by the next remark.

\begin{remark} \label{remark:Ye:case}
Assume $b(t) \equiv \hat{b}$ constant. Then the optimal controls turn into
\begin{align*}
\hat{\pi}^{\star}_{(b(t) \equiv \hat{b})}(t ; v_{0}) = {} & \frac{1}{1-\hat{b}} \Sigma^{-1} (\mu - r \mathbf{1}) \frac{V^{\star}_{(b(t) \equiv \hat{b})}(t ; v_{0}) - F(t)}{V^{\star}_{(b(t) \equiv \hat{b})}(t ; v_{0})}, \\
c^{\star}_{(b(t) \equiv \hat{b})}(t;v_{0}) = {} & \zeta(t) \left(V^{\star}_{(b(t) \equiv \hat{b})}(t ; v_{0}) -F(t)\right) + \bar{c}(t),
\end{align*}
with
\begin{align*}
\zeta(t) = \frac{\chi(t)}{\int_{t}^{T} \chi(s) ds + 1} > 0,
\end{align*}
where
\begin{align*}
\chi(t) = \left(\frac{\hat{a}}{a(t)}\right)^{\frac{1}{\hat{b}-1}} e^{- \frac{1}{\hat{b}-1} \left[\beta - \hat{b} \left(r - \frac{1}{2} \frac{1}{\hat{b}-1} \|\gamma\|^{2}\right)\right] (T-t)} > 0.
\end{align*}
The optimal investment strategy $\hat{\pi}^{\star}_{(b(t) \equiv \hat{b})}(t ; v_{0})$ now is a traditional CPPI strategy with floor $F(t)$ and constant multiple vector ${\frac{1}{1-\hat{b}} \Sigma^{-1} (\mu - r \mathbf{1})}$. The optimal consumption rate $c^{\star}_{(b(t) \equiv \hat{b})}(t;v_{0})$ is the sum of the consumption floor $\bar{c}(t)$ and the time-varying proportion $\zeta(t)$ of the cushion ${V^{\star}_{(b(t) \equiv \hat{b})}(t ; v_{0}) - F(t)}$ at time $t$. The fraction between the risky exposure (vector) and consumption is time-varying and it holds
\begin{align}
\frac{\hat{\pi}^{\star}_{(b(t) \equiv \hat{b})}(t ; v_{0})  V^{\star}_{(b(t) \equiv \hat{b})}(t ; v_{0})}{c^{\star}_{(b(t) \equiv \hat{b})}(t;v_{0})} =  \frac{1}{1-\hat{b}} \Sigma^{-1} (\mu - r \mathbf{1})  \left(\zeta(t) + \frac{\bar{c}(t)}{V^{\star}_{(b(t) \equiv \hat{b})}(t ; v_{0}) - F(t)}\right)^{-1}. \label{eq:fraction:b(t)constant}
\end{align}
Optimal consumption $c^{\star}_{(b(t) \equiv \hat{b})}(t;v_{0})$ as well as, under ${\Sigma^{-1} (\mu - r \mathbf{1}) > \mathbf{0}}$, optimal risky exposure \linebreak ${\hat{\pi}^{\star}_{(b(t) \equiv \hat{b})}(t ; v_{0}) V^{\star}_{(b(t) \equiv \hat{b})}(t ; v_{0})}$ linearly increase in the cushion ${V^{\star}_{(b(t) \equiv \hat{b})}(t ; v_{0}) - F(t)}$. Hence, the higher the surplus ${V^{\star}_{(b(t) \equiv \hat{b})}(t ; v_{0}) - F(t)}$, the more is invested risky and the more is consumed. The formula \eqref{eq:fraction:b(t)constant} shows that, under ${\Sigma^{-1} (\mu - r \mathbf{1}) > \mathbf{0}}$, an increase in the cushion ${V^{\star}_{(b(t) \equiv \hat{b})}(t ; v_{0}) - F(t)}$ leads to a stronger increase in the risky exposure ${\hat{\pi}^{\star}_{(b(t) \equiv \hat{b})}(t ; v_{0}) V^{\star}_{(b(t) \equiv \hat{b})}(t ; v_{0})}$ than in consumption $c^{\star}_{(b(t) \equiv \hat{b})}(t;v_{0})$. Therefore, for a larger surplus ${V^{\star}_{(b(t) \equiv \hat{b})}(t ; v_{0}) - F(t)}$, also the relative increase in the risky exposure is larger than the relative increase in consumption, thus investing money in stocks is preferred to consuming.

The associated optimal wealth process is given as a function of the pricing kernel
\begin{align*}
V^{\star}_{(b(t) \equiv \hat{b})}(t ; v_{0}) = \frac{1}{\zeta(t)} \tilde{Z}(t)^{\frac{1}{\hat{b}-1}} \left(v_{0} - F(0)\right) e^{\frac{\hat{b}}{\hat{b}-1} \left(r - \frac{1}{2} \frac{1}{\hat{b}-1} \|\gamma\|^{2}\right) t} \frac{\chi(t)}{\int_{0}^{T} \chi(t) dt + 1} + F(t).
\end{align*}
This special case result coincides with the findings by \cite{Ye2008}, who used the HJB approach, extended by additionally providing the optimal wealth process $V^{\star}_{(b(t) \equiv \hat{b})}(t ; v_{0})$.
\end{remark}
We aim to interpret the optimal $\hat{\pi}^{\star}(t ; v_{0})$ for time-varying $b(t)$ and particularly to point out the difference to constant $b(t)$ in Remark \ref{remark:Ye:case}. Writing ${V_{1}(t ; v_{1}^{\star}) = V^{\star}(t ; v_{0}) - V_{2}(t ; v_{2}^{\star})}$ where $V_{2}(t ; v_{2}^{\star})$ follows the wealth process of a standard CPPI strategy with floor $F_{2}(t)$ at time $t$ to the initial endowment $v_{2}^{\star}$ and constant multiplier vector ${\frac{1}{1-\hat{b}} \Sigma^{-1}  (\mu - r \mathbf{1})}$, we obtain the following representation of the optimal investment decision

\begin{align}
\hat{\pi}^{\star}(t ; v_{0}) = {} & \Sigma^{-1} (\mu - r \mathbf{1}) \frac{\frac{1}{1 - b(\tilde{t}_{1}^{\star})} \left(V^{\star}(t ; v_{0}) - V_{2}(t ; v_{2}^{\star}) - F_{1}(t)\right) + \frac{1}{1-\hat{b}} \left(V_{2}(t ; v_{2}^{\star}) - F_{2}(t)\right)}{V^{\star}(t ; v_{0})} \nonumber \\
= {} & \Sigma^{-1} (\mu - r \mathbf{1}) \left\{\frac{1}{1 - b(\tilde{t}_{1}^{\star})} \frac{V^{\star}(t ; v_{0}) - F_{1}(t)}{V^{\star}(t ; v_{0})} + \frac{\hat{b}- b(\tilde{t}_{1}^{\star})}{(1-\hat{b}) \left(1 - b(\tilde{t}_{1}^{\star})\right)} \frac{V_{2}(t ; v_{2}^{\star}) - \frac{1 - b(\tilde{t}_{1}^{\star})}{\hat{b}- b(\tilde{t}_{1}^{\star})} F_{2}(t)}{V^{\star}(t ; v_{0})}\right\} \nonumber \\
= {} & \Sigma^{-1} (\mu - r \mathbf{1}) \left\{\frac{1}{1 - b(\tilde{t}_{1}^{\star})} \frac{V^{\star}(t ; v_{0}) - F(t)}{V^{\star}(t ; v_{0})} + \frac{\hat{b} - b(\tilde{t}_{1}^{\star})}{(1-\hat{b}) \left(1 - b(\tilde{t}_{1}^{\star})\right)} \frac{V_{2}(t ; v_{2}^{\star}) - F_{2}(t)}{V^{\star}(t ; v_{0})}\right\} \nonumber \\
= {} & \Sigma^{-1} (\mu - r \mathbf{1}) \left\{\frac{1}{1 - b(\tilde{t}_{1}^{\star})} \frac{V^{\star}(t ; v_{0}) - F(t)}{V^{\star}(t ; v_{0})} + \frac{\hat{b} - b(\tilde{t}_{1}^{\star})}{(1-\hat{b}) \left(1 - b(\tilde{t}_{1}^{\star})\right)} \frac{V_{2}(t ; v_{2}^{\star})}{V^{\star}(t ; v_{0})} \frac{V_{2}(t ; v_{2}^{\star}) - F_{2}(t)}{V_{2}(t ; v_{2}^{\star})}\right\}, \label{eq:SeparationTechnique:pihat:2xPPI}
\end{align}
which can be implemented easily; $F(t)$ is defined in \eqref{eq:def:F(t)}. Formula \eqref{eq:SeparationTechnique:pihat:2xPPI} shows that the optimal relative allocation $\hat{\pi}^{\star}(t ; v_{0})$ can be written as a PPI strategy in $V^{\star}(t ; v_{0})$ with floor $F(t)$ plus a PPI strategy in $V_{2}(t ; v_{2}^{\star})$ with floor $F_{2}(t)$. Alternatively, write ${V_{2}(t ; v_{2}^{\star}) = V^{\star}(t ; v_{0}) - V_{1}(t ; v_{1}^{\star})}$, where $V_{1}(t ; v_{1}^{\star})$ is the replicating wealth process of a PPI strategy with floor $F_{1}(t)$ to the initial wealth $v_{1}^{\star}$ and now time- and state-varying multiplier vector $\frac{1}{1 - b(\tilde{t}_{1})} \Sigma^{-1} (\mu - r \mathbf{1})$ and, in contrast to $V_{2}(t ; v_{2}^{\star})$, a non-zero consumption rate process. Then $\hat{\pi}^{\star}(t ; v_{0})$ can be reformulated as
\begin{align}
\hat{\pi}^{\star}(t ; v_{0}) = {} & \Sigma^{-1} (\mu - r \mathbf{1}) \left\{\left(\frac{1}{1 - b(\tilde{t}_{1}^{\star})} - \frac{1}{1-\hat{b}}\right) \frac{V_{1}(t ; v_{1}^{\star}) - F_{1}(t)}{V^{\star}(t ; v_{0})} + \frac{1}{1-\hat{b}} \frac{V^{\star}(t ; v_{0}) - F(t)}{V^{\star}(t ; v_{0})}\right\} \nonumber \\
= {} & \Sigma^{-1} (\mu - r \mathbf{1}) \left\{\frac{1}{1-\hat{b}} \frac{V^{\star}(t ; v_{0}) - F(t)}{V^{\star}(t ; v_{0})} - \frac{\hat{b} - b(\tilde{t}_{1}^{\star})}{(1-\hat{b}) (1 - b(\tilde{t}_{1}^{\star}))} \frac{V_{1}(t ; v_{1}^{\star}) - F_{2}(t)}{V^{\star}(t ; v_{0})}\right\} \nonumber \\
\begin{split} \label{eq:SeparationTechnique:pihat:V1:1xCPPI:1xPPI}
= {} & \Sigma^{-1} (\mu - r \mathbf{1}) \left\{\frac{1}{1-\hat{b}} \frac{V^{\star}(t ; v_{0}) - F(t)}{V^{\star}(t ; v_{0})} - \frac{\hat{b} - b(\tilde{t}_{1}^{\star})}{(1-\hat{b}) (1 - b(\tilde{t}_{1}^{\star}))} \frac{V_{1}(t ; v_{1}^{\star})}{V^{\star}(t ; v_{0})} \frac{V_{1}(t ; v_{1}^{\star}) - F_{1}(t)}{V_{1}(t ; v_{1}^{\star})}\right\}.
\end{split}
\end{align}
This formula shows that the optimal relative investment $\hat{\pi}^{\star}(t ; v_{0})$ is the sum of a conventional CPPI strategy on $V^{\star}(t ; v_{0})$ with floor $F(t)$ and a PPI strategy on $V_{1}(t ; v_{1}^{\star})$ with floor $F_{1}(t)$.

Recall from Remark \ref{remark:Ye:case} that $\hat{\pi}^{\star}_{(b(t) \equiv \hat{b})}(t ; v_{0})$ for constant $b(t) \equiv \hat{b}$ follows a traditional CPPI strategy ${\frac{1}{1-\hat{b}} \Sigma^{-1} (\mu - r \mathbf{1}) \frac{V^{\star}_{(b(t) \equiv \hat{b})}(t ; v_{0}) - F(t)}{V^{\star}_{(b(t) \equiv \hat{b})}(t ; v_{0})}}$ to the floor $F(t)$. The formula for $\hat{\pi}^{\star}(t ; v_{0})$ in \eqref{eq:SeparationTechnique:pihat:V1:1xCPPI:1xPPI} shows that the optimal strategy $\hat{\pi}^{\star}(t ; v_{0})$ for time-varying $b(t)$ consists of two parts:
\begin{enumerate}
\item The first part coincides with $\hat{\pi}^{\star}_{(b(t) \equiv \hat{b})}(t ; v_{0})$ and is a traditional CPPI strategy ${\frac{1}{1-\hat{b}} \Sigma^{-1} (\mu - r \mathbf{1}) \frac{V^{\star}(t ; v_{0}) - F(t)}{V^{\star}(t ; v_{0})}}$ in $V^{\star}(t ; v_{0})$ to the floor $F(t)$.
\item The second, additional part is a time- and state-varying term which can be either positive, negative or zero; hence it can reduce or increase risky investments or can leave it unmodified in comparison with $\hat{\pi}^{\star}_{(b(t) \equiv \hat{b})}(t ; v_{0})$.
\end{enumerate}
It is the second part which leads to a deviation in $\hat{\pi}^{\star}(t ; v_{0})$ compared to $\hat{\pi}^{\star}_{(b(t) \equiv \hat{b})}(t ; v_{0})$. For this sake, we analyze this second piece in what follows. Note that by Theorem \ref{thm:Solution:ConsumptionOnly} it holds ${V_{1}(t ; v_{1}) > F_{1}(t)}$ a.s..
\begin{enumerate}
\item If ${V^{\star}(t ; v_{0}) > 0}$, for instance this is reasonable for ${v_{0} > 0}$ and an income rate that outweighs or exceeds consumption, then it follows
\begin{align*}
\frac{V_{1}(t ; v_{1}^{\star})}{V^{\star}(t ; v_{0})} \frac{V_{1}(t ; v_{1}^{\star}) - F_{1}(t)}{V_{1}(t ; v_{1}^{\star})} = \frac{V_{1}(t ; v_{1}^{\star}) - F_{1}(t)}{V^{\star}(t ; v_{0})} > 0.
\end{align*}
This implies for ${i = 1, \hdots, N}$ at time $t$:
\begin{enumerate}
\item $\hat{b} > b(\tilde{t}_{1}^{\star})$:
\begin{align*}
\left(\hat{\pi}^{\star}(t ; v_{0})\right)_{i} < \frac{1}{1-\hat{b}} \left(\Sigma^{-1} (\mu - r \mathbf{1})\right)_{i} \frac{V^{\star}(t ; v_{0}) - F(t)}{V^{\star}(t ; v_{0})}\ \Leftrightarrow\ \left(\Sigma^{-1} (\mu - r \mathbf{1})\right)_{i} > 0.
\end{align*}
\item $\hat{b} = b(\tilde{t}_{1}^{\star})$:
\begin{align*}
\hat{\pi}^{\star}(t ; v_{0}) = \frac{1}{1-\hat{b}} \Sigma^{-1} (\mu - r \mathbf{1}) \frac{V^{\star}(t ; v_{0}) - F(t)}{V^{\star}(t ; v_{0})}.
\end{align*}
\item $\hat{b} < b(\tilde{t}_{1}^{\star})$:
\begin{align*}
\left(\hat{\pi}^{\star}(t ; v_{0})\right)_{i} > \frac{1}{1-\hat{b}} \left(\Sigma^{-1} (\mu - r \mathbf{1})\right)_{i} \frac{V^{\star}(t ; v_{0}) - F(t)}{V^{\star}(t ; v_{0})}\ \Leftrightarrow\ \left(\Sigma^{-1} (\mu - r \mathbf{1})\right)_{i} < 0.
\end{align*}
\end{enumerate}
\item If ${V^{\star}(t ; v_{0}) < 0}$, for instance this is reasonable for ${v_{0} < 0}$ and a high demand for consumption in the past, then it follows
\begin{align*}
\frac{V_{1}(t ; v_{1}^{\star})}{V^{\star}(t ; v_{0})} \frac{V_{1}(t ; v_{1}^{\star}) - F_{1}(t)}{V_{1}(t ; v_{1}^{\star})} = \frac{V_{1}(t ; v_{1}^{\star}) - F_{1}(t)}{V^{\star}(t ; v_{0})} > 0.
\end{align*}
This in turn implies for ${i = 1, \hdots, N}$ at time $t$:
\begin{enumerate}
\item $\hat{b} > b(\tilde{t}_{1}^{\star})$:
\begin{align*}
\left(\hat{\pi}^{\star}(t ; v_{0})\right)_{i} > \frac{1}{1-\hat{b}} \left(\Sigma^{-1} (\mu - r \mathbf{1})\right)_{i} \frac{V^{\star}(t ; v_{0}) - F(t)}{V^{\star}(t ; v_{0})}\ \Leftrightarrow\ \left(\Sigma^{-1} (\mu - r \mathbf{1})\right)_{i} > 0.
\end{align*}
\item $\hat{b} = b(\tilde{t}_{1}^{\star})$:
\begin{align*}
\hat{\pi}^{\star}(t ; v_{0}) = \frac{1}{1-\hat{b}} \Sigma^{-1} (\mu - r \mathbf{1}) \frac{V^{\star}(t ; v_{0}) - F(t)}{V^{\star}(t ; v_{0})}.
\end{align*}
\item $\hat{b} < b(\tilde{t}_{1}^{\star})$:
\begin{align*}
\left(\hat{\pi}^{\star}(t ; v_{0})\right)_{i} < \frac{1}{1-\hat{b}} \left(\Sigma^{-1} (\mu - r \mathbf{1})\right)_{i} \frac{V^{\star}(t ; v_{0}) - F(t)}{V^{\star}(t ; v_{0})}\ \Leftrightarrow\ \left(\Sigma^{-1} (\mu - r \mathbf{1})\right)_{i} < 0.
\end{align*}
\end{enumerate}
\end{enumerate}
In particular, consider the situation ${V^{\star}(t ; v_{0}) > 0}$ and let ${\left(\Sigma^{-1} (\mu - r \mathbf{1})\right)_{i} > 0}$ hold for risky asset $i$. Under ${\hat{b} > b(\tilde{t}_{1}^{\star})}$, the optimal relative investment in stock $i$, which is $\left(\hat{\pi}^{\star}(t ; v_{0})\right)_{i}$, is reduced compared to the relative investment decision $\left(\hat{\pi}^{\star}_{(b(t) \equiv \hat{b})}(t ; v_{0})\right)_{i}$ under ${b(t) \equiv \hat{b}}$. Since ${\hat{b} > b(\tilde{t}_{1}^{\star})}$ can be interpreted as higher risk aversion for consumption than terminal wealth, this is meaningful.

In the situation ${V^{\star}(t ; v_{0}) < 0}$ the interpretation seems counterintuitive at first glance. But when looking at risky exposures rather than risky relative investments, analogue conclusions hold. The same approach shall be used when considering ${V^{\star}(t ; v_{0}) = 0}$.

Furthermore, it is worth to mention that $\hat{\pi}^{\star}(t ; v_{0})$ approaches $0$ when $V^{\star}(t ; v_{0})$ approaches $F(t)$, which can be observed in \eqref{eq:SeparationTechnique:pihat:2xPPI}; the argument is the following: When $V^{\star}(t ; v_{0})$ falls towards $F(t)$, then automatically $V_{1}(t ; v_{1}^{\star})$ approaches $F_{1}(t)$ and $V_{2}(t ; v_{2}^{\star})$ converges towards $F_{2}(t)$ simultaneously, since $V^{\star}(t ; v_{0}) = V_{1}(t ; v_{1}^{\star}) + V_{2}(t ; v_{2}^{\star})$, ${F(t) = F_{1}(t) + F_{2}(t)}$ and ${V_{1}(t ; v_{1}^{\star}) > F_{1}(t)}$, ${V_{2}(t ; v_{2}^{\star}) > F_{2}(t)}$ a.s. which was already shown in Sections \ref{sec:ConsumptionProblem:y} and \ref{sec:TerminalWealthProblem}. We moreover proved that in this case ${\hat{\pi}_{1}(t;v_{1}^{\star})}$ and ${\hat{\pi}_{2}(t;v_{2}^{\star})}$ approach $0$. By Theorem \ref{thm:Solution:Merging:OriginalProblem} it follows that also $\hat{\pi}^{\star}(t; v_{0})$ must converge to $0$. Therefore, as ${v_{0} > F(0)}$ is assumed, it follows that ${V^{\star}(t ; v_{0}) > F(t)}$ a.s., which can additionally be seen in the respective formula in Theorem \ref{thm:Solution:Merging:OriginalProblem}, and the optimal decision rules provide portfolio insurance over the whole life-cycle. $F(t)$ is called the minimum asset wealth level, it holds $F(T) = F$.

The optimal exposure to the risky assets equals the sum of the optimal risky exposures of the two sub-problems
\begin{align*}
\hat{\pi}^{\star}(t; v_{0}) V^{\star}(t ; v_{0}) = \hat{\pi}_{1}(t;v_{1}^{\star}) V_{1}(t;v_{1}^{\star}) + \hat{\pi}_{2}(t;v_{2}^{\star}) V_{2}(t;v_{2}^{\star})
\end{align*}
and by the findings in Sections \ref{sec:ConsumptionProblem:y} and \ref{sec:TerminalWealthProblem} it holds
\begin{align*}
\left(\hat{\pi}^{\star}(t ; v_{0}) V^{\star}(t ; v_{0})\right)_{i} > 0\ \Leftrightarrow\ \left(\Sigma^{-1} (\mu - r \mathbf{1})\right)_{i} > 0.
\end{align*}
For the ease of exposition we so far assumed that the income process is deterministic. The following remark shows the solution for a stochastic income process.

\begin{remark}[Stochastic income process] \label{remark:y(t)stochastic}
Let $(y(t))_{t \in [0,T]}$ be a non-negative, stochastic income-rate process with $\int_{0}^{T} y(t) dt < \infty$, $\IP$-a.s.. The stated results are still valid after replacing integrals of the form ${\int_{t}^{T} e^{-r (s-t)} y(s) ds}$ by the more general conditional expectation \linebreak ${\IE\left[\int_{t}^{T} \frac{\tilde{Z}(s)}{\tilde{Z}(t)} y(s) ds \Big| \mathcal{F}_{t}\right] = \int_{t}^{T} \IE\left[\frac{\tilde{Z}(s)}{\tilde{Z}(t)} y(s) \Big| \mathcal{F}_{t}\right] ds = \int_{t}^{T} e^{-r (s-t)} \IE_{\IQ}\left[y(s) \Big| \mathcal{F}_{t}\right] ds}$ by the Bayes formula for arbitrary ${t \in [0,T]}$, in particular in the definition of $F_{1}(t)$ and $F(t)$. If $(y(t))_{t \in [0,T]}$ is supposed to be independent to $\CMcal{F}$, i.e. independent to the market stochastics, then the conditional expectation ${\IE\left[\int_{t}^{T} \frac{\tilde{Z}(s)}{\tilde{Z}(t)} y(s) ds \Big| \mathcal{F}_{t}\right]}$ can be reduced to ${\int_{t}^{T} e^{-r (s-t)} \IE_{\IQ}\left[y(s)\right] ds}$. For the lower bounds of $v_{0}$ and $v_{1}$, \eqref{eq:Condition:v0:minimalrequirement} and \eqref{eq:Condition:v1_&_eq:def:F1(t)} need to be replaced by
\begin{align*} 
v_{0} > {} & \int_{0}^{T} e^{- r s} \left(\bar{c}(s) - \bar{y}(s)\right) ds + e^{- r T} F, \\
v_{1} > {} & \int_{0}^{T} e^{- r s} \left(\bar{c}(s) - \bar{y}(s)\right) ds,
\end{align*}
where ${\bar{y}(s) = \sup \left\{x \ge 0: \IP(y(s) \ge x) = 1\right\}}$ denotes the minimal level of income; ${\bar{y}(s) > 0}$ is meaningful due to unemployment benefits paid by the government.
\end{remark}

\section{Analysis of optimal controls and wealth process: A case study}
\label{sec:NCS}
This section targets to calibrate the life-cycle model to realistic time-dependent structures for consumption and investment observed in practice and outline the difference between our presented solution with age-depending $a(t)$ and $b(t)$ functions and the models with either only $a(t)$ or $b(t)$ time-varying or none. Hence, we not only estimate $\hat{b}$, $a(t)$ and $b(t)$ for our model, but additionally provide the respective estimates when $a(t)$ or $b(t)$, or both, are assumed to be constants. A comparison of the fit of the different models allows for making a statement on the accuracy of the models in describing the agent's behavior. For notational convenience we call the three benchmark models as follows:
\begin{itemize}
\item $M_{a,b(t)}$: $a(t) \equiv a$ constant, $b(t)$ time-varying
\item $M_{a(t),b}$: $a(t)$ time-varying, $b(t) \equiv b$ constant
\item $M_{a,b}$: $a(t) \equiv a$ and $b(t) \equiv b$ constant
\end{itemize}
The subscript thus indicates whether $a(t)$ or $b(t)$ are age-varying. Therefore, our model is denoted by $M_{a(t),b(t)}$. As already indicated in Section \ref{sec:Introduction}, $M_{a,b(t)}$ is (partially) covered by \cite{Steffensen2011}, \cite{Hentschel2016} and \cite{Aase2017}, $M_{a(t),b}$ and $M_{a,b}$ are covered by \cite{Ye2008}.

In the later Subsection \ref{sec:ComparisonCRRA}, we additionally analyze the impact of the floors $\bar{c}(t)$ and $F$, where our model $M_{a(t),b(t)}$ is compared to the same model but with CRRA utility functions, i.e. $\bar{c}(t) \equiv 0$ and $F \equiv 0$. The CRRA model is denoted by $M_{a(t),b(t)}^{CRRA}$ and is (partially) considered by \cite{Steffensen2011}, \cite{Hentschel2016} and \cite{Aase2017}.

\subsection{Assumptions}
We assume an exemplary agent with average income, liabilities etc. A similar case study can be carried out for a pension cohort, but for simplicity and data availability we consider an individual client. In detail, we make the following (simplifying) assumptions:

Let the market consist of one risk-free and one risky asset ($N = 1$) with parameters ${r = 0.5 \%}$, ${\mu = 5 \%}$, and ${\sigma = 20 \%}$; these values correspond approximately to the EURONIA Overnight Rate and the performance of the DAX 30 Performance Index as an equity index over the $11$ year period from 17 October 2007 to 17 October 2018. The risky asset can coincide with, but is not restricted to a pure equity portfolio. In general it can be any arbitrary given portfolio which consists of risky assets. The price process of the risky asset is assumed to be ${P(t) = p_{1} e^{(\mu - \frac{1}{2} \sigma^2) t + \sigma W(t)} = p_{1} e^{\frac{1}{2} (\mu + r) (1 - \frac{\sigma}{\gamma})} \tilde{Z}(t)^{- \frac{\sigma}{\gamma}}}$ with initial price ${P(0) = p_{1} = 100}$. Furthermore, let ${T = 40}$ years be the time to retirement, $t = 25$ years the current age of the investor and $65$ years the age of retirement. For the net salary function it is assumed ${y(t) = \frac{\tilde{r}}{e^{\tilde{r}} - 1} y_{0} e^{\tilde{r} t}}$ with $y_{0} = 26,200$ EUR and $\tilde{r} = 2.07 \%$. This corresponds to a net annual starting salary approximately equal to the average for a graduate in Germany in 2017 (cf. online portals \cite{Absolventa2018} or \cite{Stepstone2017}), with an annual increase equal to the average for a household's net salary in Germany over years 2011 to 2016 according to \cite{StatistischesBundesamtEinkommen2018}. Net income accumulated over the first year is ${\int_{0}^{1} y(t)dt = y_{0}}$ and income accumulated within the year from time $s$ to $s+1$ is ${\int_{s}^{s+1} y(t)dt = \frac{\tilde{r}}{e^{\tilde{r}} - 1} y_{0} \frac{e^{\tilde{r} (s+1)} - e^{\tilde{r} s}}{\tilde{r}} = y_{0} e^{\tilde{r} s}}$.

For the agent's utility functions, let ${\beta = 3 \%}$ (cf. \cite{Ye2008}) and ${\hat{a} = 1}$. Let the terminal wealth floor be ${F = 435,125}$ EUR which is motivated by the following argument: According to \cite{StatistischesBundesamtJahrbuch2017}, \cite{aktuare2017} or \cite{WKO2016} a lifetime around $81$ years can be expected for a currently $25$ year old person in Germany. Thus survival of $81 - 65 = 16$ years are expected after retirement at the age of $65$. We assume that the agent secures the income inflow during retirement to be $75 \%$ of the last wage paid from year $64$ to $65$ (replacement ratio of $75 \%$), which is ${\int_{39}^{40} y(t)dt = y_{0} e^{39 \tilde{r}} = 58,736}$ EUR. Assume that every year, half of this amount is covered by a separate pension account or plan, e.g. provided by the government. In addition, the agent wants to secure against longevity risk, hence considers $16 \times (100 + 30) \% = 20.8$ years instead of $16$ years for the remaining lifetime after the age of retirement. Thus, $F$ as value at time $T$ is chosen to be ${F = \int_{0}^{20.8} \frac{0.75 \times 58,736 \text{ EUR}}{2} e^{-r t} dt = \frac{0.75 \times 58,736 \text{ EUR}}{2} \left(\frac{1 - e^{- 20.8 \times r}}{r}\right) = 435,125 \text{ EUR}}$. Finally, the function for the net consumption floor is supposed to take the form ${\bar{c}(t) = \frac{\bar{r}}{e^{\bar{r}} - 1} \bar{c}_{0} e^{\bar{r} t}}$ with $\bar{c}_{0} = 14,880$ EUR and $\bar{r} = 1.93 \%$. This corresponds to a starting value equal to approximately $50 \%$ of the average household consumption in Germany in 2016 as starting point, with an annual increase equal to the increase in average household consumption in Germany over years 2011 to 2016 (published by \cite{StatistischesBundesamtEinkommen2018}). Minimum consumption expenses incurred within the first year is ${\int_{0}^{1} \bar{c}(t)dt = \bar{c}_{0}}$, within year $s$ to $s+1$ is ${\int_{s}^{s+1} \bar{c}(t)dt = \bar{c}_{0} e^{\bar{r} s}}$. The assumed income and consumption floor rates are visualized in Figure \ref{fig:Analysis:cbar:y}.

\begin{figure}[h] \center
\subfigure{\includegraphics[width=0.5\textwidth]{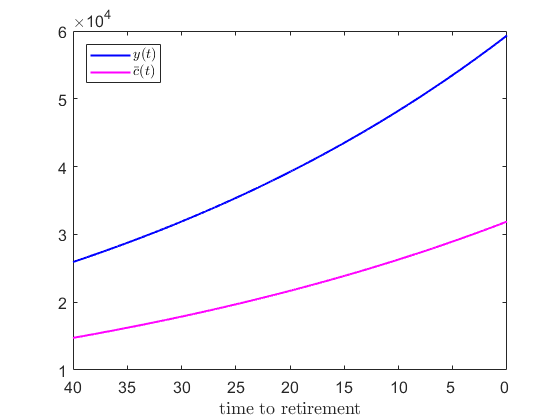}}
\caption{Income rate $y(t)$ and consumption floor rate $\bar{c}(t)$ (in EUR).}
\label{fig:Analysis:cbar:y}
\end{figure}

\subsection{Fitting / Calibration under exponential preferences and discussion}
In what follows we calibrate the remaining utility parameters $\hat{b}$, $a(t)$ and $b(t)$ to suitable curves for consumption and relative allocation. The targeted curves for parameter fitting are summarized by Table \ref{tab:Analysis:SamplePoints}. The consumption rate $c^{\star}(t;v_{0})$ is calibrated with respect to the hump-shaped type observed by \cite{Carroll1997}, \cite{GourinchasParker2002}, \cite{JensenSteffensen2015} and \cite{TangPurcalZhang2018}. The relative risky investment $\hat{\pi}^{\star}(t ; v_{0})$ is calibrated towards the $(100 - \text{age}) \%$ rule of thumb; a similar structure is frequently applied by financial advisors and asset management companies for life-cycle funds (see \cite{Malkiel1990}, \cite{BodieCrane1997}, \cite{Shiller2005}, \cite{Minderhoud2011}, \cite{Milliman2010}, \cite{Shafir2013}). Following this popular rule, the client at age $25$ years starts with a $75 \%$ equity investment, linearly decreases it by her age such that she ends with a $35 \%$ investment in equities at the age of retirement with $65$ years. We would like to mention that in particular relative risky investment curves or products provided by asset management companies are to be understood deterministic, i.e. wealth- / state-independent. Therefore, we calibrate the remaining unknown parameters with respect to the expected values for consumption and risky relative investment. In more detail, we fit the expected value for consumption, which is ${\IE\left[c^{\star}(t;v_{0})\right]}$, to the given consumption curve. For ${\IE\left[\hat{\pi}^{\star}(t ; v_{0})\right]}$ we apply the following estimate: we estimate the risky exposure ${\IE\left[\hat{\pi}^{\star}(t ; v_{0}) V^{\star}(t ; v_{0})\right]}$ without any bias and then replace $V^{\star}(t ; v_{0})$ by its unbiased expectation ${\IE\left[V^{\star}(t ; v_{0})\right]}$ to obtain the estimate ${\frac{\IE\left[\hat{\pi}^{\star}(t ; v_{0}) V^{\star}(t ; v_{0})\right]}{\IE\left[V^{\star}(t ; v_{0})\right]}}$ for ${\IE\left[\hat{\pi}^{\star}(t ; v_{0})\right]}$. By doing this we replace ${\IE\left[\hat{\pi}^{\star}(t ; v_{0})\right]}$ by ${\frac{\IE\left[\hat{\pi}^{\star}(t ; v_{0}) V^{\star}(t ; v_{0})\right]}{\IE\left[V^{\star}(t ; v_{0})\right]}}$ and fit the latter expression to the given linear relative investment curve. For further readings on deterministic investment strategies we refer to \cite{ChristiansenSteffensen2013} and \cite{ChristiansenSteffensen2018}. In summary, we have unbiased estimates for the expected values of optimal consumption, risky exposure and wealth process, and a modified estimate for the expectation of the optimal relative risky investment.

\begin{table}[t]
\raggedright
{
    \renewcommand{\arraystretch}{1.75}
    \begin{tabularx}{\textwidth}{| p{5.5cm} | X |}
        \hline
        $\hat{\pi}^{\star}(t ; v_{0})$ & $c^{\star}(t;v_{0})$  \\ \hline
        $y(t) = \frac{100 - (t+25)}{100}$ in EUR, ${t \in [0,T]}$ & $c(t) = -25 (t - 26)^{2} + 37,732$ in EUR, ${t \in [0,T]}$ \\
        $(100 - \text{age}) \%$ rule (total stock ratio) & thus ${c(0) = 20,832}$ EUR (${= 70 \%}$ of average household consumption in Germany in 2016, cf. \cite{StatistischesBundesamtEinkommen2018}, as starting consumption rate), turning point at $t = 26$ (age $51$) with a maximum targeted consumption of $37,732$ EUR. \\ \hline
    \end{tabularx}
}
\caption{Target curves for calibration.}
\label{tab:Analysis:SamplePoints}
\end{table}

Let $a(t)$ and $b(t)$ take the form of an exponential function, i.e. ${a(t) = a_{0} e^{\lambda_{a} t}}$ and ${b(t) = b_{0} e^{\lambda_{b} t}}$. Moreover, let $v_{0} = 250,000$ EUR. The estimation is carried out via the Matlab function \textit{lsqcurvefit} which solves nonlinear curve-fitting (data-fitting) problems in a least-squares sense and minimizes the sum of the squared relative distances. The underlying time points for target consumption and allocation are set weekly on an equidistant grid which yields $2,080$ points in the time interval $[0,T]$ with ${T = 40}$.

Table \ref{tab:Analysis:ParameterEstimation:Error:Benchmark:RichInvestor} gives an overview of the estimated utility parameters and provides the sum of squared relative errors as a quality criterion. The errors show that considering age-depending functions $a(t)$ and $b(t)$ simultaneously in model $M_{a(t),b(t)}$ leads to a comparatively huge improvement in accuracy of the fit compared to any of the three benchmark models: model $M_{a(t),b(t)}$ sum of squared relative distances is only $19.38 \%$ of the respective sum for model $M_{a(t),b}$ which provides the second best fit in terms of sum of squared relative residuals.

\begin{table}[h]
\raggedright
{
    \renewcommand{\arraystretch}{1.75}
    \begin{tabularx}{\textwidth}{| p{1.5cm} || p{5cm}  | X | X | X |}
        \hline
        & Sum of squared relative distances & $\hat{b}$ & $a(t)$ & $b(t)$  \\ \hline\hline
        $M_{a(t),b(t)}$ & $6.0425$ & $-0.9849$ & ${a_{0} = 5.2864 \times 10^{7}}$, ${\lambda_{a} = -0.6673}$ & ${b_{0} = -4.9731}$, ${\lambda_{b} = -0.0340}$ \\ \hline
       $M_{a,b(t)}$ & $31.3157$ & $-0.8325$ & ${a_{0} = 0.7997 \times 10^{7}}$, ${\lambda_{a} := 0}$ & ${b_{0} = -4.0243}$, ${\lambda_{b} = 0.0012}$ \\ \hline
       $M_{a(t),b}$ & $31.1801$ & $-0.8344$ & ${a_{0} = 1.8187 \times 10^{7}}$, ${\lambda_{a} = -0.0363}$ & ${b_{0} = -4.1441}$, ${\lambda_{b} := 0}$ \\ \hline
       $M_{a,b}$ & $33.5350$ & $-0.8247$ & ${a_{0} = 0.3425 \times 10^{7}}$, ${\lambda_{a} := 0}$ & ${b_{0} = -3.9697}$, ${\lambda_{b} := 0}$ \\ \hline
    \end{tabularx}
}
\caption{Estimated parameters and sum of squared relative residuals.}
\label{tab:Analysis:ParameterEstimation:Error:Benchmark:RichInvestor}
\end{table}

Figure \ref{fig:Analysis:InputFunctions:Benchmark:RichInvestor} visualizes the fitted parameters and preference functions $\hat{b}$, $a(t)$, $b(t)$. The table and figure show that the estimated coefficient of risk aversion $\hat{b}$ for our model $M_{a(t),b(t)}$ is more negative, which means a higher risk aversion, compared to the three benchmark models $M_{a,b(t)}$, $M_{a(t),b}$, $M_{a,b}$. Furthermore, $a(t)$ is decreasing both within model $M_{a(t),b(t)}$ and $M_{a(t),b}$. In contrast, $b(t)$ increases in model $M_{a(t),b(t)}$ over time whereas it decreases in the comparison model $M_{a,b(t)}$. $b(t)$ in models $M_{a,b(t)}$, $M_{a(t),b}$, $M_{a,b}$ stay very close over the whole life-cycle whereas $b(t)$ in $M_{a(t),b(t)}$ starts more negative and ends less negative. In summary, this means that in model $M_{a(t),b(t)}$ the risk aversion decreases through increasing $b(t)$, but preference of the investor between consumption and terminal wealth is shifted more and more to terminal wealth through decreasing $a(t)$.

\begin{figure}[h] \center
\subfigure[$a(t)$.]{\includegraphics[width=0.5\textwidth]{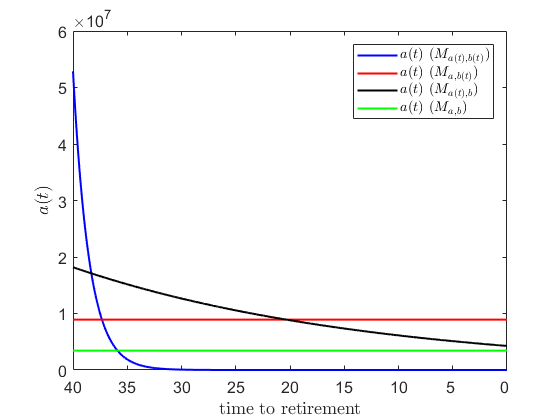}}\hfill
\subfigure[$b(t)$.]{\includegraphics[width=0.5\textwidth]{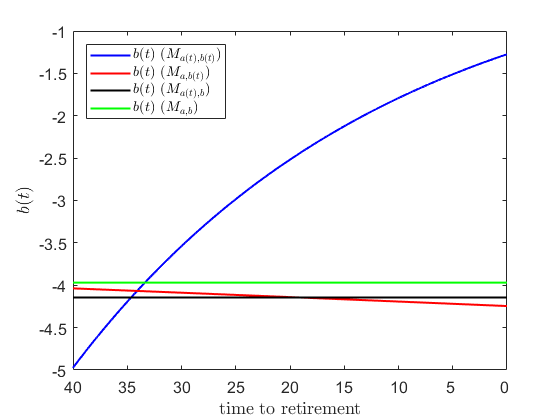}}
\caption{Estimated preference functions $a(t)$ and $b(t)$.}
\label{fig:Analysis:InputFunctions:Benchmark:RichInvestor}
\end{figure}

Figure \ref{fig:Analysis:FittedConsumptionInvestment:Benchmark:RichInvestor} illustrates the expected optimal consumption rate and relative risky investment for the fitted parameters in comparison with the given target policies or average profile. In addition to Table \ref{tab:Analysis:ParameterEstimation:Error:Benchmark:RichInvestor} the figure illustrates that, under exponential preferences $a(t)$ and $b(t)$, only the most flexible model $M_{a(t),b(t)}$ provides an accurate and precise fit for both consumption rate and risky relative allocation. We realize that the benchmark models $M_{a,b(t)}$, $M_{a(t),b(t)}$, $M_{a,b}$ apparently do not provide enough flexibility to simultaneously describe the predetermined consumption and relative allocation curves. Whereas the fits for the relative investment $\hat{\pi}^{\star}(t ; v_{0})$ look acceptable, all three benchmark models fail in explaining the targeted consumption rate $c^{\star}(t;v_{0})$. We further notice that $c^{\star}(t;v_{0})$ and $\hat{\pi}^{\star}(t ; v_{0})$ for the models $M_{a,b(t)}$ and $M_{a(t),b}$ are very similar (red and black lines in the respective figures).

In summary, Table \ref{tab:Analysis:ParameterEstimation:Error:Benchmark:RichInvestor} and Figure \ref{fig:Analysis:FittedConsumptionInvestment:Benchmark:RichInvestor} demonstrate that model $M_{a(t),b(t)}$ is the only one among our considered models which provides enough flexibility to model a hump-shaped consumption decision curve besides a linear risky allocation curve. All three benchmark models, which disregard time-dependency of $a(t)$ or $b(t)$ or both, do not lead to a satisfactory fit. In addition, fitting optimal consumption of the four models to the given consumption curve, while ignoring relative investments, shows the same picture. The result is that the sum of the squared distances associated with model $M_{a(t),b(t)}$ is only $21.26 \%$ of the respective sum associated with the second best model $M_{a,b(t)}$. This supports our findings and conclusion that time-varying preference parameters are indeed needed to model the given time-dependent hump-shaped consumption and linear risky allocation in an accurate way.

In addition to the parameter estimation for the expected path, we provide the figures for optimal consumption, risky relative portfolio and wealth process of all four models under two representative scenarios: a mostly upward (see Figure \ref{fig:Analysis:FittedConsumptionInvestment:IncreasingPath:RichInvestor}) and a mostly downward (see Figure \ref{fig:Analysis:FittedConsumptionInvestment:DecreasingPath:RichInvestor}) moving path for the underlying stock. The corresponding expected paths for the consumption rate, the relative risky investment and the wealth process can be found in Figure \ref{fig:Analysis:FittedConsumptionInvestment:Benchmark:RichInvestor}.

In the increasing stock price case optimal consumption and risky relative allocation for model $M_{a(t),b(t)}$ stay very close to the targeted curve since the corresponding wealth stays close to its expected path and shows some reverting behavior. For a stronger increasing underlying price process, consumption exceeds the given consumption curve for the expected path. When the stock price decreases, then optimal consumption and risky allocation for model $M_{a(t),b(t)}$ fall below the target curves after approximately $15$ to $20$ years. In particular higher consumption can no longer be afforded due to a poorly performing equity market. This goes hand in hand with a reduction on the relative risky allocation.

At first glance, it seems that there is a big difference in optimal consumption between our model $M_{a(t),b(t)}$ and the three benchmark models $M_{a,b(t)}$, $M_{a(t),b}$ and $M_{a,b}$ while optimal risky investments and wealth paths for all four models remain in a quite narrow area, although deviation of risky investments from its target curve can be high. This is due to different scales for wealth and consumption. Figure \ref{fig:Analysis:Difference:c:pi:V:exp} visualizes the differences, denoted by $\Delta$, in the fitted consumption and relative risky investment and the corresponding wealth process for the three benchmark models to our model within the expected path situation. It can be observed that relative risky allocation $\hat{\pi}^{\star}(t;v_{0})$ of model $M_{a(t),b(t)}$ exceeds the ones associated with the three benchmark models in the first half of the considered period of $40$ years by up to eight percentage points, and falls below in the second half. Moreover, the difference looks monotone decreasing in age. Furthermore, the wealth process which corresponds to model $M_{a(t),b(t)}$ outperforms the three benchmark models in the first half, but provides a lower wealth in the second half due to a higher consumption rate from approx. year $8$ to $30$, with a certain recovery in the wealth close to retirement.

The two exemplary scenarios and the expected development situation which was used for fitting show that the benchmark models $M_{a,b(t)}$, $M_{a(t),b}$ and $M_{a,b}$ overestimate the given consumption curve in early and older years (close to $t = 0$ and $t = 40$) and underestimate it in between. For our model $M_{a(t),b(t)}$, the optimal consumption rate stays very close to its target curve until consumption cannot be afforded anymore because of a low wealth as result of a strong market decline. We conclude that especially within phases of poor stock performance, both $c^{\star}(t;v_{0})$ and $\hat{\pi}^{\star}(t ; v_{0})$ can deviate a lot from their given curves.

\begin{figure}[h] \center
\subfigure[Fitted consumption rate.]{\includegraphics[width=0.5\textwidth]{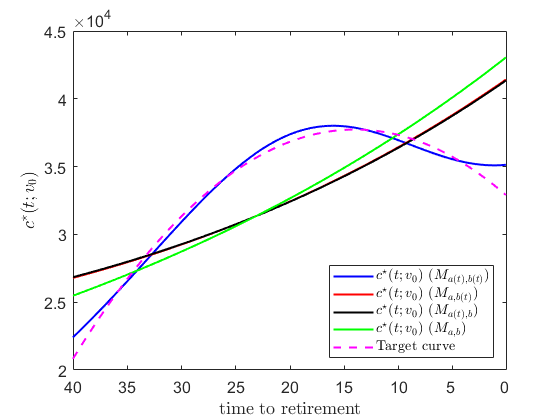}}\hfill
\subfigure[Fitted relative risky investment.]{\includegraphics[width=0.5\textwidth]{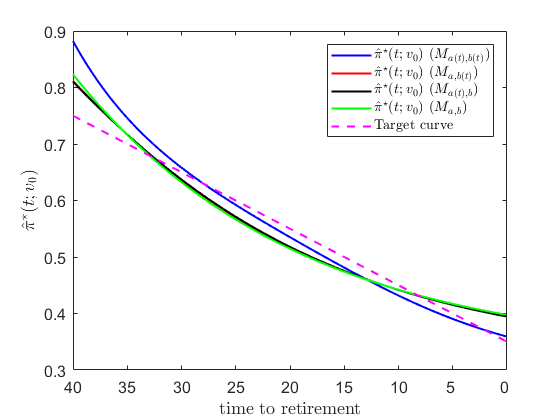}} \\
\subfigure[${\IE[V^{\star}(t;v_{0})]}$.]{\includegraphics[width=0.5\textwidth]{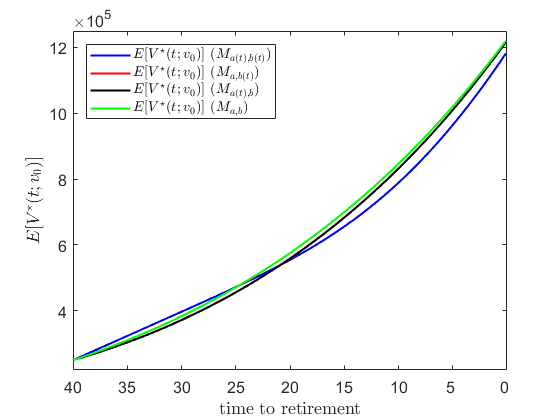}}\hfill
\subfigure[${\IE[P(t)] = p_{1} e^{\mu t}}$.]{\includegraphics[width=0.5\textwidth]{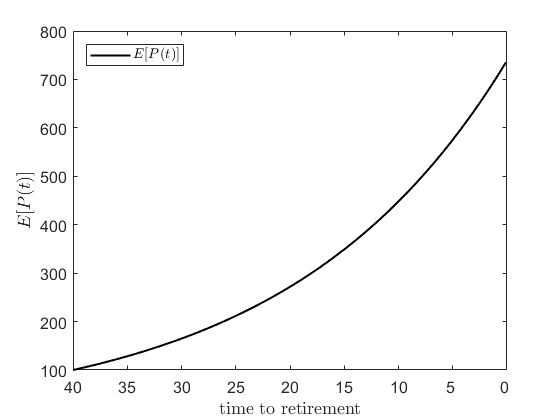}}
\caption{Fitted expected consumption rate $c^{\star}(t;v_{0})$ and relative risky investment $\hat{\pi}^{\star}(t ; v_{0})$, expected wealth process ${\IE[V^{\star}(t;v_{0})]}$ and stock price process ${\IE[P(t)]}$.}
\label{fig:Analysis:FittedConsumptionInvestment:Benchmark:RichInvestor}
\end{figure}

\begin{figure}[t] \center
\subfigure[Optimal consumption rate $c^{\star}(t;v_{0})$.]{\includegraphics[width=0.49\textwidth]{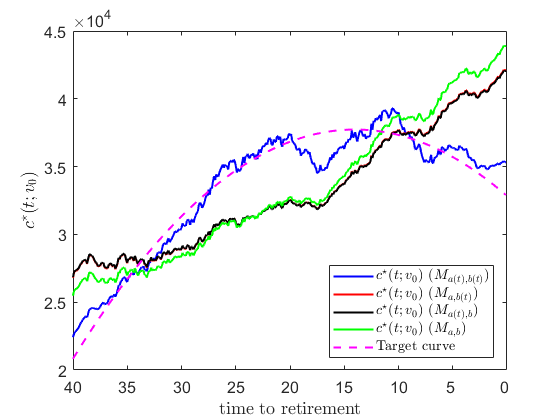}}\hfill
\subfigure[Optimal risky relative investment $\hat{\pi}^{\star}(t;v_{0})$.]{\includegraphics[width=0.49\textwidth]{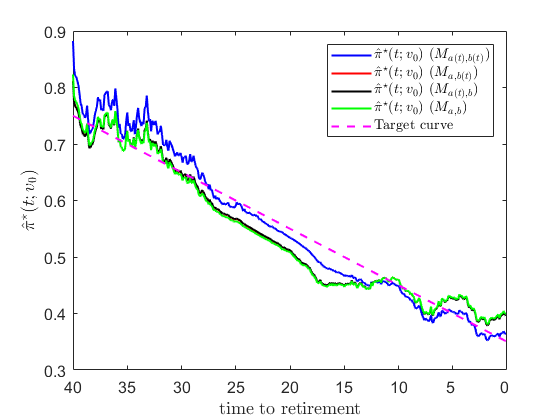}} \\
\subfigure[Optimal wealth $V^{\star}(t;v_{0})$.]{\includegraphics[width=0.49\textwidth]{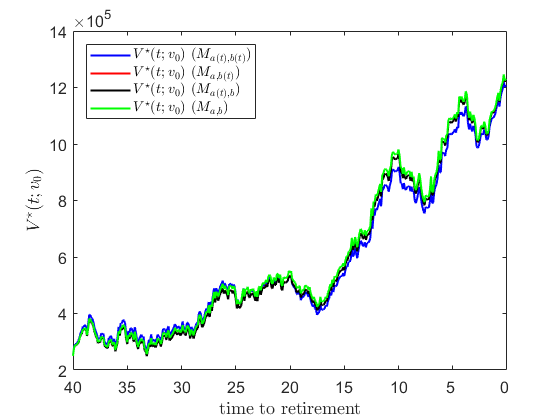}}\hfill
\subfigure[$P(t)$.]{\includegraphics[width=0.49\textwidth]{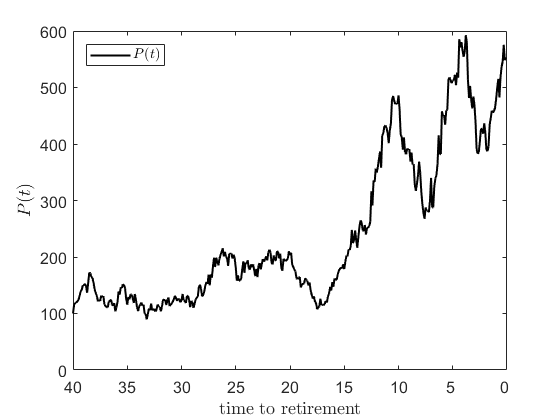}}
\caption{Optimal consumption, risky relative investment strategy and wealth under an increasing risky asset price process.}
\label{fig:Analysis:FittedConsumptionInvestment:IncreasingPath:RichInvestor}
\end{figure}

\begin{figure}[t] \center
\subfigure[Optimal consumption rate $c^{\star}(t;v_{0})$.]{\includegraphics[width=0.49\textwidth]{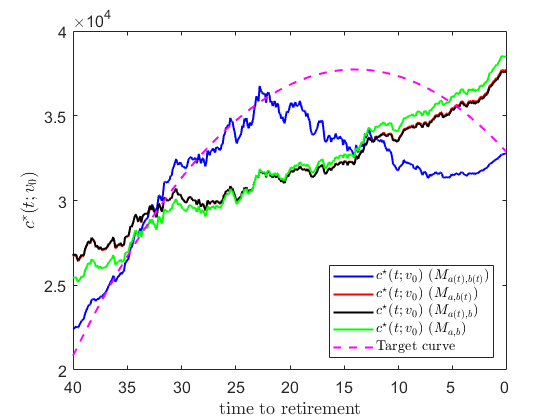}}\hfill
\subfigure[Optimal risky relative investment $\hat{\pi}^{\star}(t;v_{0})$.]{\includegraphics[width=0.49\textwidth]{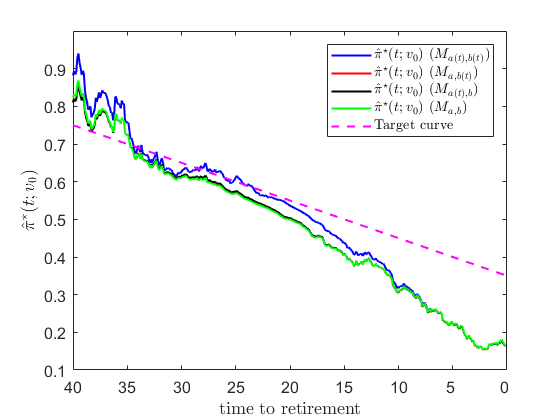}} \\
\subfigure[Optimal wealth $V^{\star}(t;v_{0})$.]{\includegraphics[width=0.49\textwidth]{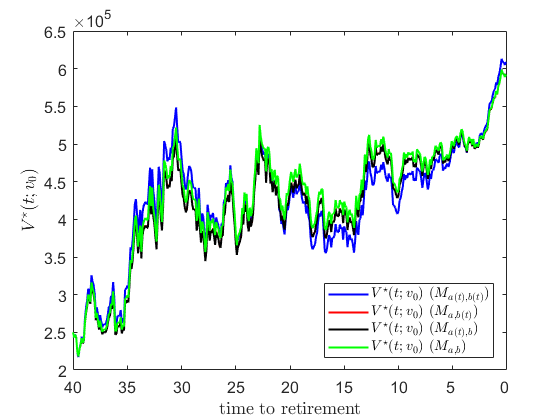}}\hfill
\subfigure[$P(t)$.]{\includegraphics[width=0.49\textwidth]{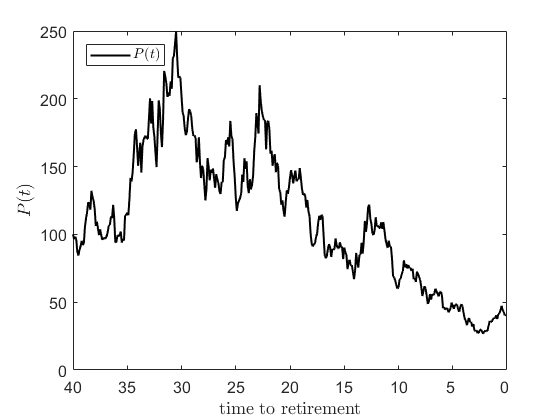}}
\caption{Optimal consumption, risky relative investment strategy and wealth under a decreasing risky asset price process.}
\label{fig:Analysis:FittedConsumptionInvestment:DecreasingPath:RichInvestor}
\end{figure}

\begin{figure}[h] \center
\subfigure[Difference in fitted consumption rate.]{\includegraphics[width=0.5\textwidth]{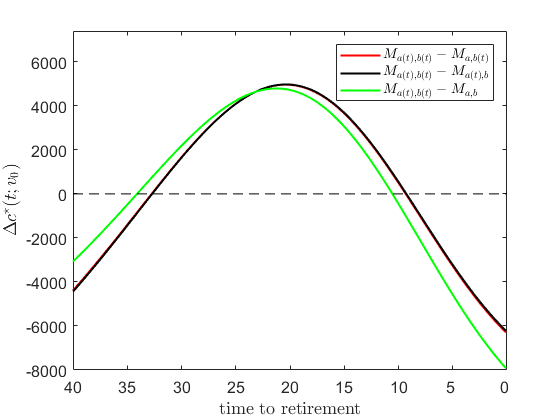}}\hfill
\subfigure[Difference in fitted relative risky investment.]{\includegraphics[width=0.5\textwidth]{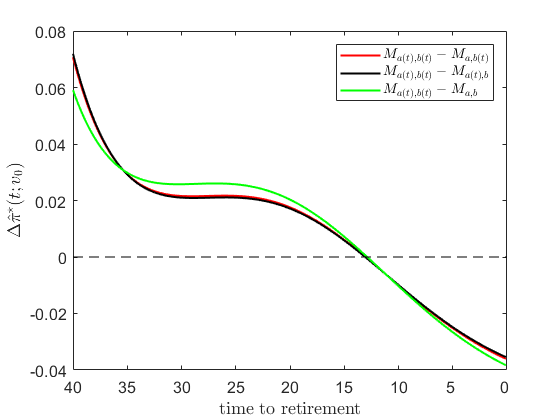}} \\
\subfigure[Difference in expected wealth ${\IE[V^{\star}(t;v_{0})]}$.]{\includegraphics[width=0.5\textwidth]{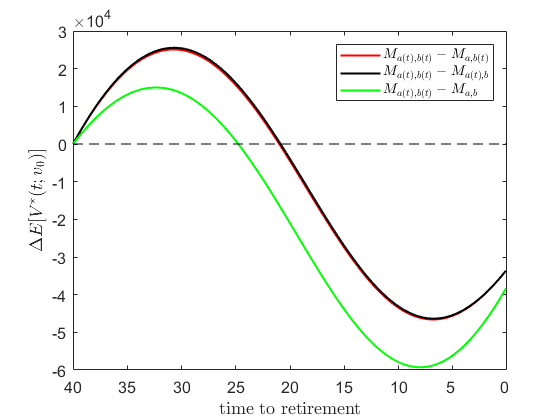}}
\caption{Difference in fitted expected consumption rate, relative risky investment and expected wealth.}
\label{fig:Analysis:Difference:c:pi:V:exp}
\end{figure}

\clearpage

\subsection{Comparison with CRRA}
\label{sec:ComparisonCRRA}
We conclude the case study section by exploring the impact of minimum consumption and wealth floors on calibration and optimal controls. For this sake, we fit the model $M_{a(t),b(t)}$ to the very same parameters and target curves as before, but now enforce $\bar{c}(t) \equiv 0$ and $F \equiv 0$. This CRRA model is referred to as $M_{a(t),b(t)}^{CRRA}$. Table \ref{tab:Analysis:ParameterEstimation:Error:Benchmark:RichInvestor:CRRA} provides the estimated parameters and the sum of the squared relative residuals. In terms of this sum, it is clear that model $M_{a(t),b(t)}$ provides a more adequate fit than model $M_{a(t),b(t)}^{CRRA}$, its sum is only $4.82 \%$ of the sum which corresponds to $M_{a(t),b(t)}^{CRRA}$. Going even further, all three benchmark models $M_{a,b(t)}$, $M_{a(t),b}$ and $M_{a,b}$ from the previous subsection, which all consider minimum levels for consumption and wealth, provide a more precise fit than $M_{a(t),b(t)}^{CRRA}$ in view of the sum of squared relative residuals. This shows that the introduction of floors for consumption and wealth in the model is essential.

\begin{table}[h]
\raggedright
{
    \renewcommand{\arraystretch}{1.75}
    \begin{tabularx}{\textwidth}{| p{1.5cm} || p{5cm}  | X | X | X |}
        \hline
        & Sum of squared relative distances & $\hat{b}$ & $a(t)$ & $b(t)$  \\ \hline\hline
        $M_{a(t),b(t)}$ & $6.0425$ & $-0.9849$ & ${a_{0} = 5.2864 \times 10^{7}}$, ${\lambda_{a} = -0.6673}$ & ${b_{0} = -4.9731}$, ${\lambda_{b} = -0.0340}$ \\ \hline
        $M_{a(t),b(t)}^{CRRA}$ & $125.3497$ & $-4.4867$ & ${a_{0} = 0.6238 \times 10^{7}}$, ${\lambda_{a} = -0.8689}$ & ${b_{0} = -9.7397}$, ${\lambda_{b} = -0.0192}$ \\ \hline
    \end{tabularx}
}
\caption{Estimated parameters and sum of squared relative residuals for CRRA.}
\label{tab:Analysis:ParameterEstimation:Error:Benchmark:RichInvestor:CRRA}
\end{table}

Figure \ref{fig:Analysis:InputFunctions:Benchmark:RichInvestor:CRRA} visualizes the estimated input functions, Figure \ref{fig:Analysis:FittedConsumptionInvestment:Benchmark:RichInvestor:CRRA} provides the graphics about the fitted consumption and relative risky portfolio process with the expected wealth and stock price path. Besides a larger sum of the squared relative distances for model $M_{a(t),b(t)}^{CRRA}$, especially the fitted risky investments $\hat{\pi}^{\star}(t ; v_{0})$ in Figure \ref{fig:Analysis:FittedConsumptionInvestment:Benchmark:RichInvestor:CRRA} show that zero floors for consumption and wealth ($\bar{c}(t) \equiv 0$ and $F \equiv 0$) leads to an imprecise calibration and a large deviation from its given target curve due to a drop in model flexibility. Table \ref{tab:Analysis:ParameterEstimation:Error:Benchmark:RichInvestor:CRRA} suggests that this drop in flexibility is attempted to be compensated by a higher risk aversion in terms of more negative estimated values for $\hat{b}$ and $b(t)$, see also Figure \ref{fig:Analysis:InputFunctions:Benchmark:RichInvestor:CRRA}.

\begin{figure}[h] \center
\subfigure[$a(t)$.]{\includegraphics[width=0.5\textwidth]{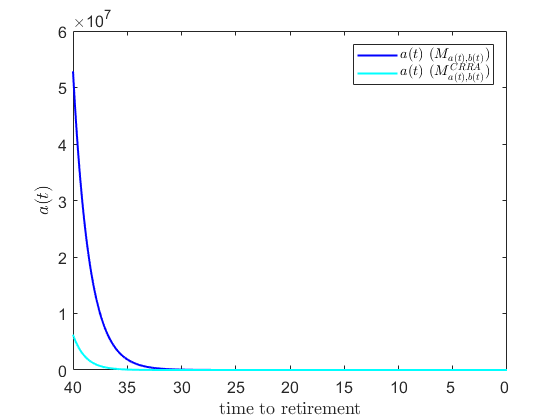}}\hfill
\subfigure[$b(t)$.]{\includegraphics[width=0.5\textwidth]{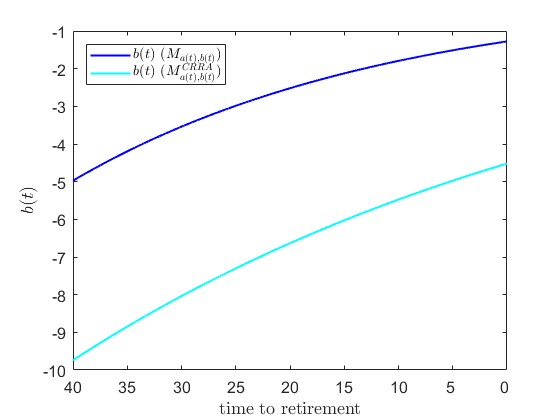}}
\caption{Estimated preference functions $a(t)$ and $b(t)$ for CRRA.}
\label{fig:Analysis:InputFunctions:Benchmark:RichInvestor:CRRA}
\end{figure}

\begin{figure}[h] \center
\subfigure[Fitted consumption rate.]{\includegraphics[width=0.5\textwidth]{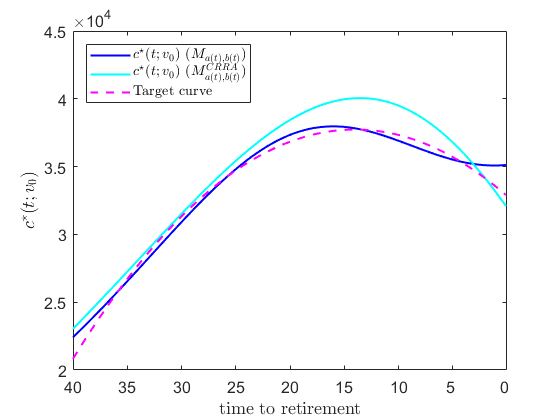}}\hfill
\subfigure[Fitted relative risky investment.]{\includegraphics[width=0.5\textwidth]{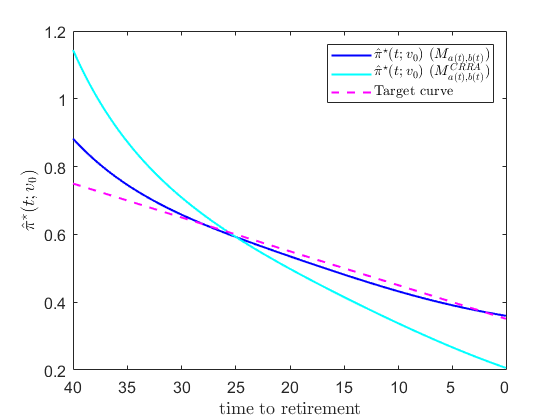}} \\
\subfigure[${\IE[V^{\star}(t;v_{0})]}$.]{\includegraphics[width=0.5\textwidth]{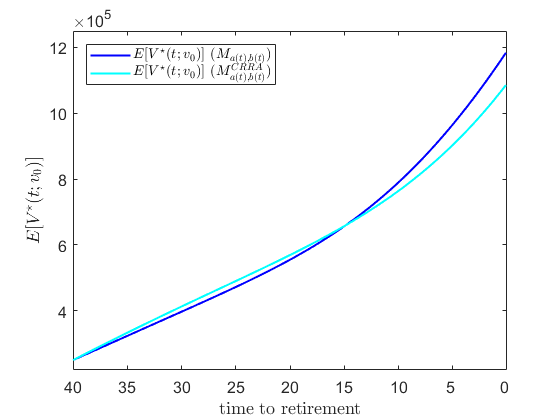}}\hfill
\subfigure[${\IE[P(t)] = p_{1} e^{\mu t}}$.]{\includegraphics[width=0.5\textwidth]{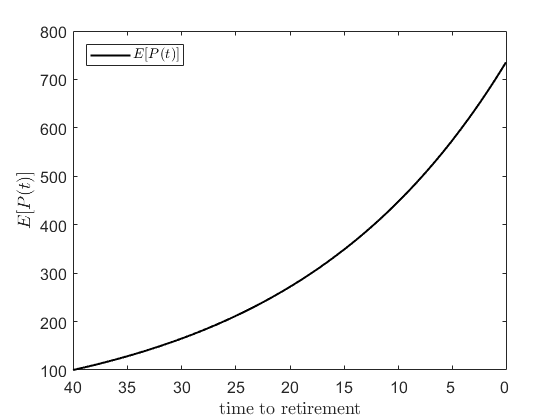}}
\caption{Fitted expected consumption rate $c^{\star}(t;v_{0})$ and relative risky investment $\hat{\pi}^{\star}(t ; v_{0})$, expected wealth process ${\IE[V^{\star}(t;v_{0})]}$ and stock price process ${\IE[P(t)]}$ for CRRA.}
\label{fig:Analysis:FittedConsumptionInvestment:Benchmark:RichInvestor:CRRA}
\end{figure}

\clearpage

\section{Conclusion}
\label{sec:Conclusion}
This paper studies the optimal quantitative and dynamic consumption and investment strategies under age-dependent risk preferences (coefficient of risk aversion $b(t)$ and preference between consumption and terminal wealth $a(t)$). The findings demonstrate that strategies applied for life-cycle pension funds or pension insurance could significantly be improved by taking age-dependent risk preferences into account. For this reason, the paper combines the elements terminal wealth with a minimum level and consumption under time-varying risk preferences and minimum level into a dynamic life-cycle consumption-investment model. A sound economic understanding of the model parts is provided. In Section \ref{sec:SeparationTechnique} the corresponding portfolio optimization problem is solved analytically with a separation approach which allows to solve the consumption and the terminal wealth part of the original consumption-investment problem separately. The formulas show that age-depending risk preferences in combination with terminal wealth considerations and minimum levels for consumption and wealth have a significant impact on the optimal controls.

Section \ref{sec:NCS} investigates the optimal controls and provides a comparison with already existing and solved benchmark models. The analysis is divided into two parts. In the first part the risk preferences are calibrated towards given realistic curves for consumption and investment. The result emphasizes that only our proposed flexible model, in comparison with the other considered benchmark models, provides an adequate fit of the agent's behavior. We draw the conclusion that time-varying preferences (risk aversion $b(t)$ and preference between consumption and terminal wealth $a(t)$) are necessary to provide a sufficient degree of flexibility to accurately fit the two control variables consumption and investment: Our proposed model turns out to be able to explain the given investor consumption and investment decisions, but the benchmark models fail. The very same result is obtained when time-dependent preference functions are considered, but the consumption and wealth floors are omitted. The second part focuses on the behavior analysis of the optimal consumption, investment and wealth under a positive and negative market environment.

Future research on this topic could deal with generalizations of the dynamic life-cycle model. For instance, investment constraints could be included to make the whole setup more applicable as budgets in practice are commonly exposed to constraints on allocation or risk. Furthermore, since unemployment risk and uncertain future income are essential for individuals, those risks and impacts on the optimal controls and wealth process could be further explored. Finally, including mortality and a life insurance product into the model could help people in determining their optimal individual life insurance investment embedded in a more realistic, flexible framework.

\section*{Acknowledgements}
Pavel V. Shevchenko acknowledges the support of Australian Research Council's Discovery Projects funding scheme (project number DP160103489).

\footnotesize

\bibliography{Bibliography}

\normalsize

\clearpage

\appendix

\section{Proofs}
\label{app:Proofs}

\subsection{The consumption problem}
\label{app:ProofsConsumptionProblem}

\begin{proof}[Proof of Theorem \ref{thm:Solution:ConsumptionOnly}]
The Lagrangian of the Problem \eqref{eq:OptimizationProblem:ConsumptionOnly} subject to \eqref{eq:BudgetConstraint:ConsumptionOnly:y} is
\begin{align*}
\mathcal{L}(c,\lambda_{1}) = {} & \IE\left[\int_{0}^{T} U_{1}(t,c(t)) dt\right] - \lambda_{1} \left(\IE\left[\int_{0}^{T} \tilde{Z}(t) \left(c(t) - y(t)\right) dt\right] - v_{1}\right) \\
= {} & \IE\left[\int_{0}^{T} U_{1}(t,c(t)) - \lambda_{1} \left(\tilde{Z}(t) \left(c(t) - y(t)\right) - \frac{1}{T} v_{1}\right) dt\right].
\end{align*}
By the structure of the utility function, the optimal $c_{1}$ fulfills $c_{1}(t;v_{1}) > \bar{c}(t)$ and thus the first order conditions involve existence of a Lagrange multiplier $\lambda_{1} = \lambda_{1}(v_{1}) > 0$ such that the optimal $c_{1}$ maximizes $\mathcal{L}(c,\lambda_{1})$ and such that complementary slackness holds true. Hence it can be shown that the Karush-Kuhn-Tucker conditions besides the first derivative condition are satisfied.

Following \cite{Aase2017}, let $\nabla_{h} \mathcal{L}(c,\lambda_{1};h)$ denote the directional derivative of $\mathcal{L}(c,\lambda_{1})$ in the feasible direction $h$. The directional derivative of a function $f$ in the direction $h$ is generally defined by
\begin{align*}
\nabla_{h} f(x) = \lim_{y \to 0} \frac{f(x + h y) - f(x)}{y}.
\end{align*}
If $f$ is differentiable at $x$ this results in
\begin{align*}
\nabla_{h} f(x) = f^{\prime}(x) h.
\end{align*}
In our case, for the inner function it holds
\begin{align*}
\nabla_{h} \left(U_{1}(t,c(t)) - \lambda_{1} \left(\tilde{Z}(t) \left(c(t) - y(t)\right) - \frac{1}{T} v_{1}\right)\right) = {} & \frac{\partial}{\partial c} \left(U_{1}(t,c(t)) - \lambda_{1} \left(\tilde{Z}(t) \left(c(t) - y(t)\right) - \frac{1}{T} v_{1}\right)\right) h(t) \\
= {} & \left(\frac{\partial}{\partial c} U_{1}(t,c(t)) - \lambda_{1} \tilde{Z}(t)\right) h(t).
\end{align*}
By the dominated convergence theorem, which allows interchanging expectation and differentiation, the first order condition gives
\begin{align*}
0 = {} & \IE\left[\int_{0}^{T} \left(\frac{\partial}{\partial c} U_{1}(t,c(t)) - \lambda_{1} \tilde{Z}(t)\right) h(t) dt\right] \\
= {} & \IE\left[\int_{0}^{T} \left(e^{- \beta  t}  a(t)  \left(\frac{1}{1-b(t)}  \left(c(t) - \bar{c}(t)\right)\right)^{b(t)-1} - \lambda_{1} \tilde{Z}(t)\right) h(t) dt\right]
\end{align*}
for all feasible $h$. In order to fulfill this condition for any $h$, the optimal consumption rate process must be
\begin{align} \label{eq:Solution:c1:ConsumptionOnly}
c_{1}(t;v_{1}) = (1-b(t)) \left(\lambda_{1} \frac{e^{\beta  t}}{a(t)} \tilde{Z}(t)\right)^{\frac{1}{b(t)-1}} + \bar{c}(t),\ t \in [0,T].
\end{align}
Since $U_{1}(t,c)$ strictly increases in $c$, the budget constraint \eqref{eq:BudgetConstraint:ConsumptionOnly:y} for the optimal solution in \eqref{eq:OptimizationProblem:ConsumptionOnly} turns to equality, i.e.
\begin{align*}
\IE\left[\int_{0}^{T} \tilde{Z}(t) \left(c_{1}(t;v_{1}) - y(t)\right) dt\right] = v_{1}.
\end{align*}
When plugging in \eqref{eq:Solution:c1:ConsumptionOnly} and by Fubini, the budget condition turns into
\begin{align*}
v_{1} = {} & \IE\left[\int_{0}^{T} \tilde{Z}(t) \left((1-b(t)) \left(\lambda_{1} \frac{e^{\beta  t}}{a(t)} \tilde{Z}(t)\right)^{\frac{1}{b(t)-1}} + \bar{c}(t) - y(t)\right) dt\right] \\
= {} & \int_{0}^{T} (1-b(t)) \left(\lambda_{1} \frac{e^{\beta  t}}{a(t)}\right)^{\frac{1}{b(t)-1}} \IE\left[\tilde{Z}(t)^{\frac{b(t)}{b(t)-1}}\right] dt + \int_{0}^{T} \IE\left[\tilde{Z}(t)\right] \left(\bar{c}(t) - y(t)\right) dt \\
= {} & \int_{0}^{T} (1-b(t)) \left(\lambda_{1} \frac{e^{\beta  t}}{a(t)}\right)^{\frac{1}{b(t)-1}} e^{- \frac{b(t)}{b(t)-1} \left(r + \frac{1}{2}\|\gamma\|^{2}\right) t + \frac{1}{2} \left(\frac{b(t)}{b(t)-1}\right)^{2} \|\gamma\|^{2} t} dt \\
& + \int_{0}^{T} e^{- \left(r + \frac{1}{2}\|\gamma\|^{2}\right) t + \frac{1}{2} \|\gamma\|^{2} t} \left(\bar{c}(t) - y(t)\right) dt \\
= {} & \int_{0}^{T} (1-b(t)) \left(\frac{e^{\left[\beta - b(t) \left(r - \frac{1}{2} \frac{1}{b(t)-1} \|\gamma\|^{2}\right)\right] t}}{a(t)}\right)^{\frac{1}{b(t)-1}} \lambda_{1}^{\frac{1}{b(t)-1}} dt + \int_{0}^{T} e^{- r t} \left(\bar{c}(t) - y(t)\right) dt \\
= {} & \int_{0}^{T} (1-b(t)) \left(\frac{e^{\left[\beta - b(t) \left(r - \frac{1}{2} \frac{1}{b(t)-1} \|\gamma\|^{2}\right)\right] t}}{a(t)}\right)^{\frac{1}{b(t)-1}} \lambda_{1}^{\frac{1}{b(t)-1}} dt + F_{1}(0).
\end{align*}
Here we used that $\tilde{Z}(t)$ is a log-normal random variable and so is $\tilde{Z}(t)^{\frac{b(t)}{b(t)-1}}$. For any $v_{1} > F_{1}(0) = \int_{0}^{T} e^{- r t} \left(\bar{c}(t) - y(t)\right) dt$, the above equality determines $\lambda_{1} > 0$ uniquely, since the integral in which $\lambda_{1}$ appears strictly decreases in $\lambda_{1}$ and has the limits $0$ and $\infty$ as $\lambda_{1}$ approaches $\infty$ and $0$. It follows immediately that the condition $v_{1} > \int_{0}^{T} e^{- r t} \left(\bar{c}(t) - y(t)\right) dt$ in \eqref{eq:Condition:v1_&_eq:def:F1(t)} is inevitable. The optimal wealth process $V_{1}(t ; v_{1})$ which arises by applying $c_{1}(t ; v_{1})$ is
\begin{align*}
V_{1}(t ; v_{1}) = {} & \IE\left[\int_{t}^{T} \frac{\tilde{Z}(s)}{\tilde{Z}(t)} \left(c_{1}(s;v_{1}) - y(s)\right) ds \Big| \mathcal{F}_{t}\right] \\
= {} & \frac{1}{\tilde{Z}(t)} \IE\left[\int_{t}^{T} \tilde{Z}(s) \left\{(1-b(s)) \left(\lambda_{1} \frac{e^{\beta  s}}{a(s)} \tilde{Z}(s)\right)^{\frac{1}{b(s)-1}} + \bar{c}(s) - y(s)\right\} ds \Big| \mathcal{F}_{t}\right] \\
= {} & \frac{1}{\tilde{Z}(t)} \left\{\IE\left[\int_{t}^{T} (1-b(s)) \left(\lambda_{1} \frac{e^{\beta  s}}{a(s)}\right)^{\frac{1}{b(s)-1}} \tilde{Z}(s)^{\frac{b(s)}{b(s)-1}} ds \Big| \mathcal{F}_{t}\right] + \IE\left[\int_{t}^{T} \tilde{Z}(s) \left(\bar{c}(s) - y(s)\right) ds \Big| \mathcal{F}_{t}\right]\right\} \displaybreak \\
= {} & \frac{1}{\tilde{Z}(t)} \left\{\int_{t}^{T} (1-b(s)) \left(\lambda_{1} \frac{e^{\beta  s}}{a(s)}\right)^{\frac{1}{b(s)-1}} \IE\left[\tilde{Z}(s)^{\frac{b(s)}{b(s)-1}} \Big| \mathcal{F}_{t}\right] ds + \int_{t}^{T} \left(\bar{c}(s) - y(s)\right) \IE\left[\tilde{Z}(s) \Big| \mathcal{F}_{t}\right] ds\right\}.
\end{align*}
$\tilde{Z}(s)$ can be written as $\frac{\tilde{Z}(s)}{\tilde{Z}(t)} \tilde{Z}(t)$ where $\frac{\tilde{Z}(s)}{\tilde{Z}(t)}$ is independent of $\mathcal{F}_{t}$ and $\tilde{Z}(t)$ is $\mathcal{F}_{t}$-measurable. Therefore it follows
\begin{align*}
\IE\left[\tilde{Z}(s) \Big| \mathcal{F}_{t}\right] = {} & \tilde{Z}(t) \IE\left[\frac{\tilde{Z}(s)}{\tilde{Z}(t)}\right] = \tilde{Z}(t) e^{- \left(r + \frac{1}{2}\|\gamma\|^{2}\right) (s-t) + \frac{1}{2} \|\gamma\|^{2} (s-t)} = \tilde{Z}(t) e^{- r (s-t)}, \\
\IE\left[\tilde{Z}(s)^{\eta} \Big| \mathcal{F}_{t}\right] = {} & \tilde{Z}(t)^{\eta} \IE\left[\left(\frac{\tilde{Z}(s)}{\tilde{Z}(t)}\right)^{\eta}\right] = \tilde{Z}(t)^{\eta} e^{- \eta \left(r + \frac{1}{2}\|\gamma\|^{2}\right) (s-t) + \frac{1}{2} \eta^{2} \|\gamma\|^{2} (s-t)} = \tilde{Z}(t)^{\eta} e^{- \eta \left(r - \frac{1}{2} (\eta - 1) \|\gamma\|^{2}\right) (s-t)}
\end{align*}
for any $\eta \in \IR$, where we used that $\frac{\tilde{Z}(s)}{\tilde{Z}(t)}$ and thus $\left(\frac{\tilde{Z}(s)}{\tilde{Z}(t)}\right)^{\eta}$ are log-normally distributed. Define the function $g$ by
\begin{align*}
g(s,t; v_{1}) = (1-b(s)) \left(\frac{e^{\beta  s - b(s) \left(r - \frac{1}{2} \frac{1}{b(s)-1} \|\gamma\|^{2}\right) (s-t)}}{a(s)}\right)^{\frac{1}{b(s)-1}} \lambda_{1}^{\frac{1}{b(s)-1}},
\end{align*}
then the optimal wealth process is given by
\begin{align} \label{eq:V1:ConsumptionProblem}
V_{1}(t ; v_{1}) = \int_{t}^{T} g(s,t; v_{1}) \tilde{Z}(t)^{\frac{1}{b(s)-1}} ds + F_{1}(t)
\end{align}
with $F_{1}(t)$ defined in \eqref{eq:Condition:v1_&_eq:def:F1(t)}. The dynamics can be calculated as
\begin{align*}
d V_{1}(t ; v_{1}) = {} & \left(- g(t,t; v_{1}) \tilde{Z}(t)^{\frac{1}{b(t)-1}} dt + \int_{t}^{T} d_{t} \left(g(s,t; v_{1}) \tilde{Z}(t)^{\frac{1}{b(s)-1}}\right) ds\right) \\
& + \left(- \left(\bar{c}(t) - y(t)\right) dt + \int_{t}^{T} d_{t} \left(e^{- r (s-t)} \left(\bar{c}(s) - y(s)\right)\right) ds\right) \\
= {} & - g(t,t; v_{1}) \tilde{Z}(t)^{\frac{1}{b(t)-1}} dt + \int_{t}^{T} d_{t} \left(g(s,t; v_{1}) \tilde{Z}(t)^{\frac{1}{b(s)-1}}\right) ds - \left(\bar{c}(t) - y(t)\right) dt \\
& + \left(\int_{t}^{T} r e^{- r (s-t)} \left(\bar{c}(s) - y(s)\right) ds\right) dt \\
= {} & \left(- g(t,t; v_{1}) \tilde{Z}(t)^{\frac{1}{b(t)-1}} - \left(\bar{c}(t) - y(t)\right) + \int_{t}^{T} r e^{- r (s-t)} \left(\bar{c}(s) - y(s)\right) ds\right) dt \\
& + \int_{t}^{T} d_{t} \left(g(s,t; v_{1}) \tilde{Z}(t)^{\frac{1}{b(s)-1}}\right) ds.
\end{align*}
Notice that by It\^{o}'s formula,
\begin{align*}
d \left(\tilde{Z}(t)^{\frac{1}{b(s)-1}}\right) = {} & \frac{1}{b(s)-1} \tilde{Z}(t)^{\frac{1}{b(s)-1} - 1} d \tilde{Z}(t) + \frac{1}{2} \frac{1}{b(s)-1} \left(\frac{1}{b(s)-1} - 1\right) \tilde{Z}(t)^{\frac{1}{b(s)-1} - 2} \tilde{Z}(t)^{2} \|\gamma\|^{2} dt \\
= {} & \tilde{Z}(t)^{\frac{1}{b(s)-1}} \left\{\left[- \frac{1}{b(s)-1} r + \frac{1}{2} \frac{1}{b(s)-1} \left(\frac{1}{b(s)-1} - 1\right) \|\gamma\|^{2}\right] dt - \frac{1}{b(s)-1} \gamma' dW(t)\right\}.
\end{align*}
Moreover, it holds
\begin{align*}
d_{t} g(s,t; v_{1}) = {} & (1-b(s)) \left(\frac{e^{\beta  s - b(s) \left(r - \frac{1}{2} \frac{1}{b(s)-1} \|\gamma\|^{2}\right) s}}{a(s)}\right)^{\frac{1}{b(s)-1}} \lambda_{1}^{\frac{1}{b(s)-1}} d_{t} \left(e^{\frac{b(s)}{b(s)-1} \left(r - \frac{1}{2} \frac{1}{b(s)-1} \|\gamma\|^{2}\right) t}\right) \\
= {} & (1-b(s)) \left(\frac{e^{\beta  s - b(s) \left(r - \frac{1}{2} \frac{1}{b(s)-1} \|\gamma\|^{2}\right) s}}{a(s)}\right)^{\frac{1}{b(s)-1}} \lambda_{1}^{\frac{1}{b(s)-1}} \\
& \times \frac{b(s)}{b(s)-1} \left(r - \frac{1}{2} \frac{1}{b(s)-1} \|\gamma\|^{2}\right) e^{\frac{b(s)}{b(s)-1} \left(r - \frac{1}{2} \frac{1}{b(s)-1} \|\gamma\|^{2}\right) t} dt \\
= {} & \frac{b(s)}{b(s)-1} \left(r - \frac{1}{2} \frac{1}{b(s)-1} \|\gamma\|^{2}\right) (1-b(s)) \left(\frac{e^{\beta  s - b(s) \left(r - \frac{1}{2} \frac{1}{b(s)-1} \|\gamma\|^{2}\right) (s-t)}}{a(s)}\right)^{\frac{1}{b(s)-1}} \lambda_{1}^{\frac{1}{b(s)-1}} dt \\
= {} & \frac{b(s)}{b(s)-1} \left(r - \frac{1}{2} \frac{1}{b(s)-1} \|\gamma\|^{2}\right) g(s,t; v_{1}) dt.
\end{align*}
With this we obtain
\begin{align*}
d_{t} \left(g(s,t; v_{1}) \tilde{Z}(t)^{\frac{1}{b(s)-1}}\right) = {} & g(s,t; v_{1}) d \left(\tilde{Z}(t)^{\frac{1}{b(s)-1}}\right) + \tilde{Z}(t)^{\frac{1}{b(s)-1}} d_{t} g(s,t; v_{1}) + 0 \\
= {} & g(s,t; v_{1}) \tilde{Z}(t)^{\frac{1}{b(s)-1}} \\
& \times \left\{\left[- \frac{1}{b(s)-1} r + \frac{1}{2} \frac{1}{b(s)-1} \left(\frac{1}{b(s)-1} - 1\right) \|\gamma\|^{2}\right] dt - \frac{1}{b(s)-1} \gamma' dW(t)\right\} \\
& + \tilde{Z}(t)^{\frac{1}{b(s)-1}} \frac{b(s)}{b(s)-1} \left(r - \frac{1}{2} \frac{1}{b(s)-1} \|\gamma\|^{2}\right) g(s,t; v_{1}) dt \\
= {} & g(s,t; v_{1}) \tilde{Z}(t)^{\frac{1}{b(s)-1}} \left\{\left(r - \frac{1}{b(s)-1} \|\gamma\|^{2}\right) dt - \frac{1}{b(s)-1} \gamma' dW(t)\right\}.
\end{align*}
Define
\begin{align*}
Y(t) = \int_{t}^{T} \frac{1}{b(s)-1} g(s,t; v_{1}) \tilde{Z}(t)^{\frac{1}{b(s)-1}} ds.
\end{align*}
In summary, the dynamics of the optimal wealth process is then given by
\begin{align}
d V_{1}(t ; v_{1}) = {} & \left(- g(t,t; v_{1}) \tilde{Z}(t)^{\frac{1}{b(t)-1}} - \left(\bar{c}(t) - y(t)\right) + \int_{t}^{T} r e^{- r (s-t)} \left(\bar{c}(s) - y(s)\right) ds\right) dt \nonumber \\
& + \int_{t}^{T} g(s,t; v_{1}) \tilde{Z}(t)^{\frac{1}{b(s)-1}} \left\{\left(r - \frac{1}{b(s)-1} \|\gamma\|^{2}\right) dt - \frac{1}{b(s)-1} \gamma' dW(t)\right\} ds \nonumber \\
= {} & \Bigg(- g(t,t; v_{1}) \tilde{Z}(t)^{\frac{1}{b(t)-1}} - \left(\bar{c}(t) - y(t)\right) + \int_{t}^{T} r e^{- r (s-t)} \left(\bar{c}(s) - y(s)\right) ds \nonumber \\
& + \int_{t}^{T} \left(r - \frac{1}{b(s)-1} \|\gamma\|^{2}\right) g(s,t; v_{1}) \tilde{Z}(t)^{\frac{1}{b(s)-1}} ds\Bigg) dt \nonumber \\
& - \underbrace{\left(\int_{t}^{T} \frac{1}{b(s)-1} g(s,t; v_{1}) \tilde{Z}(t)^{\frac{1}{b(s)-1}} ds\right)}_{= Y(t)} \gamma' dW(t) \nonumber \\
= {} & \Bigg\{r \underbrace{\left(\int_{t}^{T} e^{- r (s-t)} \left(\bar{c}(s) - y(s)\right) ds + \int_{t}^{T} g(s,t; v_{1}) \tilde{Z}(t)^{\frac{1}{b(s)-1}} ds\right)}_{= V_{1}(t ; v_{1})} \nonumber \\
& - g(t,t; v_{1}) \tilde{Z}(t)^{\frac{1}{b(t)-1}} - \left(\bar{c}(t) - y(t)\right) - \|\gamma\|^{2} \underbrace{\int_{t}^{T} \frac{1}{b(s)-1} g(s,t; v_{1}) \tilde{Z}(t)^{\frac{1}{b(s)-1}} ds}_{= Y(t)}\Bigg\} dt \nonumber \\
& - Y(t) \gamma' dW(t) \nonumber \\
= {} & \left(r V_{1}(t ; v_{1}) - g(t,t; v_{1}) \tilde{Z}(t)^{\frac{1}{b(t)-1}} - \left(\bar{c}(t) - y(t)\right) - \|\gamma\|^{2} Y(t)\right) dt - Y(t) \gamma' dW(t) \nonumber \\
= {} & \mu_{V_{1}}(t) dt - Y(t) \gamma' dW(t) \label{eq:SDE:OptimalV:ConsumptionOnly}
\end{align}
with drift
\begin{align*}
\mu_{V_{1}}(t) = {} & r V_{1}(t ; v_{1}) - g(t,t; v_{1}) \tilde{Z}(t)^{\frac{1}{b(t)-1}} - \bar{c}(t) + y(t) - \|\gamma\|^{2} Y(t).
\end{align*}
By \eqref{eq:Solution:c1:ConsumptionOnly} it follows
\begin{align*}
c_{1}(t;v_{1}) = (1-b(t)) \left(\lambda_{1} \frac{e^{\beta  t}}{a(t)} \tilde{Z}(t)\right)^{\frac{1}{b(t)-1}} + \bar{c}(t) = g(t,t; v_{1}) \tilde{Z}(t)^{\frac{1}{b(t)-1}} + \bar{c}(t).
\end{align*}
Hence
\begin{align*}
\mu_{V_{1}}(t) = {} & r V_{1}(t ; v_{1}) - c_{1}(t;v_{1}) + y(t) - \|\gamma\|^{2} Y(t).
\end{align*}
In order to determine the optimal investment strategy $\pi_{1}(t ; v_{1})$ to Problem \eqref{eq:OptimizationProblem:ConsumptionOnly} we compare the optimal wealth dynamics in \eqref{eq:SDE:V:y} and \eqref{eq:SDE:OptimalV:ConsumptionOnly}:
\begin{align*}
d V_{1}(t ; v_{1}) = {} & V_{1}(t ; v_{1})  \left[\left(r + \hat{\pi}_{1}(t ; v_{1})' \left(\mu - r \mathbf{1}\right)\right) dt + \hat{\pi}_{1}(t ; v_{1})'\sigma dW(t)\right] - c_{1}(t;v_{1}) dt + y(t) dt, \\
d V_{1}(t ; v_{1}) = {} & \left(r V_{1}(t ; v_{1}) - c_{1}(t;v_{1}) + y(t) - \|\gamma\|^{2} Y(t)\right) dt - Y(t) \gamma' dW(t).
\end{align*}
Matching the diffusion terms yields the equality
\begin{align*}
\hat{\pi}_{1}(t ; v_{1}) = - \frac{Y(t)}{V_{1}(t ; v_{1})} \Sigma^{-1} (\mu - r \mathbf{1})
\end{align*}
which simultaneously matches the drift terms. By the first mean value theorem for integrals\footnote{For two integrable functions $f(x)$ and $g(x)$ on the interval $(a,b)$, where $f(x)$ is continuous and $g(x)$ does not change sign on $(a,b)$, there exists $d \in (a,b)$ such that
\begin{align*}
\int_{a}^{b} f(x) g(x) dx = f(d) \int_{a}^{b} g(x) dx.
\end{align*}} it furthermore follows that there exists $\tilde{t}_{1} \in (t,T)$ such that
\begin{align*}
Y(t) = {} & \int_{t}^{T} \frac{1}{b(s)-1} g(s,t; v_{1}) \tilde{Z}(t)^{\frac{1}{b(s)-1}} ds = \frac{1}{b(\tilde{t}_{1})-1} \int_{t}^{T} g(s,t; v_{1}) \tilde{Z}(t)^{\frac{1}{b(s)-1}} ds \\
\stackrel{\eqref{eq:V1:ConsumptionProblem}}{=} {} & \frac{1}{b(\tilde{t}_{1})-1} \left(V_{1}(t ; v_{1}) - F_{1}(t)\right).
\end{align*}
This determines the optimal investment strategy to be
\begin{align}
\hat{\pi}_{1}(t ; v_{1}) = \frac{1}{1 - b(\tilde{t}_{1})} \Sigma^{-1} (\mu - r \mathbf{1}) \frac{V_{1}(t ; v_{1}) - F_{1}(t)}{V_{1}(t ; v_{1})}.
\end{align}
\end{proof}

\begin{proof}[Proof of Theorem \ref{thm:Solution:ConsumptionOnly:ValueFunction:lambda}]
Firstly, the value function of this problem is
\begin{align*}
\mathcal{V}_{1}(v_{1}) = {} & \IE\left[\int_{0}^{T} U_{1}(t,c_{1}(t;v_{1})) dt\right] = \IE\left[\int_{0}^{T} e^{- \beta  t}  \frac{1-b(t)}{b(t)}  a(t)  \left(\frac{1}{1-b(t)}  \left(c_{1}(t;v_{1}) - \bar{c}(t)\right)\right)^{b(t)} dt\right] \\
= {} & \IE\left[\int_{0}^{T} e^{- \beta  t}  \frac{1-b(t)}{b(t)}  a(t)  \left(\lambda_{1} \frac{e^{\beta  t}}{a(t)} \tilde{Z}(t)\right)^{\frac{b(t)}{b(t)-1}} dt\right] \\
= {} & \int_{0}^{T} e^{- \beta  t}  \frac{1-b(t)}{b(t)}  a(t)  \left(\lambda_{1} \frac{e^{\beta  t}}{a(t)}\right)^{\frac{b(t)}{b(t)-1}} \IE\left[\tilde{Z}(t)^{\frac{b(t)}{b(t)-1}}\right] dt \\
= {} & \int_{0}^{T} \frac{1-b(t)}{b(t)} \left(\frac{e^{\beta  t}}{a(t)}\right)^{\frac{1}{b(t)-1}} \lambda_{1}^{\frac{b(t)}{b(t)-1}} \IE\left[\tilde{Z}(t)^{\frac{b(t)}{b(t)-1}}\right] dt \\
= {} & \int_{0}^{T} \frac{1-b(t)}{b(t)} \left(\frac{e^{\beta  t}}{a(t)}\right)^{\frac{1}{b(t)-1}} \lambda_{1}^{\frac{b(t)}{b(t)-1}} e^{- \frac{b(t)}{b(t)-1} \left(r + \frac{1}{2}\|\gamma\|^{2}\right) t + \frac{1}{2} \left(\frac{b(t)}{b(t)-1}\right)^{2} \|\gamma\|^{2} t} dt \displaybreak \\
= {} & \int_{0}^{T} \frac{1-b(t)}{b(t)} \left(\frac{e^{\left[\beta - b(t) \left(r - \frac{1}{2} \frac{1}{b(t)-1} \|\gamma\|^{2}\right)\right] t}}{a(t)}\right)^{\frac{1}{b(t)-1}} \lambda_{1}^{\frac{b(t)}{b(t)-1}} dt,
\end{align*}
where $\lambda_{1}$ is subject to \eqref{eq:ConsumptionOnly:lambda}. From differentiating both sides of Equation \eqref{eq:ConsumptionOnly:lambda} with respect to $v_{1}$ we derive
\begin{align}
1 = {} & \frac{\partial}{\partial v_{1}} \int_{0}^{T} (1-b(t)) \left(\frac{e^{\left[\beta - b(t) \left(r - \frac{1}{2} \frac{1}{b(t)-1} \|\gamma\|^{2}\right)\right] t}}{a(t)}\right)^{\frac{1}{b(t)-1}} \lambda_{1}^{\frac{1}{b(t)-1}} dt \nonumber \\
= {} & \int_{0}^{T} (1-b(t)) \left(\frac{e^{\left[\beta - b(t) \left(r - \frac{1}{2} \frac{1}{b(t)-1} \|\gamma\|^{2}\right)\right] t}}{a(t)}\right)^{\frac{1}{b(t)-1}} \frac{\partial}{\partial v_{1}} \left(\lambda_{1}^{\frac{1}{b(t)-1}}\right) dt. \label{eq:Lagrange:ConsumptionOnly:helpderivative}
\end{align}
This helps to identify $\mathcal{V}_{1}^{\prime}(v_{1})$ to be
\begin{align*}
\mathcal{V}_{1}^{\prime}(v_{1}) = {} & \frac{\partial}{\partial v_{1}} \int_{0}^{T} \frac{1-b(t)}{b(t)} \left(\frac{e^{\left[\beta - b(t) \left(r - \frac{1}{2} \frac{1}{b(t)-1} \|\gamma\|^{2}\right)\right] t}}{a(t)}\right)^{\frac{1}{b(t)-1}} \lambda_{1}^{\frac{b(t)}{b(t)-1}} dt \\
= {} & \int_{0}^{T} \frac{1-b(t)}{b(t)} \left(\frac{e^{\left[\beta - b(t) \left(r - \frac{1}{2} \frac{1}{b(t)-1} \|\gamma\|^{2}\right)\right] t}}{a(t)}\right)^{\frac{1}{b(t)-1}} \frac{\partial}{\partial v_{1}} \left(\lambda_{1}^{\frac{b(t)}{b(t)-1}}\right) dt \\
= {} & \int_{0}^{T} \frac{1-b(t)}{b(t)} \left(\frac{e^{\left[\beta - b(t) \left(r - \frac{1}{2} \frac{1}{b(t)-1} \|\gamma\|^{2}\right)\right] t}}{a(t)}\right)^{\frac{1}{b(t)-1}} \frac{\partial}{\partial v_{1}} \left(\left(\lambda_{1}^{\frac{1}{b(t)-1}}\right)^{b(t)}\right) dt \\
= {} & \int_{0}^{T} \frac{1-b(t)}{b(t)} \left(\frac{e^{\left[\beta - b(t) \left(r - \frac{1}{2} \frac{1}{b(t)-1} \|\gamma\|^{2}\right)\right] t}}{a(t)}\right)^{\frac{1}{b(t)-1}} b(t) \left(\lambda_{1}^{\frac{1}{b(t)-1}}\right)^{b(t)-1} \frac{\partial}{\partial v_{1}} \left(\lambda_{1}^{\frac{1}{b(t)-1}}\right) dt \\
= {} & \lambda_{1} \int_{0}^{T} (1-b(t)) \left(\frac{e^{\left[\beta - b(t) \left(r - \frac{1}{2} \frac{1}{b(t)-1} \|\gamma\|^{2}\right)\right] t}}{a(t)}\right)^{\frac{1}{b(t)-1}} \frac{\partial}{\partial v_{1}} \left(\lambda_{1}^{\frac{1}{b(t)-1}}\right) dt \stackrel{\eqref{eq:Lagrange:ConsumptionOnly:helpderivative}}{=} \lambda_{1}.
\end{align*}
\eqref{eq:Lagrange:ConsumptionOnly:helpderivative} further implies concavity of $\mathcal{V}_{1}(v_{1})$ as
\begin{align*}
1 = {} & \int_{0}^{T} (1-b(t)) \left(\frac{e^{\left[\beta - b(t) \left(r - \frac{1}{2} \frac{1}{b(t)-1} \|\gamma\|^{2}\right)\right] t}}{a(t)}\right)^{\frac{1}{b(t)-1}} \frac{\partial}{\partial v_{1}} \left(\lambda_{1}^{\frac{1}{b(t)-1}}\right) dt \\
= {} & \int_{0}^{T} (1-b(t)) \left(\frac{e^{\left[\beta - b(t) \left(r - \frac{1}{2} \frac{1}{b(t)-1} \|\gamma\|^{2}\right)\right] t}}{a(t)}\right)^{\frac{1}{b(t)-1}} \frac{1}{b(t)-1} \lambda_{1}^{\frac{1}{b(t)-1} - 1} \lambda_{1}^{\prime} dt \\
= {} & - \lambda_{1}^{\prime} \int_{0}^{T} \left(\frac{e^{\left[\beta - b(t) \left(r - \frac{1}{2} \frac{1}{b(t)-1} \|\gamma\|^{2}\right)\right] t}}{a(t)}\right)^{\frac{1}{b(t)-1}} \lambda_{1}^{- \frac{b(t)-2}{b(t)-1}}  dt
\end{align*}
and thus
\begin{align*}
\mathcal{V}_{1}^{\prime\prime}(v_{1}) = \lambda_{1}^{\prime} = - \left(\int_{0}^{T} \left(\frac{e^{\left[\beta - b(t) \left(r - \frac{1}{2} \frac{1}{b(t)-1} \|\gamma\|^{2}\right)\right] t}}{a(t)}\right)^{\frac{1}{b(t)-1}} \lambda_{1}^{- \frac{b(t)-2}{b(t)-1}}  dt\right)^{-1} < 0.
\end{align*}
\end{proof}

\subsection{The terminal wealth problem}
\label{app:ProofsTerminalWealthProblem}

\begin{proof}[Proof of Theorem \ref{thm:Solution:TerminalWealthOnly:WithoutProbabilityConstraint}]
The Lagrangian of the Problem \eqref{eq:OptimizationProblem:TerminalWealthOnly} subject to \eqref{eq:BudgetConstraint:TerminalWealthOnly} is
\begin{align*}
\mathcal{L}(V,\lambda_{2}) = {} & \IE\left[U_{2}(V)\right] - \lambda_{2} \left(\IE\left[\tilde{Z}(T) V\right] - v_{2}\right) = \IE\left[U_{2}(V) - \lambda_{2} \left(\tilde{Z}(T) V - v_{2}\right)\right].
\end{align*}
First of all, it is clear that ${c_{2}(t ; v_{2}) \equiv 0}$. By the structure of the utility function, the optimal $V_{2}$ fulfills $V_{2}(T;v_{2}) > F$ and thus the first order conditions involve existence of a Lagrange multiplier $\lambda_{2} = \lambda_{2}(v_{2}) > 0$ such that the optimal $V_{2}$ maximizes $\mathcal{L}(V,\lambda_{2})$ and such that complementary slackness holds true. Hence it can be shown that the Karush-Kuhn-Tucker conditions besides the first derivative condition are satisfied. By the dominated convergence theorem, the first order condition with respect to the directional derivative gives
\begin{align*}
0 = {} & \IE\left[\left(\frac{\partial}{\partial V} U_{2}(V) - \lambda_{2} \tilde{Z}(T)\right) h\right] = \IE\left[\left(e^{- \beta  T}  \hat{a} \left(\frac{1}{1-\hat{b}}  (V-F)\right)^{\hat{b}-1} - \lambda_{2} \tilde{Z}(T)\right) h\right],
\end{align*}
which has to be satisfied for all suitable $h$; hence the optimal terminal wealth has to fulfill
\begin{align} \label{eq:Solution:V2T:TerminalWealthOnly}
V_{2}(T ; v_{2}) = (1-\hat{b}) \left(\lambda_{2} \frac{e^{\beta  T}}{\hat{a}} \tilde{Z}(T)\right)^{\frac{1}{\hat{b}-1}} + F.
\end{align}
Since $U_{2}(V)$ strictly increases in $V$, complementary slackness implies equality for the budget constraint
\begin{align*}
\IE\left[\tilde{Z}(T) V_{2}(T ; v_{2})\right] = v_{2}.
\end{align*}
Using \eqref{eq:Solution:V2T:TerminalWealthOnly} and Fubini this gives
\begin{align*}
v_{2} = {} & \IE\left[\tilde{Z}(T) \left((1-\hat{b}) \left(\lambda_{2} \frac{e^{\beta  T}}{\hat{a}} \tilde{Z}(T)\right)^{\frac{1}{\hat{b}-1}} + F\right)\right] = (1-\hat{b}) \left(\lambda_{2} \frac{e^{\beta  T}}{\hat{a}}\right)^{\frac{1}{\hat{b}-1}} \IE\left[\tilde{Z}(T)^{\frac{\hat{b}}{\hat{b}-1}}\right] + F \IE\left[\tilde{Z}(T)\right] \\
= {} & (1-\hat{b}) \left(\lambda_{2} \frac{e^{\beta  T}}{\hat{a}}\right)^{\frac{1}{\hat{b}-1}} e^{- \frac{\hat{b}}{\hat{b}-1} \left(r + \frac{1}{2}\|\gamma\|^{2}\right) T + \frac{1}{2} \left(\frac{\hat{b}}{\hat{b}-1}\right)^{2} \|\gamma\|^{2} T} + F e^{- \left(r + \frac{1}{2}\|\gamma\|^{2}\right) T + \frac{1}{2} \|\gamma\|^{2} T} \\
= {} & (1-\hat{b}) \left(\frac{e^{\left[\beta - \hat{b} \left(r - \frac{1}{2} \frac{1}{\hat{b}-1} \|\gamma\|^{2}\right)\right] T}}{\hat{a}}\right)^{\frac{1}{\hat{b}-1}}  \lambda_{2}^{\frac{1}{\hat{b}-1}} + e^{- r T} F \displaybreak \\
= {} & (1-\hat{b}) \left(\frac{e^{\left[\beta - \hat{b} \left(r - \frac{1}{2} \frac{1}{\hat{b}-1} \|\gamma\|^{2}\right)\right] T}}{\hat{a}}\right)^{\frac{1}{\hat{b}-1}}  \lambda_{2}^{\frac{1}{\hat{b}-1}} + F_{2}(0).
\end{align*}
Solving for $\lambda_{2}$ yields
\begin{align} \label{eq:Lagrange:TerminalWealthOnly}
\lambda_{2} = \left(\frac{v_{2} - F_{2}(0)}{(1-\hat{b}) \left(\frac{e^{\left[\beta - \hat{b} \left(r - \frac{1}{2} \frac{1}{\hat{b}-1} \|\gamma\|^{2}\right)\right] T}}{\hat{a}}\right)^{\frac{1}{\hat{b}-1}}}\right)^{\hat{b}-1} = e^{-\left[\beta - \hat{b} \left(r - \frac{1}{2} \frac{1}{\hat{b}-1} \|\gamma\|^{2}\right)\right] T}  \left(1-\hat{b}\right)^{1-\hat{b}}  \hat{a} \left(v_{2} - F_{2}(0)\right)^{\hat{b}-1}
\end{align}
where ${v_{2} > F_{2}(0) = e^{- r T} F}$ in \eqref{eq:Condition:v2:WithoutProbabilityConstraint_&_eq:def:F2(t)} is required. Plugging this back into \eqref{eq:Solution:V2T:TerminalWealthOnly}, the optimal terminal wealth is
\begin{align}
V_{2}(T ; v_{2}) = {} & (1-\hat{b}) \left(\frac{e^{\beta  T}}{\hat{a}} \tilde{Z}(T)\right)^{\frac{1}{\hat{b}-1}} \lambda_{2}^{\frac{1}{\hat{b}-1}} + F \nonumber \\
= {} & (1-\hat{b}) \left(\frac{e^{\beta  T}}{\hat{a}} \tilde{Z}(T)\right)^{\frac{1}{\hat{b}-1}} \left(\frac{v_{2} - F_{2}(0)}{(1-\hat{b}) \left(\frac{e^{\left[\beta - \hat{b} \left(r - \frac{1}{2} \frac{1}{\hat{b}-1} \|\gamma\|^{2}\right)\right] T}}{\hat{a}}\right)^{\frac{1}{\hat{b}-1}}}\right) + F \nonumber \\
= {} & \left(v_{2} - F_{2}(0)\right) \left(e^{\hat{b} \left(r - \frac{1}{2} \frac{1}{\hat{b}-1} \|\gamma\|^{2}\right) T} \tilde{Z}(T)\right)^{\frac{1}{\hat{b}-1}} + F \nonumber \\
= {} & \left(v_{2} - F_{2}(0)\right) e^{\frac{\hat{b}}{\hat{b}-1} \left(r - \frac{1}{2} \frac{1}{\hat{b}-1} \|\gamma\|^{2}\right) T} \tilde{Z}(T)^{\frac{1}{\hat{b}-1}} + F. \label{eq:Solution:V2T:TerminalWealthOnly:lambdaInserted}
\end{align}
The optimal wealth process replicates $V_{2}(T ; v_{2})$ and is uniquely given by
\begin{align*}
V_{2}(t ; v_{2}) = {} & \IE\left[\frac{\tilde{Z}(T)}{\tilde{Z}(t)} V_{2}(T ; v_{2}) \Big| \mathcal{F}_{t}\right] = \frac{1}{\tilde{Z}(t)} \IE\left[\tilde{Z}(T) \left\{\left(v_{2} - F_{2}(0)\right) e^{\frac{\hat{b}}{\hat{b}-1} \left(r - \frac{1}{2} \frac{1}{\hat{b}-1} \|\gamma\|^{2}\right) T} \tilde{Z}(T)^{\frac{1}{\hat{b}-1}} + F\right\} \Big| \mathcal{F}_{t}\right] \\
= {} & \frac{1}{\tilde{Z}(t)} \left\{\left(v_{2} - F_{2}(0)\right) e^{\frac{\hat{b}}{\hat{b}-1} \left(r - \frac{1}{2} \frac{1}{\hat{b}-1} \|\gamma\|^{2}\right) T} \IE\left[\tilde{Z}(T)^{\frac{\hat{b}}{\hat{b}-1}} \Big| \mathcal{F}_{t}\right] + F \IE\left[\tilde{Z}(T) \Big| \mathcal{F}_{t}\right]\right\} \\
= {} & \frac{1}{\tilde{Z}(t)} \left\{\left(v_{2} - F_{2}(0)\right) e^{\frac{\hat{b}}{\hat{b}-1} \left(r - \frac{1}{2} \frac{1}{\hat{b}-1} \|\gamma\|^{2}\right) T} \tilde{Z}(t)^{\frac{\hat{b}}{\hat{b}-1}} e^{- \frac{\hat{b}}{\hat{b}-1} \left(r - \frac{1}{2} \frac{1}{\hat{b}-1} \|\gamma\|^{2}\right) (T-t)} + F \tilde{Z}(t) e^{- r (T-t)}\right\}.
\end{align*}
This finally gives
\begin{align} \label{eq:V2:TerminalWealthProblem}
V_{2}(t ; v_{2}) = \left(v_{2} - F_{2}(0)\right) e^{\frac{\hat{b}}{\hat{b}-1} \left(r - \frac{1}{2} \frac{1}{\hat{b}-1} \|\gamma\|^{2}\right) t} \tilde{Z}(t)^{\frac{1}{\hat{b}-1}} + F_{2}(t)
\end{align}
with $F_{2}(t)$ defined in \eqref{eq:Condition:v2:WithoutProbabilityConstraint_&_eq:def:F2(t)}. Recall that 
\begin{align*}
d \left(\tilde{Z}(t)^{\frac{1}{\hat{b}-1}}\right) = \tilde{Z}(t)^{\frac{1}{\hat{b}-1}} \left\{\left[- \frac{1}{\hat{b}-1} r + \frac{1}{2} \frac{1}{\hat{b}-1} \left(\frac{1}{\hat{b}-1} - 1\right) \|\gamma\|^{2}\right] dt - \frac{1}{\hat{b}-1} \gamma' dW(t)\right\}.
\end{align*}
It follows by It\^{o}
\begin{align*}
d & \left(e^{\frac{\hat{b}}{\hat{b}-1} \left(r - \frac{1}{2} \frac{1}{\hat{b}-1} \|\gamma\|^{2}\right) t} \tilde{Z}(t)^{\frac{1}{\hat{b}-1}}\right) \\
& = e^{\frac{\hat{b}}{\hat{b}-1} \left(r - \frac{1}{2} \frac{1}{\hat{b}-1} \|\gamma\|^{2}\right) t} d \left(\tilde{Z}(t)^{\frac{1}{\hat{b}-1}}\right) + \tilde{Z}(t)^{\frac{1}{\hat{b}-1}} d \left(e^{\frac{\hat{b}}{\hat{b}-1} \left(r - \frac{1}{2} \frac{1}{\hat{b}-1} \|\gamma\|^{2}\right) t}\right) + 0 \\
& = e^{\frac{\hat{b}}{\hat{b}-1} \left(r - \frac{1}{2} \frac{1}{\hat{b}-1} \|\gamma\|^{2}\right) t} \tilde{Z}(t)^{\frac{1}{\hat{b}-1}} \left\{\left[- \frac{1}{\hat{b}-1} r + \frac{1}{2} \frac{1}{\hat{b}-1} \left(\frac{1}{\hat{b}-1} - 1\right) \|\gamma\|^{2}\right] dt - \frac{1}{\hat{b}-1} \gamma' dW(t)\right\} \\
& \quad + \tilde{Z}(t)^{\frac{1}{\hat{b}-1}} \frac{\hat{b}}{\hat{b}-1} \left(r - \frac{1}{2} \frac{1}{\hat{b}-1} \|\gamma\|^{2}\right) e^{\frac{\hat{b}}{\hat{b}-1} \left(r - \frac{1}{2} \frac{1}{\hat{b}-1} \|\gamma\|^{2}\right) t} dt \\
& = e^{\frac{\hat{b}}{\hat{b}-1} \left(r - \frac{1}{2} \frac{1}{\hat{b}-1} \|\gamma\|^{2}\right) t} \tilde{Z}(t)^{\frac{1}{\hat{b}-1}} \\
& \quad \times \left\{\left[- \frac{1}{\hat{b}-1} r + \frac{1}{2} \frac{1}{\hat{b}-1} \left(\frac{1}{\hat{b}-1} - 1\right) \|\gamma\|^{2}\right] dt - \frac{1}{\hat{b}-1} \gamma' dW(t) + \frac{\hat{b}}{\hat{b}-1} \left(r - \frac{1}{2} \frac{1}{\hat{b}-1} \|\gamma\|^{2}\right) dt\right\} \\
& = e^{\frac{\hat{b}}{\hat{b}-1} \left(r - \frac{1}{2} \frac{1}{\hat{b}-1} \|\gamma\|^{2}\right) t} \tilde{Z}(t)^{\frac{1}{\hat{b}-1}}  \left\{\frac{1}{\hat{b}-1}\left[(\hat{b}-1) r + \frac{1}{2} \left(\frac{1}{\hat{b}-1} - 1 - \frac{\hat{b}}{\hat{b}-1}\right) \|\gamma\|^{2}\right] dt - \frac{1}{\hat{b}-1} \gamma' dW(t)\right\} \\
& = e^{\frac{\hat{b}}{\hat{b}-1} \left(r - \frac{1}{2} \frac{1}{\hat{b}-1} \|\gamma\|^{2}\right) t} \tilde{Z}(t)^{\frac{1}{\hat{b}-1}} \left\{\left(r - \frac{1}{\hat{b}-1} \|\gamma\|^{2}\right) dt - \frac{1}{\hat{b}-1} \gamma' dW(t)\right\}.
\end{align*}

Then the optimal wealth dynamics can be calculated as
\begin{align*}
d V_{2}(t ; v_{2}) = {} & \left(v_{2} - F_{2}(0)\right) d \left(e^{\frac{\hat{b}}{\hat{b}-1} \left(r - \frac{1}{2} \frac{1}{\hat{b}-1} \|\gamma\|^{2}\right) t} \tilde{Z}(t)^{\frac{1}{\hat{b}-1}}\right) + r F_{2}(t) dt \\
= {} & \underbrace{\left(v_{2} - F_{2}(0)\right) e^{\frac{\hat{b}}{\hat{b}-1} \left(r - \frac{1}{2} \frac{1}{\hat{b}-1} \|\gamma\|^{2}\right) t} \tilde{Z}(t)^{\frac{1}{\hat{b}-1}}}_{\stackrel{\eqref{eq:V2:TerminalWealthProblem}}{=} V_{2}(t ; v_{2}) - F_{2}(t)} \left\{\left(r - \frac{1}{\hat{b}-1} \|\gamma\|^{2}\right) dt - \frac{1}{\hat{b}-1} \gamma' dW(t)\right\} + r F_{2}(t) dt \\
= {} & r V_{2}(t ; v_{2}) dt + \left(V_{2}(t ; v_{2}) - F_{2}(t)\right) \left\{- \frac{1}{\hat{b}-1} \|\gamma\|^{2} dt - \frac{1}{\hat{b}-1} \gamma' dW(t)\right\}.
\end{align*}
Comparing the diffusion term with the one from \eqref{eq:SDE:V:y} for $y(t) \equiv 0$ implies
\begin{align}
\hat{\pi}_{2}(t ; v_{2}) = \frac{1}{1-\hat{b}} \Sigma^{-1}  (\mu - r \mathbf{1}) \frac{V_{2}(t ; v_{2}) - F_{2}(t)}{V_{2}(t ; v_{2})}
\end{align}
which automatically matches the drifts iff ${c_{2}(t ; v_{2}) \equiv 0}$.
\end{proof}

\begin{proof}[Proof of Theorem \ref{thm:Solution:TerminalWealthOnly:ValueFunction:lambda}]
The value function of this problem is given by
\begin{align*}
\mathcal{V}_{2}(v_{2}) = {} & \IE\left[U_{2}(V_{2}(T ; v_{2}))\right] = \IE\left[U_{2}\left(\left(v_{2} - F_{2}(0)\right) e^{\frac{\hat{b}}{\hat{b}-1} \left(r - \frac{1}{2} \frac{1}{\hat{b}-1} \|\gamma\|^{2}\right) T} \tilde{Z}(T)^{\frac{1}{\hat{b}-1}} + F\right)\right] \\
= {} & e^{- \beta  T}  \frac{1-\hat{b}}{\hat{b}}  \hat{a} \left(\frac{1}{1-\hat{b}}\right)^{\hat{b}} \left(v_{2} - F_{2}(0)\right)^{\hat{b}} e^{\frac{\hat{b}^{2}}{\hat{b}-1} \left(r - \frac{1}{2} \frac{1}{\hat{b}-1} \|\gamma\|^{2}\right) T} \IE\left[\tilde{Z}(T)^{\frac{\hat{b}}{\hat{b}-1}}\right] \\
= {} & e^{- \beta  T}  \frac{\left(1-\hat{b}\right)^{1-\hat{b}}}{\hat{b}}  \hat{a} \left(v_{2} - F_{2}(0)\right)^{\hat{b}} e^{\frac{\hat{b}^{2}}{\hat{b}-1} \left(r - \frac{1}{2} \frac{1}{\hat{b}-1} \|\gamma\|^{2}\right) T} e^{- \frac{\hat{b}}{\hat{b}-1} \left(r - \frac{1}{2} \frac{1}{\hat{b}-1} \|\gamma\|^{2}\right) T} \displaybreak \\
= {} & e^{\left[- \beta + \hat{b} \left(r - \frac{1}{2} \frac{1}{\hat{b}-1} \|\gamma\|^{2}\right)\right] T}  \frac{\left(1-\hat{b}\right)^{1-\hat{b}}}{\hat{b}}  \hat{a} \left(v_{2} - F_{2}(0)\right)^{\hat{b}}.
\end{align*}
This implies
\begin{align*}
\mathcal{V}_{2}^{\prime}(v_{2}) = {} & e^{\left[- \beta + \hat{b} \left(r - \frac{1}{2} \frac{1}{\hat{b}-1} \|\gamma\|^{2}\right)\right] T}  \frac{\left(1-\hat{b}\right)^{1-\hat{b}}}{\hat{b}}  \hat{a} \hat{b} \left(v_{2} - F_{2}(0)\right)^{\hat{b}-1} \\
= {} & e^{\left[- \beta + \hat{b} \left(r - \frac{1}{2} \frac{1}{\hat{b}-1} \|\gamma\|^{2}\right)\right] T}  \left(1-\hat{b}\right)^{1-\hat{b}}  \hat{a} \left(v_{2} - F_{2}(0)\right)^{\hat{b}-1} \\
\stackrel{\eqref{eq:Lagrange:TerminalWealthOnly}}{=} {} \lambda_{2} > 0.
\end{align*}
Due to the assumption ${v_{2} - F_{2}(0) > 0}$ in \eqref{eq:Condition:v2:WithoutProbabilityConstraint_&_eq:def:F2(t)}, it is straightforward that ${\mathcal{V}_{2}^{\prime\prime}(v_{2}) = \lambda_{2}^{\prime} < 0}$:
\begin{align*}
\mathcal{V}_{2}^{\prime\prime}(v_{2}) = - e^{\left[- \beta + \hat{b} \left(r - \frac{1}{2} \frac{1}{\hat{b}-1} \|\gamma\|^{2}\right)\right] T}  \left(1-\hat{b}\right)^{2-\hat{b}}  \hat{a} \left(v_{2} - F_{2}(0)\right)^{\hat{b}-2} < 0.
\end{align*}
\end{proof}

\subsection{Optimal merging of the individual solutions}
\label{app:ProofsOptimalMerging}

\begin{proof}[Proof of Theorem \ref{thm:ConnectionValueFunctions}]
~
\begin{enumerate}
\item $\mathcal{V}(v_{0}) \ge \sup_{v_{1} \ge F_{1}(0),\ v_{2} \ge F_{2}(0),\ v_{1} + v_{2} = v_{0}} \left\{\mathcal{V}_{1}(v_{1}) + \mathcal{V}_{2}(v_{2})\right\}$:

Let $(\pi_{1}(t;v_{1}), c_{1}(t;v_{1}))$ and $(\pi_{2}(t;v_{2}), c_{2}(t;v_{2}))$ denote the optimal controls to Problems \eqref{eq:OptimizationProblem:ConsumptionOnly} and \eqref{eq:OptimizationProblem:TerminalWealthOnly} with optimal wealth processes $V_{1}(t;v_{1})$ and $V_{2}(t;v_{2})$ to the initial wealths ${v_{1} \ge F_{1}(0)}$ and ${v_{2} \ge F_{2}(0)}$. Then, as the budget constraints for the optimal solutions to all three problems hold with equality,
\begin{align*}
\mathcal{V}_{1}(v_{1}) + \mathcal{V}_{2}(v_{2}) = {} & \IE\left[\int_{0}^{T} U_{1}(t,c_{1}(t;v_{1})) dt + U_{2}(V_{2}(T;v_{2}))\right] \\
\le {} & \sup_{(\pi,c) \in \Lambda} J(\pi,c;v_{0}) = \mathcal{V}(v_{0})
\end{align*}
for all $v_{1},v_{2}$ with ${v_{1} + v_{2} = v_{0}}$. Thus
\begin{align*}
\mathcal{V}(v_{0}) \ge \sup_{v_{1} \ge F_{1}(0),\ v_{2} \ge F_{2}(0),\ v_{1} + v_{2} = v_{0}} \left\{\mathcal{V}_{1}(v_{1}) + \mathcal{V}_{2}(v_{2})\right\}.
\end{align*}
\item $\mathcal{V}(v_{0}) \le \sup_{v_{1} \ge F_{1}(0),\ v_{2} \ge F_{2}(0),\ v_{1} + v_{2} = v_{0}} \left\{\mathcal{V}_{1}(v_{1}) + \mathcal{V}_{2}(v_{2})\right\}$:

Let $(\pi^{\star}, c^{\star})$ denote the optimal controls which maximize $\mathcal{V}(v_{0})$ with optimal wealth process $V^{\star}$ to the initial wealth $v_{0} > 0$. Define
\begin{align*}
v_{1} = \IE\left[\int_{0}^{T} \tilde{Z}(t) \left(c^{\star}(t) - y(t)\right) dt\right],\ v_{2} = \IE\left[\tilde{Z}(T) V^{\star}(T)\right].
\end{align*}
Then, ${v_{1} + v_{2} = v_{0}}$ and
\begin{align*}
\mathcal{V}(v_{0}) = {} & \IE\left[\int_{0}^{T} U_{1}(t,c^{\star}(t)) dt\right] + \IE\left[U_{2}(V^{\star}(T))\right] \le \mathcal{V}_{1}(v_{1}) + \mathcal{V}_{2}(v_{2}).
\end{align*}
Hence
\begin{align*}
\mathcal{V}(v_{0}) \le \sup_{v_{1} \ge F_{1}(0),\ v_{2} \ge F_{2}(0),\ v_{1} + v_{2} = v_{0}} \left\{\mathcal{V}_{1}(v_{1}) + \mathcal{V}_{2}(v_{2})\right\}.
\end{align*}
\end{enumerate}
\end{proof}

\begin{proof}[Proof of Lemma \ref{lemma:EqualityConditionValueFunctions}]
In accordance with Theorem \ref{thm:ConnectionValueFunctions} and by expressing $v_{2} = v_{0} - v_{1}$, the candidate for the optimal ${v_{1}^{\star}}$ is the one that satisfies the first order derivative condition on the budget
\begin{align*}
0 = \frac{\partial}{\partial v_{1}} \left(\mathcal{V}_{1}(v_{1}) + \mathcal{V}_{2}(v_{0} - v_{1})\right) = \mathcal{V}_{1}^{\prime}(v_{1}) - \mathcal{V}_{2}^{\prime}(v_{0} - v_{1})
\end{align*}
such that ${v_{1}^{\star} \ge F_{1}(0)}$, ${v_{2}^{\star} = v_{0} - v_{1}^{\star}}$ with ${v_{2}^{\star} \ge F_{2}(0)}$; thus ${F_{1}(0) \le v_{1}^{\star} \le v_{0} - F_{2}(0)}$. Theorems \ref{thm:Solution:ConsumptionOnly:ValueFunction:lambda} and \ref{thm:Solution:TerminalWealthOnly:ValueFunction:lambda} tell that $\mathcal{V}_{1}(v_{1})$ and $\mathcal{V}_{2}(v_{2})$ are strictly concave functions in $v_{1}$ respectively $v_{2}$. Therefore, it follows
\begin{align*}
0 = \frac{\partial^{2}}{\partial v_{1}^{2}} \left(\mathcal{V}_{1}(v_{1}) + \mathcal{V}_{2}(v_{0} - v_{1})\right) = \mathcal{V}_{1}^{\prime\prime}(v_{1}) + \mathcal{V}_{2}^{\prime\prime}(v_{0} - v_{1}) < 0.
\end{align*}
This implies that the candidates $v_{1}^{\star}$ and ${v_{2}^{\star} = v_{0} - v_{1}^{\star}}$ are the solution when the constraint $F_{1}(0) \le v_{1}^{\star} \le v_{0} - F_{2}(0)$ applies.
\end{proof}

\begin{proof}[Proof of Lemma \ref{lemma:EqualityConditionValueFunctions:Solution}]
In accordance with Theorems \ref{thm:Solution:ConsumptionOnly:ValueFunction:lambda} and \ref{thm:Solution:TerminalWealthOnly:ValueFunction:lambda} we have
\begin{align*}
\mathcal{V}_{1}^{\prime}(v_{1}) = {} & \lambda_{1}, \\
\mathcal{V}_{2}^{\prime}(v_{2}) = {} & \lambda_{2} = e^{-\left[\beta - \hat{b} \left(r - \frac{1}{2} \frac{1}{\hat{b}-1} \|\gamma\|^{2}\right)\right] T}  \left(1-\hat{b}\right)^{1-\hat{b}}  \hat{a} \left(v_{2} - F_{2}(0)\right)^{\hat{b}-1}.
\end{align*}
By equating $\mathcal{V}_{1}^{\prime}(v_{1})$ and $\mathcal{V}_{2}^{\prime}(v_{0} - v_{1})$ we obtain
\begin{align*}
\eqref{eq:OptimalMerging:ValueFunction} \text{ in Lemma \ref{lemma:EqualityConditionValueFunctions}} \ \Leftrightarrow\ & \mathcal{V}_{1}^{\prime}(v_{1}) = \mathcal{V}_{2}^{\prime}(v_{0} - v_{1}) \\
\Leftrightarrow\ & \lambda_{1} = \lambda_{2} = e^{-\left[\beta - \hat{b} \left(r - \frac{1}{2} \frac{1}{\hat{b}-1} \|\gamma\|^{2}\right)\right] T}  \left(1-\hat{b}\right)^{1-\hat{b}}  \hat{a} \left(v_{0} - v_{1} - F_{2}(0)\right)^{\hat{b}-1}.
\end{align*}
Inserting $\lambda_{1}$ in Equation \eqref{eq:ConsumptionOnly:lambda}, the optimal $v_{1}^{\star}$ is the solution to
\begin{align*}
v_{1} - \int_{0}^{T} \chi(t) \left(v_{0} - v_{1} - F_{2}(0)\right)^{\frac{\hat{b}-1}{b(t)-1}} dt = F_{1}(0),
\end{align*}
where the continuous function $\chi(t)$ is defined by
\begin{align*}
\chi(t) = (1-b(t)) \left(1-\hat{b}\right)^{\frac{1-\hat{b}}{b(t)-1}} \left(\frac{\hat{a}}{a(t)}\right)^{\frac{1}{b(t)-1}} \left(\frac{e^{\left[\beta - b(t) \left(r - \frac{1}{2} \frac{1}{b(t)-1} \|\gamma\|^{2}\right)\right] t}}{e^{\left[\beta - \hat{b} \left(r - \frac{1}{2} \frac{1}{\hat{b}-1} \|\gamma\|^{2}\right)\right] T}}\right)^{\frac{1}{b(t)-1}} > 0.
\end{align*}
It remains to verify ${F_{1}(0) \le v_{1}^{\star} \le v_{0} - F_{2}(0)}$ and uniqueness of $v_{1}^{\star}$. For this sake, define the function $f$ by
\begin{align*}
f: (-\infty,v_{0} - F_{2}(0)],\ f(x) = x - \int_{0}^{T} \chi(t) \left(v_{0} - x - F_{2}(0)\right)^{\frac{\hat{b}-1}{b(t)-1}} dt - F_{1}(0).
\end{align*}
$v_{1}^{\star}$ is the root of the function $f$, i.e. ${f(v_{1}^{\star}) = 0}$, if it holds ${v_{1}^{\star} \ge F_{1}(0)}$. $f$ is continuous in $x$, the exponent ${\frac{\hat{b}-1}{b(t)-1}}$ within the first integral is positive. Furthermore, due to ${v_{0} > F(0)}$ claimed in \eqref{eq:Condition:v0:minimalrequirement} and ${F(t) = F_{1}(t) + F_{2}(t)}$, we have for the limits
\begin{align*}
\lim_{x \searrow F_{1}(0)} f(x) = {} & - \int_{0}^{T} \chi(t) \left(v_{0} - F_{1}(0) - F_{2}(0)\right)^{\frac{\hat{b}-1}{b(t)-1}} dt = - \int_{0}^{T} \chi(t) \left(v_{0} - F(0)\right)^{\frac{\hat{b}-1}{b(t)-1}} dt < 0, \\
\lim_{x \nearrow v_{0} - F_{2}(0)} f(x) = {} & v_{0} - F_{2}(0) - \int_{0}^{T} \chi(t) \left(v_{0} - \left(v_{0} - F_{2}(0)\right) - F_{2}(0)\right)^{\frac{\hat{b}-1}{b(t)-1}} dt - F_{1}(0) \\
= {} & v_{0} - F_{2}(0) - F_{1}(0) = v_{0} - F(0) > 0.
\end{align*}
Note, ${F_{1}(0) \le v_{1} = v_{0} - v_{2} \le v_{0} - F_{2}(0)}$ for general $v_{1}$ and $v_{2}$. Additionally, $f$ is strictly monotone increasing in $x$ since
\begin{align*}
f^{\prime}(x) = 1 + \int_{0}^{T} \chi(t) \frac{\hat{b}-1}{b(t)-1} \left(v_{0} - x - F_{2}(0)\right)^{\frac{\hat{b}-b(t)}{b(t)-1}} dt > 0,\ \forall x \le v_{0} - F_{2}(0).
\end{align*}
We conclude that there exists a unique root ${x \in [F_{1}(0), v_{0} - F_{2}(0)]}$ such that ${f(x) = 0}$. Therefore, we conclude that the optimal $v_{1}^{*}$ and $v_{2}^{\star} = v_{0} - v_{1}^{*}$ exist and are unique. $v_{1}^{*}$ is the solution to Equation \eqref{eq:Separation:Optimalv1}. The optimal Lagrange multiplier $\lambda_{1}^{\star} = \lambda_{1}(v_{1}^{\star})$ is then given by
\begin{align*}
\lambda_{1}^{\star} = \left(1-\hat{b}\right)^{1-\hat{b}} \hat{a} e^{-\left[\beta - \hat{b} \left(r - \frac{1}{2} \frac{1}{\hat{b}-1} \|\gamma\|^{2}\right)\right] T} \left(v_{0} - v_{1}^{\star} - F_{2}(0)\right)^{\hat{b}-1}.
\end{align*}
\end{proof}

\begin{proof}[Proof of Theorem \ref{thm:Solution:Merging:OriginalProblem}]
Starting with ${V^{\star}(t; v_{0}) = V_{1}(t;v_{1}^{\star}) + V_{2}(t;v_{2}^{\star})}$ we compare the dynamics of both sides of the equation:
\begin{align} \label{eq:Merging:dV:dV1+dV2}
d V^{\star}(t; v_{0}) = d V_{1}(t;v_{1}^{\star}) + d V_{2}(t;v_{2}^{\star}).
\end{align}
Equation \eqref{eq:SDE:V:y} for $V^{\star}(t; v_{0})$, $V_{1}(t;v_{1}^{\star})$ and $V_{2}(t;v_{2}^{\star})$, with $y(t) \equiv 0$ for $V_{2}(t;v_{2}^{\star})$, provides
\begin{align*}
d V^{\star}(t; v_{0}) = {} & V^{\star}(t; v_{0})  \left[\left(r + \hat{\pi}^{\star}(t; v_{0})' \left(\mu - r \mathbf{1}\right)\right) dt + \hat{\pi}^{\star}(t; v_{0})'\sigma dW(t)\right] - c^{\star}(t; v_{0}) dt + y(t) dt, \\
d V_{1}(t;v_{1}^{\star}) = {} & V_{1}(t;v_{1}^{\star}) \left[\left(r + \hat{\pi}_{1}(t;v_{1}^{\star})' \left(\mu - r \mathbf{1}\right)\right) dt + \hat{\pi}_{1}(t;v_{1}^{\star})'\sigma dW(t)\right] - c_{1}(t;v_{1}^{\star}) dt + y(t) dt, \\
d V_{2}(t;v_{2}^{\star}) = {} & V_{2}(t;v_{2}^{\star})  \left[\left(r + \hat{\pi}_{2}(t;v_{2}^{\star})' \left(\mu - r \mathbf{1}\right)\right) dt + \hat{\pi}_{2}(t;v_{2}^{\star})'\sigma dW(t)\right].
\end{align*}
Comparing the diffusion terms in \eqref{eq:Merging:dV:dV1+dV2} gives
\begin{align*}
\hat{\pi}^{\star}(t; v_{0}) = \frac{\hat{\pi}_{1}(t;v_{1}^{\star}) V_{1}(t;v_{1}^{\star}) + \hat{\pi}_{2}(t;v_{2}^{\star}) V_{2}(t;v_{2}^{\star})}{V^{\star}(t; v_{0})}.
\end{align*}
Inserting this back and comparing the drift terms finally leads to
\begin{align*}
c^{\star}(t; v_{0}) = c_{1}(t;v_{1}^{\star}).
\end{align*}
Notice that the pair ${(\hat{\pi}^{\star},c^{\star})}$ is admissible, i.e. ${(\hat{\pi}^{\star},c^{\star}) \in \Lambda}$ because ${(\hat{\pi}_{1},c_{1}) \in \Lambda_{1}}$ and ${(\hat{\pi}_{2},0) \in \Lambda_{2}}$ which implies
\begin{align*}
V^{\star}(t; v_{0}) = \underbrace{V_{1}(t;v_{1}^{\star})}_{\ge - \int_{t}^{T} e^{- r (s-t)} y(s) ds} + \underbrace{V_{2}(t;v_{2}^{\star})}_{\ge 0} \ge - \int_{t}^{T} e^{- r (s-t)} y(s) ds,\ \IP-a.s.,\ \forall t \in [0,T].
\end{align*}
Using the solutions in Theorems \ref{thm:Solution:ConsumptionOnly} and \ref{thm:Solution:TerminalWealthOnly:WithoutProbabilityConstraint} we derive the following for the utility setup in \eqref{eq:utilitymodel:new}:
\begin{align*}
\hat{\pi}^{\star}(t ; v_{0}) = {} & \Sigma^{-1} (\mu - r \mathbf{1}) \frac{\frac{1}{1 - b(\tilde{t}_{1}^{\star})} \left(V_{1}(t ; v_{1}^{\star}) - F_{1}(t)\right) + \frac{1}{1-\hat{b}} \left(V_{2}(t ; v_{2}^{\star}) - F_{2}(t)\right)}{V^{\star}(t ; v_{0})}, \\
c^{\star}(t;v_{0}) = {} & c_{1}(t;v_{1}^{\star}) = g(t,t; v_{1}^{\star}) \tilde{Z}(t)^{\frac{1}{b(t)-1}} + \bar{c}(t) = (1-b(t)) \left(\lambda_{1}^{\star} \frac{e^{\beta  t}}{a(t)} \tilde{Z}(t)\right)^{\frac{1}{b(t)-1}} + \bar{c}(t), \\
V^{\star}(t ; v_{0}) = {} & V_{1}(t;v_{1}^{\star}) + V_{2}(t;v_{2}^{\star}) \\
= {} &  \int_{t}^{T} g(s,t; v_{1}^{\star}) \tilde{Z}(t)^{\frac{1}{b(s)-1}} ds + F_{1}(t) + \left(v_{2}^{\star} - F_{2}(0)\right) e^{\frac{\hat{b}}{\hat{b}-1} \left(r - \frac{1}{2} \frac{1}{\hat{b}-1} \|\gamma\|^{2}\right) t} \tilde{Z}(t)^{\frac{1}{\hat{b}-1}} + F_{2}(t) \\
= {} &  \int_{t}^{T} g(s,t; v_{1}^{\star}) \tilde{Z}(t)^{\frac{1}{b(s)-1}} ds + \left(v_{2}^{\star} - F_{2}(0)\right) e^{\frac{\hat{b}}{\hat{b}-1} \left(r - \frac{1}{2} \frac{1}{\hat{b}-1} \|\gamma\|^{2}\right) t} \tilde{Z}(t)^{\frac{1}{\hat{b}-1}} + F(t), \\
V^{\star}(T ; v_{0}) = {} & \left(v_{2}^{\star} - F_{2}(0)\right) e^{\frac{\hat{b}}{\hat{b}-1} \left(r - \frac{1}{2} \frac{1}{\hat{b}-1} \|\gamma\|^{2}\right) T} \tilde{Z}(T)^{\frac{1}{\hat{b}-1}} + F, \\
V_{1}(t ; v_{1}^{\star}) = {} & \int_{t}^{T} g(s,t; v_{1}^{\star}) \tilde{Z}(t)^{\frac{1}{b(s)-1}} ds + F_{1}(t), \\
V_{2}(t ; v_{2}^{\star}) = {} & \left(v_{2}^{\star} - F_{2}(0)\right) e^{\frac{\hat{b}}{\hat{b}-1} \left(r - \frac{1}{2} \frac{1}{\hat{b}-1} \|\gamma\|^{2}\right) t} \tilde{Z}(t)^{\frac{1}{\hat{b}-1}} + F_{2}(t),
\end{align*}
for all $t \in [0,T]$, with
\begin{align*}
g(s,t; v_{1}^{\star}) = {} & (1-b(s)) \left(\frac{e^{\beta  s - b(s) \left(r - \frac{1}{2} \frac{1}{b(s)-1} \|\gamma\|^{2}\right) (s-t)}}{a(s)}\right)^{\frac{1}{b(s)-1}} \left(\lambda_{1}^{\star}\right)^{\frac{1}{b(s)-1}} \\
= {} & (1-b(s)) \left(1-\hat{b}\right)^{\frac{1-\hat{b}}{b(s)-1}} \left(\frac{\hat{a}}{a(s)}\right)^{\frac{1}{b(s)-1}} \left(\frac{e^{\beta  s - b(s) \left(r - \frac{1}{2} \frac{1}{b(s)-1} \|\gamma\|^{2}\right) (s-t)}}{e^{\left[\beta - \hat{b} \left(r - \frac{1}{2} \frac{1}{\hat{b}-1} \|\gamma\|^{2}\right)\right] T}}\right)^{\frac{1}{b(s)-1}} \left(v_{0} - v_{1}^{\star} - F_{2}(0)\right)^{\frac{\hat{b}-1}{b(s)-1}} \\
\stackrel{\eqref{eq:definition:chi}}{=} {} & \chi(s) e^{\frac{b(s)}{b(s)-1} \left(r - \frac{1}{2} \frac{1}{b(s)-1} \|\gamma\|^{2}\right) t} \left(v_{0} - v_{1}^{\star} - F_{2}(0)\right)^{\frac{\hat{b}-1}{b(s)-1}}.
\end{align*}
Furthermore, $\tilde{t}_{1}^{\star} = \tilde{t}_{1}(v_{1}^{\star}) \in (t,T)$ solves \eqref{eq:ConsumptionOnly:tautilde}:
\begin{align*}
b(\tilde{t}_{1}^{\star}) = {} & 1 + \frac{\int_{t}^{T} g(s,t; v_{1}^{\star}) \tilde{Z}(t)^{\frac{1}{b(s)-1}} ds}{\int_{t}^{T} \frac{1}{b(s)-1} g(s,t; v_{1}^{\star}) \tilde{Z}(t)^{\frac{1}{b(s)-1}} ds}.
\end{align*}
\end{proof}

\begin{proof}[Proof of Remark \ref{remark:Ye:case}]
The formula for the optimal investment strategy is straightforward from Theorem \ref{thm:Solution:Merging:OriginalProblem} as ${b(\tilde{t}_{1}^{\star}) \equiv \hat{b}}$ and ${V^{\star}(t ; v_{0}) = V_{1}(t ; v_{1}^{\star}) + V_{2}(t ; v_{2}^{\star})}$ for any $t \in [0,T]$. The optimal $v_{1}^{\star}$ can be determined by Lemma \ref{lemma:EqualityConditionValueFunctions:Solution} as the solution to Equation \eqref{eq:Separation:Optimalv1}:
\begin{align*}
v_{1}^{\star} - \left(v_{0} - v_{1}^{\star} - F_{2}(0)\right) \int_{0}^{T} \chi(t) dt = F_{1}(0),
\end{align*}
where
\begin{align*}
\chi(t) = \left(\frac{\hat{a}}{a(t)}\right)^{\frac{1}{\hat{b}-1}} e^{- \frac{1}{\hat{b}-1} \left[\beta - \hat{b} \left(r - \frac{1}{2} \frac{1}{\hat{b}-1} \|\gamma\|^{2}\right)\right] (T-t)}.
\end{align*}
Therefore,
\begin{align*}
v_{1}^{\star} = \frac{\left(v_{0} - F_{2}(0)\right) \int_{0}^{T} \chi(t) dt + F_{1}(0)}{\int_{0}^{T} \chi(t) dt + 1}.
\end{align*}
is the optimal budget to the consumption problem, ${v_{2}^{\star} = v_{0} - v_{1}^{\star}}$ is the optimal budget to the terminal wealth problem. Furthermore, by Lemma \ref{lemma:EqualityConditionValueFunctions:Solution} one knows
\begin{align*}
\left(\lambda_{1}^{\star}\right)^{\frac{1}{\hat{b}-1}} = \frac{1}{1-\hat{b}} \left(\frac{\hat{a}}{e^{\left[\beta - \hat{b} \left(r - \frac{1}{2} \frac{1}{\hat{b}-1} \|\gamma\|^{2}\right)\right] T}}\right)^{\frac{1}{\hat{b}-1}} \left(v_{0} - v_{1}^{\star} - F_{2}(0)\right).
\end{align*}
This enables us to calculate $g(s,t; v_{1}^{\star})$ to be
\begin{align*}
g(s,t; v_{1}^{\star}) = {} & (1-\hat{b}) \left(\frac{e^{\beta  s - \hat{b} \left(r - \frac{1}{2} \frac{1}{\hat{b}-1} \|\gamma\|^{2}\right) (s-t)}}{a(s)}\right)^{\frac{1}{\hat{b}-1}} \left(\lambda_{1}^{\star}\right)^{\frac{1}{\hat{b}-1}} \\
= {} & \left(\frac{\hat{a}}{a(s)}\right)^{\frac{1}{\hat{b}-1}} e^{- \frac{1}{\hat{b}-1} \left[\beta  (T-s) + \hat{b} \left(r - \frac{1}{2} \frac{1}{\hat{b}-1} \|\gamma\|^{2}\right) (s-t-T)\right]} \left(v_{0} - v_{1}^{\star} - F_{2}(0)\right) \\
= {} & \left(\frac{\hat{a}}{a(s)}\right)^{\frac{1}{\hat{b}-1}} e^{- \frac{1}{\hat{b}-1} \left[\beta  (T-s) + \hat{b} \left(r - \frac{1}{2} \frac{1}{\hat{b}-1} \|\gamma\|^{2}\right) (s-t-T)\right]}  \left(v_{0} - \frac{\left(v_{0} - F_{2}(0)\right) \int_{0}^{T} \chi(t) dt + F_{1}(0)}{\int_{0}^{T} \chi(t) dt + 1} - F_{2}(0)\right) \\
= {} & \left(\frac{\hat{a}}{a(s)}\right)^{\frac{1}{\hat{b}-1}} e^{- \frac{1}{\hat{b}-1} \left[\beta  (T-s) + \hat{b} \left(r - \frac{1}{2} \frac{1}{\hat{b}-1} \|\gamma\|^{2}\right) (s-t-T)\right]} \left(\frac{v_{0} - F_{2}(0) - F_{1}(0)}{\int_{0}^{T} \chi(t) dt + 1}\right) \\
= {} & \left(\frac{\hat{a}}{a(s)}\right)^{\frac{1}{\hat{b}-1}} e^{- \frac{1}{\hat{b}-1} \left[\beta  (T-s) + \hat{b} \left(r - \frac{1}{2} \frac{1}{\hat{b}-1} \|\gamma\|^{2}\right) (s-t-T)\right]} \left(\frac{v_{0} - F(0)}{\int_{0}^{T} \chi(t) dt + 1}\right)
\end{align*}
with $F(0) = \int_{0}^{T} e^{- r s} \left(\bar{c}(s) - y(s)\right) ds + e^{- r T} F$ defined in \eqref{eq:def:F(t)}, and thus using Theorem \ref{thm:Solution:Merging:OriginalProblem}:
\begin{align*}
V_{1}(t ; v_{1}^{\star}) = {} & \int_{t}^{T} g(s,t; v_{1}^{\star}) \tilde{Z}(t)^{\frac{1}{\hat{b}-1}} ds + F_{1}(t) \\
= {} & \tilde{Z}(t)^{\frac{1}{\hat{b}-1}} \left(\frac{v_{0} - F(0)}{\int_{0}^{T} \chi(t) dt + 1}\right) \int_{t}^{T} \left(\frac{\hat{a}}{a(s)}\right)^{\frac{1}{\hat{b}-1}} e^{- \frac{1}{\hat{b}-1} \left[\beta  (T-s) + \hat{b} \left(r - \frac{1}{2} \frac{1}{\hat{b}-1} \|\gamma\|^{2}\right) (s-t-T)\right]} ds + F_{1}(t) \\
= {} & \tilde{Z}(t)^{\frac{1}{\hat{b}-1}} \left(v_{0} - F(0)\right) e^{\frac{\hat{b}}{\hat{b}-1} \left(r - \frac{1}{2} \frac{1}{\hat{b}-1} \|\gamma\|^{2}\right) t} \frac{\int_{t}^{T} \chi(s) ds}{\int_{0}^{T} \chi(t) dt + 1} + F_{1}(t).
\end{align*}
With, again from Theorem \ref{thm:Solution:Merging:OriginalProblem},
\begin{align*}
V_{2}(t ; v_{2}^{\star}) = {} & \left(v_{2}^{\star} - F_{2}(0)\right) e^{\frac{\hat{b}}{\hat{b}-1} \left(r - \frac{1}{2} \frac{1}{\hat{b}-1} \|\gamma\|^{2}\right) t} \tilde{Z}(t)^{\frac{1}{\hat{b}-1}} + F_{2}(t) \\
= {} & \left(v_{0} - v_{1}^{\star} - F_{2}(0)\right) e^{\frac{\hat{b}}{\hat{b}-1} \left(r - \frac{1}{2} \frac{1}{\hat{b}-1} \|\gamma\|^{2}\right) t} \tilde{Z}(t)^{\frac{1}{\hat{b}-1}} + F_{2}(t) \\
= {} & \left(\frac{v_{0} - F_{2}(0) - F_{1}(0)}{\int_{0}^{T} \chi(t) dt + 1}\right) e^{\frac{\hat{b}}{\hat{b}-1} \left(r - \frac{1}{2} \frac{1}{\hat{b}-1} \|\gamma\|^{2}\right) t} \tilde{Z}(t)^{\frac{1}{\hat{b}-1}} + F_{2}(t) \\
= {} & \tilde{Z}(t)^{\frac{1}{\hat{b}-1}} \left(v_{0} - F(0)\right) \frac{e^{\frac{\hat{b}}{\hat{b}-1} \left(r - \frac{1}{2} \frac{1}{\hat{b}-1} \|\gamma\|^{2}\right) t}}{\int_{0}^{T} \chi(t) dt + 1} + F_{2}(t)
\end{align*}
because ${F(t) = F_{1}(t) + F_{2}(t)}$ $\forall t \in [0,T]$, it follows
\begin{align*}
V^{\star}(t ; v_{0}) = {} & V_{1}(t ; v_{1}^{\star}) + V_{2}(t ; v_{2}^{\star}) \\
= {} & \tilde{Z}(t)^{\frac{1}{\hat{b}-1}} \left(v_{0} - F(0)\right) \frac{1}{\int_{0}^{T} \chi(t) dt + 1} \left\{e^{\frac{\hat{b}}{\hat{b}-1} \left(r - \frac{1}{2} \frac{1}{\hat{b}-1} \|\gamma\|^{2}\right) t} \int_{t}^{T} \chi(s) ds + e^{\frac{\hat{b}}{\hat{b}-1} \left(r - \frac{1}{2} \frac{1}{\hat{b}-1} \|\gamma\|^{2}\right) t}\right\} \\
& + F_{1}(t) + F_{2}(t) \displaybreak \\
= {} & \tilde{Z}(t)^{\frac{1}{\hat{b}-1}} \left(v_{0} - F(0)\right) e^{\frac{\hat{b}}{\hat{b}-1} \left(r - \frac{1}{2} \frac{1}{\hat{b}-1} \|\gamma\|^{2}\right) t} \frac{\int_{t}^{T} \chi(s) ds + 1}{\int_{0}^{T} \chi(t) dt + 1} + F(t).
\end{align*}
Finally, the optimal consumption rate process can then be determined from Theorem \ref{thm:Solution:Merging:OriginalProblem} as
\begin{align*}
c^{\star}(t;v_{0}) = {} & (1-\hat{b}) \left(\lambda_{1}^{\star} \frac{e^{\beta  t}}{a(t)} \tilde{Z}(t)\right)^{\frac{1}{\hat{b}-1}} + \bar{c}(t) \\
= {} & \tilde{Z}(t)^{\frac{1}{\hat{b}-1}} \left(\frac{e^{\beta  t}}{a(t)}\right)^{\frac{1}{\hat{b}-1}} \left(\frac{\hat{a}}{e^{\left[\beta - \hat{b} \left(r - \frac{1}{2} \frac{1}{\hat{b}-1} \|\gamma\|^{2}\right)\right] T}}\right)^{\frac{1}{\hat{b}-1}} \left(\frac{v_{0} - F(0)}{\int_{0}^{T} \chi(t) dt + 1}\right) + \bar{c}(t) \\
= {} & \tilde{Z}(t)^{\frac{1}{\hat{b}-1}} \left(v_{0} - F(0)\right) \left(\frac{\hat{a}}{a(t)}\right)^{\frac{1}{\hat{b}-1}} e^{- \frac{1}{\hat{b}-1} \beta  (T-t) + \frac{\hat{b}}{\hat{b}-1} \left(r - \frac{1}{2} \frac{1}{\hat{b}-1} \|\gamma\|^{2}\right) T} \left(\frac{1}{\int_{0}^{T} \chi(t) dt + 1}\right) + \bar{c}(t) \\
= {} & \tilde{Z}(t)^{\frac{1}{\hat{b}-1}} \left(v_{0} - F(0)\right) e^{\frac{\hat{b}}{\hat{b}-1} \left(r - \frac{1}{2} \frac{1}{\hat{b}-1} \|\gamma\|^{2}\right) 
t} \frac{\chi(t)}{\int_{0}^{T} \chi(t) dt + 1} + \bar{c}(t) \\
= {} & \tilde{Z}(t)^{\frac{1}{\hat{b}-1}} \left(v_{0} - F(0)\right) e^{\frac{\hat{b}}{\hat{b}-1} \left(r - \frac{1}{2} \frac{1}{\hat{b}-1} \|\gamma\|^{2}\right) 
t} \frac{\int_{t}^{T} \chi(s) ds + 1}{\int_{0}^{T} \chi(t) dt + 1} \frac{\chi(t)}{\int_{t}^{T} \chi(s) ds + 1} + \bar{c}(t) \\
= {} & \frac{\chi(t)}{\int_{t}^{T} \chi(s) ds + 1} \left(V^{\star}(t ; v_{0}) - F(t)\right) + \bar{c}(t).
\end{align*}
By defining
\begin{align*}
\zeta(t) = \frac{\chi(t)}{\int_{t}^{T} \chi(s) ds + 1} > 0
\end{align*}
we obtain
\begin{align*}
c^{\star}(t;v_{0}) = \zeta(t) \left(V^{\star}(t ; v_{0}) - F(t)\right) + \bar{c}(t).
\end{align*}
With the definition of $\zeta(t)$, the optimal wealth process finally can be written as
\begin{align*}
V^{\star}(t ; v_{0}) = {} & \tilde{Z}(t)^{\frac{1}{\hat{b}-1}} \left(v_{0} - F(0)\right) e^{\frac{\hat{b}}{\hat{b}-1} \left(r - \frac{1}{2} \frac{1}{\hat{b}-1} \|\gamma\|^{2}\right) t} \frac{\int_{t}^{T} \chi(s) ds + 1}{\int_{0}^{T} \chi(t) dt + 1} + F(t) \\
= {} & \frac{1}{\zeta(t)} \tilde{Z}(t)^{\frac{1}{\hat{b}-1}} \left(v_{0} - F(0)\right) e^{\frac{\hat{b}}{\hat{b}-1} \left(r - \frac{1}{2} \frac{1}{\hat{b}-1} \|\gamma\|^{2}\right) t} \frac{\chi(t)}{\int_{0}^{T} \chi(t) dt + 1} + F(t).
\end{align*}
\end{proof}

\end{document}